\documentclass[reqno,11pt]{amsart}
\usepackage[utf8]{inputenc}
\usepackage{enumitem}
\usepackage{graphicx}
\usepackage{amscd}
\usepackage{slashed}
\usepackage{amssymb}
\usepackage{esint}
\usepackage{faktor}
\usepackage{mathtools} 
\usepackage[usenames,dvipsnames]{pstricks}
\usepackage[mathscr]{eucal}
\numberwithin{equation}{section}
\allowdisplaybreaks[1]

\title[Linear Dynamics of Wave Functions]{The Linear Dynamics of Wave Functions \\
in Causal Fermion Systems}

\author[F.\ Finster]{Felix Finster}
\address{Fakult\"at f\"ur Mathematik \\ Universit\"at Regensburg \\ D-93040 Regensburg \\ Germany}
\email{finster@ur.de}

\author[N.\ Kamran]{Niky Kamran}
\address{Department of Mathematics and Statistics \\ McGill University \\ Montr{\'e}al \\ Canada}
\email{nkamran@math.mcgill.ca}

\author[M.\ Oppio]{Marco Oppio  \\ \\ January 2021}
\address{Fakult\"at f\"ur Mathematik \\ Universit\"at Regensburg \\ D-93040 Regensburg \\ Germany}
\email{marco.oppio@ur.de}

\newtheorem{Def}{Definition}[section]
\newtheorem{Thm}[Def]{Theorem}
\newtheorem{Prp}[Def]{Proposition}
\newtheorem{Lemma}[Def]{Lemma}
\newtheorem{Remark}[Def]{Remark}
\newtheorem{Corollary}[Def]{Corollary}

\newcommand{\Thanks}{\vspace*{.5em} \noindent \thanks}
\newcommand{\beq}{\begin{equation}}
\newcommand{\eeq}{\end{equation}}
\newcommand{\Proof}{\begin{proof}}
\newcommand{\QED}{\end{proof} \noindent}
\newcommand{\QEDrem}{\ \hfill $\Diamond$}

\newcommand{\la}{\langle}
\newcommand{\ra}{\rangle}
\newcommand{\lla}{\langle\!\langle}
\newcommand{\rra}{\rangle\!\rangle}
\newcommand{\llla}{\langle\!\langle\!\langle}
\newcommand{\rrra}{\rangle\!\rangle\!\rangle}
\newcommand{\bra}{\mathopen{<}}
\newcommand{\ket}{\mathclose{>}}
\newcommand{\Sl}{\mathopen{\prec}}
\newcommand{\Sr}{\mathclose{\succ}}

\newcommand{\C}{\mathbb{C}}
\newcommand{\R}{\mathbb{R}}
\newcommand{\1}{\mbox{\rm 1 \hspace{-1.05 em} 1}}

\newcommand{\N}{\mathbb{N}}
\newcommand{\Pdd}{\mbox{$\partial$ \hspace{-1.2 em} $/$}}
\newcommand{\slsh}{\mbox{ \hspace{-1.13 em} $/$}}

\renewcommand{\H}{\mathscr{H}}

\newcommand{\U}{{\rm{U}}}

\newcommand{\hT}{\hat{\mathscr{T}}}

\newcommand{\bep}{\begin{pmatrix}}
\newcommand{\enp}{\end{pmatrix}}

\renewcommand{\O}{\mathscr{O}}

\newcommand{\F}{{\mathscr{F}}}
\newcommand{\reg}{{\text{\rm{reg}}}}
\newcommand{\sing}{{\text{\rm{sing}}}}
\newcommand{\dyn}{{\text{\rm{dyn}}}}
\newcommand{\loc}{\text{\rm{loc}}}

\newcommand{\D}{{\mathscr{D}}}
\newcommand{\K}{{\mathscr{K}}}

\renewcommand{\O}{{\mathscr{O}}}
\renewcommand{\L}{{\mathcal{L}}}
\newcommand{\Sact}{{\mathcal{S}}}
\newcommand{\s}{{\mathfrak{s}}}

\newcommand{\Lin}{\text{\rm{L}}}
\newcommand{\T}{{\mathscr{T}}}

\newcommand{\fermi}{{\mathrm{f}}}

\newcommand{\scrM}{\myscr M}

\newcommand{\J}{\mathfrak{J}}

\newcommand{\Jlin}{\mathfrak{J}^\text{\rm{\tiny{lin}}}}
\newcommand{\Jtest}{\mathfrak{J}^\text{\rm{\tiny{test}}}}
\newcommand{\gen}{\text{\rm{\tiny{gen}}}}

\newcommand{\Gtest}{\Gamma^\text{\rm{\tiny{test}}}}
\newcommand{\Gdiff}{\Gamma^\text{\rm{\tiny{diff}}}}

\newcommand{\Glin}{\Gamma^\text{\rm{\tiny{lin}}}}
\newcommand{\Gfermi}{\Gamma^\text{\rm{\tiny{f}}}}

\newcommand{\GComm}{\Gamma^\mathscr{C}}
\newcommand{\Ctest}{C^\text{\rm{\tiny{test}}}}
\newcommand{\vary}{\text{\rm{\tiny{vary}}}}
\newcommand{\JComm}{\mathfrak{J}^{\mathscr{C}}}
\newcommand{\Jdiff}{\mathfrak{J}^\text{\rm{\tiny{diff}}}}

\newcommand{\scrU}{{\mathscr{U}}}
\newcommand{\scrW}{\mathscr{W}}
\newcommand{\scrA}{{\mathscr{A}}}
\newcommand{\scrB}{{\mathscr{B}}}
\renewcommand{\div}{{\rm{div}}\,}
\DeclareMathOperator{\norm}{| \hspace*{-0.1em}| \hspace*{-0.1em}|}
\newcommand{\Comm}{{\mathscr{C}}}

\newcommand{\itemD}{\item[{\raisebox{0.125em}{\tiny $\blacktriangleright$}}]}
\newcommand{\tmax}{{t_{\max}}}

\DeclareFontFamily{OT1}{rsfso}{}
\DeclareFontShape{OT1}{rsfso}{m}{n}{ <-7> rsfso5 <7-10> rsfso7 <10-> rsfso10}{}
\DeclareMathAlphabet{\myscr}{OT1}{rsfso}{m}{n}

\DeclareMathOperator{\re}{Re}

\DeclareMathOperator{\Tr}{Tr}
\DeclareMathOperator{\tr}{tr}

\DeclareMathOperator{\supp}{supp}

\newcommand{\Sig}{\mathscr{S}}
\renewcommand{\u}{\mathfrak{u}}
\renewcommand{\v}{\mathfrak{v}}

\newcommand{\bu}{{\mathbf{u}}}
\newcommand{\bv}{\mathbf{v}}

\newcommand{\bitem}{\begin{itemize}[leftmargin=2em]}
\newcommand{\eitem}{\end{itemize}}

\renewcommand{\sc}{\text{\rm{sc}}}

\begin{document}

\maketitle

\begin{abstract}
The dynamics of spinorial wave functions in a causal fermion system is studied.
A so-called dynamical wave equation is derived. Its solutions form a Hilbert space,
whose scalar product is represented by a conserved surface layer integral.
We prove under general assumptions that the initial value problem for the dynamical
wave equation admits a unique global solution. Causal Green's operators are constructed
and analyzed.
Our findings are illustrated in the example of the regularized Minkowski vacuum.
\end{abstract}

\tableofcontents

\section{Introduction}
The theory of {\em{causal fermion systems}} is a recent approach to fundamental physics
(see the basics in Section~\ref{secprelim}, the reviews~\cite{ nrstg, review, dice2014}, the textbook~\cite{cfs}
or the website~\cite{cfsweblink}).
In this approach, spacetime and all objects therein are described by a measure~$\rho$
on a set~$\F$ of linear operators of a Hilbert space~$(\H, \la .|. \ra_\H)$. 
The physical equations are formulated via the so-called {\em{causal action principle}},
a nonlinear variational principle where an action~$\Sact$ is minimized under variations of the measure~$\rho$.

It is a basic feature of a causal fermion system that it incorporates an
ensemble of spinorial wave functions in spacetime.
Indeed, defining {\em{spacetime}}~$M$ as the support of the measure~$\rho$,
\[ M := \supp \rho \subset \F \:, \]
to every spacetime point represented by an operator~$x \in M$ one associates the {\em{spin space}} $S_x=x(\H)$
defined as the image of this operator (for more details see the preliminaries in Section~\ref{seccfs}).
Then every vector~$u \in \H$ gives rise to a
wave function~$\psi^u(x)$ defined by
\[ \psi^u(x) = \pi_x u \in S_x \:, \]
where~$\pi_x : \H \rightarrow S_x$ is the orthogonal projection on the spin space.
This wave function is referred to as the {\em{physical wave function}} of~$u$.
The ensemble of all physical wave function plays a central role because, apart from
describing quantum mechanical matter, it also encodes the causal and geometric structures of
spacetime. The causal action principle can be understood as a variational principle
which aims at bringing all physical wave functions collectively into an ``optimal'' configuration. This
minimizing configuration then encodes all spacetime structures.

The goal of the present paper is to unveil the dynamics of the physical wave functions
of a causal fermion system. One obvious difficulty is that, due to the nonlinearity of the causal action 
principle, all physical wave functions interact with each other.
This can be understood in similar terms as the back reaction of a quantum field in curved
spacetime on the metric.
Even more, it is a specific feature of causal fermion systems that all the spacetime structures are
encoded in and derived from the physical wave functions. Therefore, modifying a physical wave function
also changes all the other structures (like the integration measure, the form of the scalar product, etc.).
As a consequence, it is a difficult task to disentangle the variation of the physical wave function
from the resulting changes of all the other structures,
so as to obtain a separate dynamics of a single physical wave function.
The idea for making this concept precise is to perturb the system linearly by varying only one
wave function while leaving all the other wave functions and the geometry of spacetime unchanged.
This single wave function should then obey a linear wave equation, in the same way as 
the wave function of a ``test particle'' in relativistic quantum mechanics satisfies the Dirac equation.
Generally speaking, our goal is to derive a linear wave equation which generalizes the Dirac
equation to spacetimes which could have a nontrivial,
not necessarily smooth microscopic structure. Clearly, in such generality this wave equation
cannot be formulated as a PDE. Instead, our linear wave equation must be derived from the the
dynamics as described by the causal action principle.

Linearizations have already been used in the context of causal fermion systems
in the derivation of the so-called {\em{linearized field equations}}
(see~\cite{jet} and the analysis in~\cite{linhyp}). However, this procedure does not
immediately apply to the spinorial wave functions.
The basic difficulty is that it is not obvious what a ``linear variation of a wave function'' should be.
This difficulty can already be understood in the analogy to a Dirac wave function coupled
to a classical electromagnetic field. A first order perturbation of this Dirac-Maxwell system
consists of a linear perturbation of the Dirac wave function complemented by a
linear perturbation of the Maxwell field generated by the Dirac wave function.
Here we need to take into account the Maxwell field generated by the linear perturbation of
the Dirac wave function. Varying only the Dirac wave function in general does not preserve the
coupled Dirac-Maxwell equations. Similarly, when perturbing the causal fermion system,
a linear variation of a physical wave function may give rise to
additional linear perturbations of the system.
These additional perturbations may affect the microscopic form of the
conserved quantities of the causal fermion system. In this interacting situation, the problem arises how to
distinguish the perturbation of the physical wave function from all the other linear perturbations of the system.
This is the reason why the analysis in the present paper goes beyond the methods in~\cite{linhyp}
and proceeds in a different direction.

More specifically, we proceed as follows. The first step is to represent the Hilbert space
scalar product as a surface layer integral. Recall that a {\em{surface layer integral}} generalizes the
concept of a surface integral to causal fermion systems. It is a double integral of the form
\[ 
\int_\Omega \bigg( \int_{M \setminus \Omega} (\cdots)\: \L_\kappa(x,y)\: d\rho(y) \bigg)\, d\rho(x) \:, \]
where $(\cdots)$ stands for a differential operator acting on the Lagrangian~$\L_\kappa$
(for details see Section~\ref{secosi} in the preliminaries). Here the set~$\Omega$ should be thought of
as the past of a spacelike hypersurface~$\partial \Omega$ (for details see Section~\ref{secpastset}).
A surface layer integral is {\em{conserved}} if the above double integral does not depend on
the choice of~$\Omega$. We make use of the conservation law obtained from a Noether-like theorem
for symmetry transformations given by unitary transformations. This conservation law was
first obtained in~\cite{noether}, where it was also shown that it generalizes the conservation of the Dirac
current to the setting of causal fermion systems. We here recover this conservation law
as a special case of a more general conservation law derived in~\cite{osi} when evaluated for
the infinitesimal generators of the symmetry, the so-called {\em{commutator jets}}.
We always assume that the resulting conserved sesquilinear form
gives back the Hilbert space scalar product. We thus obtain for any~$u, v \in \H$ the equation
\beq \label{OSIintro}
\la u|v \ra^\Omega_\rho = -2i \,\bigg( \int_{\Omega} \!d\rho(x) \int_{M \setminus \Omega} \!\!\!\!\!\!\!d\rho(y) 
- \int_{M \setminus \Omega} \!\!\!\!\!\!\!d\rho(x) \int_{\Omega} \!d\rho(y) \bigg)\:
\Sl \psi^u(x) \:|\: Q(x,y)\, \psi^v(y) \Sr_x \:,
\eeq
where~$\psi^u(x)$ is the physical wave function of~$u$, and the kernel~$Q(x,y)$ describes first variations of the Lagrangian (for details see~\eqref{delLdef}
and Proposition~\ref{prpcomm}). This kernel also appears in the EL equations for the physical wave functions,
which can be written as (see~\eqref{ELQ} in the preliminaries)
\beq \label{linintro}
\int_M Q(x,y)\, \psi^u(y) \:d\rho(y) = \mathfrak{r}\, \psi^u(x) \qquad \text{for all~$x \in M$}\:,
\eeq
where~$\mathfrak{r}$ is a real-valued Lagrange parameter.
The basic idea is to consider~\eqref{linintro} as the linear equation describing the dynamics of the
spinorial wave functions (which should generalize the Dirac equation), and~\eqref{OSIintro} as the
corresponding conserved scalar product (generalizing the scalar product on Dirac wave functions
as used in quantum mechanics). However, this idea cannot be implemented directly, because
the equation~\eqref{linintro} is not satisfied if~$\psi^u$ is replaced by a perturbation of
a physical wave function (in the analogy to the Dirac equation, \eqref{linintro} holds in the vacuum
only for the solutions of negative energy, but not for the solutions of positive energy).

In order to extend~\eqref{OSIintro} to more general wave functions, we consider
finite perturbations of the system which preserve the surface layer integral in~\eqref{OSIintro}.
Denoting the perturbed objects with a tilde, we thus obtain
\beq \label{OSItildeintro}
\begin{split}
\la u|v \ra^\Omega_\rho = -2i \,\bigg( \int_{\tilde{\Omega}} d\tilde{\rho}(x)
\int_{\tilde{M} \setminus \tilde{\Omega}} d\tilde{\rho}(y) 
&- \int_{\tilde{M} \setminus \tilde{\Omega}} d\tilde{\rho}(x) \int_{\tilde{\Omega}} d\tilde{\rho}(y) \bigg) \\
&\times \Sl \tilde{\psi}^u(x) \:|\: \tilde{Q}^\reg(x,y)\, \tilde{\psi}^v(y) \Sr_x \:,
\end{split}
\eeq
One should keep in mind that both the physical wave functions and the form of the inner product change
under this perturbation.
Using spectral methods in Hilbert spaces, we transform the perturbed inner product back to the
unperturbed form. This also transforms the perturbed wave functions, giving rise to new wave functions
which again satisfy the conservation law~\eqref{linintro}. Extending~$\H$ by these new wave functions, we obtain
the {\em{extended Hilbert space}}~$(\H^{\fermi, \Omega}_\rho, \la .|. \ra^\Omega_\rho)$
(again in the analogy to the Dirac equation, this Hilbert space can be understood as the whole solution
space, including the positive-energy solutions).

Having extended the conservation law, we next consider the question whether also the
dynamical equation~\eqref{linintro} can be extended to~$\H^{\fermi, \Omega}_\rho$.
We show that this can indeed be done, if we restrict attention to a special class of linear
perturbations of the system, which must satisfy suitable compatibility conditions.
The detailed analysis leads us to modify the kernel~$Q(x,y)$ to the ``dynamical'' kernel~$Q^\dyn(x,y)$.
This modification does not change the surface layer integral in~\eqref{OSItildeintro}.
We thus end up with the {\em{dynamical wave equation}}
\[ \int_M Q^\dyn(x,y)\: \psi(y)\: d\rho(y) = 0 \]
which comes with a corresponding conserved scalar product
\[ \la \psi | \phi \ra^\Omega_\rho = -2i \,\bigg( \int_{\Omega} \!d\rho(x) \int_{M \setminus \Omega} \!\!\!\!\!\!\!d\rho(y) 
- \int_{M \setminus \Omega} \!\!\!\!\!\!\!d\rho(x) \int_{\Omega} \!d\rho(y) \bigg)\:
\Sl \psi(x) \:|\: Q^\dyn(x,y)\, \phi(y) \Sr_x \:. \]

We also analyze the Cauchy problem for the inhomogeneous dynamical wave equation.
Using energy methods similar to those employed in~\cite{linhyp}, we prove under general assumptions that
this Cauchy problem is well-posed and admits global solutions which propagate with finite speed.
Based on these results, we introduce the advanced and retarded Green's operators
and the causal fundamental solution. The properties of the solution space are expressed
by an exact sequence involving the operators~$Q^\dyn$ and the causal fundamental solution.
All our constructions are illustrated in the example of the regularized Dirac sea vacuum in Minkowski space.

The paper is organized as follows. Section~\ref{secprelim} provides the necessary preliminaries
on causal fermion systems and the causal action principle.
In Section~\ref{seccomminner} we introduce commutator jets and use them for representing
the Hilbert space scalar product as a surface layer integral.
In Section~\ref{secHextend} the extended Hilbert space is constructed.
Section~\ref{secdynamics} is devoted to the derivation of the dynamical wave equation.
In Section~\ref{secQdyn}, it is shown that, under general assumptions,
the Cauchy problem for the dynamical wave
equation admits unique global solutions. Moreover, Green's operators are introduced and analyzed.
Two appendices clarify the role of commutator jets and inner solutions in various surface layer integrals.

\section{Preliminaries} \label{secprelim}
In this section we provide the necessary mathematical background.
More details can be found in the textbook~\cite{cfs} or in the research articles~\cite{jet, fockbosonic}.

\subsection{Causal Fermion Systems and the Causal Action Principle} \label{seccfs}
We now recall the basic setup.
\begin{Def} \label{defcfs} (causal fermion systems) {\em{ 
Given a separable complex Hilbert space~$\H$ with scalar product~$\la .|. \ra_\H$
and a parameter~$n \in \N$ (the {\em{``spin dimension''}}), we let~$\F \subset \Lin(\H)$ be the set of all
selfadjoint operators on~$\H$ of finite rank, which (counting multiplicities) have
at most~$n$ positive and at most~$n$ negative eigenvalues. On~$\F$ we are given
a positive measure~$\rho$ (defined on a $\sigma$-algebra of subsets of~$\F$), the so-called
{\em{universal measure}}. We refer to~$(\H, \F, \rho)$ as a {\em{causal fermion system}}.
}}
\end{Def} \noindent
A causal fermion system describes a spacetime together
with all structures and objects therein.
In order to single out the physically admissible
causal fermion systems, one must formulate physical equations. To this end, we impose that
the universal measure should be a minimizer of the causal action principle,
which we now introduce. For any~$x, y \in \F$, the product~$x y$ is an operator of rank at most~$2n$. 
However, in general it is no longer a selfadjoint operator because~$(xy)^* = yx$,
and this is different from~$xy$ unless~$x$ and~$y$ commute.
As a consequence, the eigenvalues of the operator~$xy$ are in general complex.
We denote these eigenvalues counting algebraic multiplicities
by~$\lambda^{xy}_1, \ldots, \lambda^{xy}_{2n} \in \C$
(more specifically,
denoting the rank of~$xy$ by~$k \leq 2n$, we choose~$\lambda^{xy}_1, \ldots, \lambda^{xy}_{k}$ as all
the non-zero eigenvalues and set~$\lambda^{xy}_{k+1}, \ldots, \lambda^{xy}_{2n}=0$).
We introduce the Lagrangian and the causal action by
\begin{align}
\text{\em{Lagrangian:}} && \L(x,y) &= \frac{1}{4n} \sum_{i,j=1}^{2n} \Big( \big|\lambda^{xy}_i \big|
- \big|\lambda^{xy}_j \big| \Big)^2 \label{Lagrange} \\
\text{\em{causal action:}} && \Sact(\rho) &= \iint_{\F \times \F} \L(x,y)\: d\rho(x)\, d\rho(y) \:. \label{Sdef}
\end{align}
The {\em{causal action principle}} is to minimize~$\Sact$ by varying the measure~$\rho$
under the following constraints:
\begin{align}
\text{\em{volume constraint:}} && \rho(\F) = \text{const} \quad\;\; & \label{volconstraint} \\
\text{\em{trace constraint:}} && \int_\F \tr(x)\: d\rho(x) = \text{const}& \label{trconstraint} \\
\text{\em{boundedness constraint:}} && \iint_{\F \times \F} 
|xy|^2
\: d\rho(x)\, d\rho(y) &\leq C \:, \label{Tdef}
\end{align}
where~$C$ is a given parameter, $\tr$ denotes the trace of a linear operator on~$\H$, and
the absolute value of~$xy$ is the so-called spectral weight,
\beq \label{sw}
|xy| := \sum_{j=1}^{2n} \big|\lambda^{xy}_j \big| \:.
\eeq
This variational principle is mathematically well-posed if~$\H$ is finite-dimensional,
in which case also the constants in~\eqref{volconstraint}--\eqref{Tdef} are finite.
For the existence theory and the analysis of general properties of minimizing measures
we refer to~\cite{discrete, continuum, lagrange}.
In the existence theory one varies in the class of regular Borel measures
(with respect to the topology on~$\Lin(\H)$ induced by the operator norm),
and the minimizing measure is again in this class. With this in mind, here we always assume that
\[ 
\text{$\rho$ is a regular Borel measure}\:. \]

We note for clarity that the above variational principle also makes mathematical sense
if the constants in~\eqref{volconstraint}--\eqref{Tdef} are
infinite. In this case, the constraints can still be given a mathematical meaning by
restricting attention to variations for which the differences of the left sides in~\eqref{volconstraint}--\eqref{Tdef}
before and after the variation are well-defined and finite.
This makes it possible to implement the constraints by demanding
that these finite numbers be zero and non-positive, respectively.
This procedure is explained in detail in the context of causal variational principles in~\cite[Section~2.1]{jet};
see also corresponding existence results in infinite volume in~\cite{noncompact}.

\subsection{Spacetime and Physical Wave Functions}
Let~$\rho$ be a {\em{minimizing}} measure. {\em{Spacetime}}
is defined as the support of this measure,
\[ 
M := \supp \rho \:. \]
Thus the spacetime points are selfadjoint linear operators on~$\H$.
On~$M$ we consider the topology induced by~$\F$ (generated by the $\sup$-norm
on~$\Lin(\H)$). Moreover, the universal measure~$\rho|_M$ restricted to~$M$ can be regarded as a volume
measure on spacetime. This makes spacetime into a {\em{topological measure space}}.

The operators in~$M$ contain a lot of information which, if interpreted correctly,
gives rise to spacetime structures like causal and metric structures, spinors
and interacting fields (for details see~\cite[Chapter~1]{cfs}).
He we restrict attention to those structures needed in what follows.
We begin with a basic notion of causality:

\begin{Def} (causal structure) \label{def2}
{\em{ For any~$x, y \in \F$, the product~$x y$ is an operator
of rank at most~$2n$. We denote its non-trivial eigenvalues (counting algebraic multiplicities)
by~$\lambda^{xy}_1, \ldots, \lambda^{xy}_{2n}$.
The points~$x$ and~$y$ are
called {\em{spacelike}} separated if all the~$\lambda^{xy}_j$ have the same absolute value.
They are said to be {\em{timelike}} separated if the~$\lambda^{xy}_j$ are all real and do not all 
have the same absolute value.
In all other cases (i.e.\ if the~$\lambda^{xy}_j$ are not all real and do not all 
have the same absolute value),
the points~$x$ and~$y$ are said to be {\em{lightlike}} separated. }}
\end{Def} \noindent
Restricting the causal structure of~$\F$ to~$M$, we get causal relations in spacetime.

Next, for every~$x \in \F$ we define the {\em{spin space}}~$S_x$ by~$S_x = x(\H)$;
it is a subspace of~$\H$ of dimension at most~$2n$.
It is endowed with the {\em{spin inner product}} $\Sl .|. \Sr_x$ defined by
\[ 
\Sl u | v \Sr_x = -\la u | x v \ra_\H \qquad \text{(for all $u,v \in S_x$)}\:. \]
A {\em{wave function}}~$\psi$ is defined as a function
which to every~$x \in M$ associates a vector of the corresponding spin space,
\[ 
\psi \::\: M \rightarrow \H \qquad \text{with} \qquad \psi(x) \in S_x \quad \text{for all~$x \in M$}\:. \]
In order to introduce a notion of continuity of a wave function,
we need to compare the wave function at different spacetime points.
Noting that the natural norm on the spin space~$(S_x, \Sl .|. \Sr_x)$ is given by
\beq \label{spinnorm}
\big| \psi(x) \big|_x^2 := \big\la \psi(x) \,\big|\, |x|\, \psi(x) \big\ra_\H = \Big\| \sqrt{|x|} \,\psi(x) \Big\|_\H^2
\eeq
(where~$|x|$ is the absolute value of the symmetric operator~$x$ on~$\H$, and~$\sqrt{|x|}$
is the square root thereof), we say that
the wave function~$\psi$ is {\em{continuous}} at~$x$ if
for every~$\varepsilon>0$ there is~$\delta>0$ such that
\beq \label{wavecontinuous}
\big\| \sqrt{|y|} \,\psi(y) -  \sqrt{|x|}\, \psi(x) \big\|_\H < \varepsilon
\qquad \text{for all~$y \in M$ with~$\|y-x\| \leq \delta$} \:.
\eeq
Likewise, $\psi$ is said to be continuous on~$M$ if it is continuous at every~$x \in M$.
We denote the set of continuous wave functions by~$C^0(M, SM)$.

It is an important observation that every vector~$u \in \H$ of the Hilbert space gives rise to a unique
wave function. To obtain this wave function, denoted by~$\psi^u$, we simply project the vector~$u$
to the corresponding spin spaces,
\[ 
\psi^u \::\: M \rightarrow \H\:,\qquad \psi^u(x) = \pi_x u \in S_x \:. \]
We refer to~$\psi^u$ as the {\em{physical wave function}} of~$u \in \H$.
A direct computation shows that the physical wave functions are continuous
(in the sense~\eqref{wavecontinuous}). Associating to every vector~$u \in \H$
the corresponding physical wave function gives rise to the {\em{wave evaluation operator}}
\[ 
\Psi \::\: \H \rightarrow C^0(M, SM)\:, \qquad u \mapsto \psi^u \:. \]
Every~$x \in M$ can be written as (for the derivation see~\cite[Lemma~1.1.3]{cfs})
\beq
x = - \Psi(x)^* \,\Psi(x) \label{Fid} \::\: \H \rightarrow \H \:.
\eeq
In words, every spacetime point operator is the local correlation operator of the wave evaluation operator
at this point. This formula is very useful when varying the system, as will be explained in Section~\ref{secELCFS} below.

\subsection{Connection to the Setting of Causal Variational Principles} \label{seccfscvp}
For the analysis of the causal action principle it is most convenient to get into the
simpler setting of causal variational principles. In this setting, $\F$ is a (possibly non-compact)
smooth manifold of dimension~$m \geq 1$ and~$\rho$ a positive Borel measure on~$\F$
(the {\em{universal measure}}).
Moreover, we are given a non-negative function~$\L : \F \times \F \rightarrow \R^+_0$
(the {\em{Lagrangian}}) with the following properties:
\bitem
\item[(i)] $\L$ is symmetric: $\L(x,y) = \L(y,x)$ for all~$x,y \in \F$.\label{Cond1}
\item[(ii)] $\L$ is lower semi-continuous, i.e.\ for all \label{Cond2}
sequences~$x_n \rightarrow x$ and~$y_{n'} \rightarrow y$,
\[ \L(x,y) \leq \liminf_{n,n' \rightarrow \infty} \L(x_n, y_{n'})\:. \]
\eitem
The {\em{causal variational principle}} is to minimize the action
\beq \label{Sact} 
\Sact (\rho) = \int_\F d\rho(x) \int_\F d\rho(y)\: \L(x,y) 
\eeq
under variations of the measure~$\rho$, keeping the total volume~$\rho(\F)$ fixed
({\em{volume constraint}}).
If the total volume~$\rho(\F)$ is finite, one minimizes~\eqref{Sact}
over all regular Borel measures with the same total volume.
If the total volume~$\rho(\F)$ is infinite, however, it is not obvious how to implement the volume constraint,
making it necessary to proceed as follows.
We need the following additional assumptions:
\bitem
\item[(iii)] The measure~$\rho$ is {\em{locally finite}}
(meaning that any~$x \in \F$ has an open neighborhood~$U$ with~$\rho(U)< \infty$).\label{Cond3}
\item[(iv)] The function~$\L(x,.)$ is $\rho$-integrable for all~$x \in \F$, giving
a lower semi-continuous and bounded function on~$\F$. \label{Cond4}
\eitem
Given a regular Borel measure~$\rho$ on~$\F$, we then vary over all
regular Borel measures~$\tilde{\rho}$ with
\[ 
\big| \tilde{\rho} - \rho \big|(\F) < \infty \qquad \text{and} \qquad
\big( \tilde{\rho} - \rho \big) (\F) = 0 \]
(where~$|.|$ denotes the total variation of a measure).
These variations of the causal action are well-defined.
The existence theory for minimizers is developed in~\cite{noncompact}.

There are several ways to get from the causal action principle to causal variational principles,
as we now explain in detail. If the Hilbert space~$\H$ is {\em{finite-dimensional}} and the total volume~$\rho(\F)$
is finite, one can proceed as follows:
As a consequence of the trace constraint~\eqref{trconstraint}, for any minimizing measure~$\rho$
the local trace is constant in spacetime, i.e.\
there is a real constant~$c \neq 0$ such that (see~\cite[Theorem~1.3]{lagrange} or~\cite[Proposition~1.4.1]{cfs})
\[ 
\tr x = c \qquad \text{for all~$x \in M$} \:. \]
Restricting attention to operators with fixed trace, the trace constraint~\eqref{trconstraint}
is equivalent to the volume constraint~\eqref{volconstraint} and may be disregarded.
The boundedness constraint, on the other hand, can be treated with a Lagrange multiplier.
More precisely, in~\cite[Theorem~1.3]{lagrange} it is shown that for every minimizing measure~$\rho$, 
there is a Lagrange multiplier~$\kappa>0$ such that~$\rho$ is a critical point of the causal action
with the Lagrangian replaced by
\beq \label{Lkappa}
\L_\kappa(x,y) := \L(x,y) + \kappa\, |xy|^2 \:,
\eeq
leaving out the boundedness constraint.
Having treated the constraints, the difference to causal variational principles is that
in the setting of causal fermion systems, the set of operators~$\F \subset \Lin(\H)$ does not have the
structure of a manifold. In order to give this set a manifold structure,
we assume that a given minimizing measure~$\rho$ is {\em{regular}} in the sense that all operators in its support
have exactly~$n$ positive and exactly~$n$ negative eigenvalues.
This leads us to introduce the set~$\F^\reg$ 
as the set of all operators~$F$ on~$\H$ with the following properties:
\begin{itemize}[leftmargin=2em]
\item[(i)] $F$ is selfadjoint, has finite rank and (counting multiplicities) has
exactly~$n$ positive and~$n$ negative eigenvalues. \\[-0.8em]
\item[(ii)] The trace is constant, i.e
\beq \tr(F) = c>0 \:. \label{trconst2}
\eeq
\end{itemize}
The set~$\F^\reg$ has a smooth manifold structure
(see the concept of a flag manifold in~\cite{helgason} or the detailed construction
in~\cite[Section~3]{gaugefix}). In this way, the causal action principle becomes an
example of a causal variational principle.

This finite-dimensional setting has the drawback that the total volume~$\rho(\F)$ of spacetime
is finite, which is not suitable for describing asymptotically flat spacetimes or spacetimes of
infinite lifetime like Minkowski space. Therefore, it is important to also consider
the {\em{infinite-dimensional setting}} where~$\dim \H=\infty$ and consequently also~$\rho(\F) = \infty$ 
(see~\cite[Exercise~1.3]{cfs}). In this case, the set~$\F^\reg$
has the structure of an infinite-dimensional Banach manifold (for details see~\cite{banach}).
Here we shall not enter the subtleties of infinite-dimensional analysis. Instead,
we get by with the following simple method: Given a minimizing measure~$\rho$,
we choose~$\F^\reg$ as a finite-dimensional manifold which contains~$M:=\supp \rho$.
We then restrict attention to variations of~$\rho$ in the class of regular Borel measures on~$\F^\reg$.
In this way, we again get into the setting of causal variational principles.
We refer to this method by saying that we {\em{restrict attention to locally compact variations}}.
Keeping in mind that the dimension of~$\F^\reg$ can be chosen arbitrarily large, this
method seems a sensible technical simplification. In situations when it is important to
work in infinite dimensions (for example for getting the connection
to the renormalization program in quantum field theory), it may be necessary to analyze
the limit when the dimension of~$\F^\reg$ tends to infinity, or alternatively it may be suitable to work in
the infinite-dimensional setting as developed in~\cite{banach}. However, this is not a concern of the present paper,
where we try to keep the mathematical setup as simple as possible.

For ease of notation, in what follows we will omit the superscript ``$\reg$.'' Thus~$\F$ stands
for a smooth (in general non-compact) manifold which contains the support~$M$ of a given minimizing
measure~$\rho$.

\subsection{The Euler-Lagrange Equations and Jet Spaces} \label{secEL}
A minimizer of a causal variational principle
satisfies the following {\em{Euler-Lagrange (EL) equations}}:
For a suitable value of the parameter~$\s>0$,
the lower semi-continuous function~$\ell : \F \rightarrow \R_0^+$ defined by
\[ 
\ell(x) := \int_M \L(x,y)\: d\rho(y) - \s \]
is minimal and vanishes on spacetime~$M:= \supp \rho$,
\beq \label{EL}
\ell|_M \equiv \inf_\F \ell = 0 \:.
\eeq
The parameter~$\s$ can be understood as the Lagrange parameter
corresponding to the volume constraint. For the derivation and further details we refer to~\cite[Section~2]{jet}.

The EL equations~\eqref{EL} are nonlocal in the sense that
they make a statement on the function~$\ell$ even for points~$x \in \F$ which
are far away from spacetime~$M$.
It turns out that for the applications we have in mind, it is preferable to
evaluate the EL equations only locally in a neighborhood of~$M$.
This leads to the {\em{weak EL equations}} introduced in~\cite[Section~4]{jet}.
Here we give a slightly less general version of these equations which
is sufficient for our purposes. In order to explain how the weak EL equations come about,
we begin with the simplified situation that the function~$\ell$ is smooth.
In this case, the minimality of~$\ell$ implies that the derivative of~$\ell$
vanishes on~$M$, i.e.\
\beq \label{ELweak}
\ell|_M \equiv 0 \qquad \text{and} \qquad D \ell|_M \equiv 0
\eeq
(where~$D \ell(p) : T_p \F \rightarrow \R$ is the derivative).
In order to combine these two equations in a compact form,
it is convenient to consider a pair~$\u := (a, \bu)$
consisting of a real-valued function~$a$ on~$M$ and a vector field~$\bu$
on~$T\F$ along~$M$, and to denote the combination of 
multiplication of directional derivative by
\beq \label{Djet}
\nabla_{\u} \ell(x) := a(x)\, \ell(x) + \big(D_\bu \ell \big)(x) \:.
\eeq
Then the equations~\eqref{ELweak} imply that~$\nabla_{\u} \ell(x)$
vanishes for all~$x \in M$.
The pair~$\u=(a,\bu)$ is referred to as a {\em{jet}}.

In the general lower-continuous setting, one must be careful because
the directional derivative~$D_\bu \ell$ in~\eqref{Djet} need not exist.
Our method for dealing with this problem is to restrict attention to vector fields
for which the directional derivative is well-defined.
Moreover, we must specify the regularity assumptions on~$a$ and~$\bu$.
To begin with, we always assume that~$a$ and~$\bu$ are {\em{smooth}} in the sense that they
have a smooth extension to the manifold~$\F$
(for more details see~\cite[Section~2.2]{fockbosonic}). Thus the jet~$\u$ should be
an element of the jet space
\[ 
\J_\rho := \big\{ \u = (a,\bu) \text{ with } a \in C^\infty(M, \R) \text{ and } \bu \in \Gamma(M, T\F) \big\} \:, \]
where~$C^\infty(M, \R)$ and~$\Gamma(M,T\F)$ denote the space of smooth real-valued functions and
smooth vector fields on~$M$, respectively.

Clearly, the fact that a jet~$\u$ is smooth does not imply that the functions~$\ell$
or~$\L$ are differentiable in the direction of~$\u$. This must be ensured by additional
conditions which are satisfied by suitable subspaces of~$\J_\rho$
which we now introduce.
First, we let~$\Gdiff_\rho$ be those vector fields for which the
directional derivative of the function~$\ell$ exists,
\[ \Gdiff_\rho = \big\{ \bu \in C^\infty(M, T\F) \;\big|\; \text{$D_{\bu} \ell(x)$ exists for all~$x \in M$} \big\} \:. \]
This gives rise to the jet space
\[ \Jdiff_\rho := C^\infty(M, \R) \oplus \Gdiff_\rho \;\subset\; \J_\rho \:. \]
For the jets in~$\Jdiff_\rho$, the combination of multiplication and directional derivative
in~\eqref{Djet} is well-defined. 
We choose a linear subspace~$\Jtest_\rho \subset \Jdiff_\rho$ with the property
that its scalar and vector components are both vector spaces,
\[ \Jtest_\rho = \Ctest(M, \R) \oplus \Gtest_\rho \;\subseteq\; \Jdiff_\rho \:, \]
and the scalar component is nowhere trivial in the sense that
\beq \label{Cnontriv}
\text{for all~$x \in M$ there is~$a \in \Ctest(M, \R)$ with~$a(x) \neq 0$}\:.
\eeq
We shall also impose conditions on the vector components (see for example~\eqref{fermiassume} below).
Finally, compactly supported jets are denoted by a subscript zero, like for example
\[ 
\Jtest_{\rho,0} := \{ \u \in \Jtest_\rho \:|\: \text{$\u$ has compact support} \} \:. \]
Then the {\em{weak EL equations}} read (for details cf.~\cite[(eq.~(4.10)]{jet})
\beq \label{ELtest}
\nabla_{\u} \ell|_M = 0 \qquad \text{for all~$\u \in \Jtest_\rho$}\:.
\eeq
Before going on, we point out that the weak EL equations~\eqref{ELtest}
do not hold only for minimizers, but also for critical points of
the causal action. With this in mind, all methods and results of this paper 
do not apply only to
minimizers, but more generally to critical points of the causal variational principle.
For brevity, we also refer to a measure with satisfies the weak EL equations~\eqref{ELtest}
as a {\em{critical measure}}.

When taking higher jet derivatives, we always take the partial derivatives computed in distinguished
charts (for details see~\cite[Section~5.2]{banach}).
Here and throughout this paper, we use the following conventions for partial derivatives and jet derivatives:
\bitem
\itemD Partial and jet derivatives with an index $i \in \{ 1,2 \}$ only act on the respective variable of the function $\L$.
This implies, for example, that the derivatives commute,
\[ 
\nabla_{1,\v} \nabla_{1,\u} \L(x,y) = \nabla_{1,\u} \nabla_{1,\v} \L(x,y) \:. \]
\itemD The partial or jet derivatives which do not carry an index act as partial derivatives
on the corresponding argument of the Lagrangian. This implies, for example, that
\[ \nabla_\u \int_\F \nabla_{1,\v} \, \L(x,y) \: d\rho(y) =  \int_\F \nabla_{1,\u} \nabla_{1,\v}\, \L(x,y) \: d\rho(y) \:. \]
\eitem
We point out that, with these conventions, {\em{jets are never differentiated}}.

\subsection{The Euler-Lagrange Equations for the Physical Wave Functions} \label{secELCFS}
For causal fermion systems, the EL equations can be expressed in terms
of the physical wave functions, as we now recall. These equations were first derived
in~\cite[\S1.4.1]{cfs} (based on a weaker version in~\cite[Section~3.5]{pfp}),
even before the jet formalism was developed. We now make the connection between the
different formulations, in a way most convenient for our constructions.
Our starting point is the formula~\eqref{Fid} expressing the spacetime point operator
as a local correlation operator. Varying the wave evaluation operator gives a vector field~$\bu$ on~$\F$
along~$M$,
\beq \label{ufermi}
\bu(x) = -\delta \Psi(x)^*\, \Psi(x) - \Psi(x)^*\, \delta \Psi(x) \:.
\eeq
In order to make mathematical sense of this formula in agreement with the concept of restricting
attention to locally compact variations, we choose a finite-dimensional subspace~$\H^\fermi \subset \H$, i.e.\
\[ f^\fermi := \dim \H^\fermi  < \infty \]
and impose the following assumptions on~$\delta \Psi$
(similar variations were first considered in~\cite[Section~7]{perturb}):
\begin{itemize}[leftmargin=2em]
\item[\rm{(a)}] 
The variation is trivial on the orthogonal complement of~$\H^\fermi$,
\[ \delta \Psi |_{(\H^\fermi)^\perp} = 0 \:. \]
\item[\rm{(b)}] The variations of all physical wave functions are continuous and compactly supported, i.e.
\[ \delta \Psi : \H \rightarrow C^0_0(M, SM) \:. \]
\end{itemize}
We choose~$\Gfermi_{\rho,0}$ as a space of vector fields of the form~\eqref{ufermi}
for~$\delta \Psi$ satisfying the above conditions~(a) and~(b).
For convenience, we identify the vector field with the first variation~$\delta \Psi$ and
write~$\delta \Psi \in \Gfermi_{\rho,0}$ (this representation of~$\bu$ in terms of~$\delta \Psi$
may not be unique, but this is of no relevance for what follows).
Choosing trivial scalar components, we obtain a corresponding space of jets~$\J^\fermi_{\rho,0}$,
referred to as the {\em{fermionic jets}}. We always assume that the fermionic jets are
admissible for testing, i.e.\
\beq \label{fermiassume}
\J^\fermi_{\rho,0} := \{0\} \oplus \Gfermi_{\rho,0} \;\subset\; \Jtest_\rho \:.
\eeq
Moreover, in analogy to the condition~\eqref{Cnontriv} for the scalar components of the test jets,
we assume that the variation can have arbitrary values at any spacetime point, i.e.\
\beq \label{Gnontriv}
\text{for all~$x \in M, \chi \in S_x$ and~$\phi \in \H^\fermi$ there is~$\delta \Psi \in \Gfermi_{\rho,0}$ with~$\delta \Psi(x)\, \phi = \chi$}\:.
\eeq

For the computation of the variation of the Lagrangian, one can make use of the fact
that for any $p \times q$-matrix~$A$ and any~$q \times p$-matrix~$B$,
the matrix products~$AB$ and~$BA$ have the same non-zero eigenvalues, with the same
algebraic multiplicities. As a consequence, applying again~\eqref{Fid},
\beq
x y 
= \Psi(x)^* \,\big( \Psi(x)\, \Psi(y)^* \Psi(y) \big)
\simeq \big( \Psi(x)\, \Psi(y)^* \Psi(y) \big)\,\Psi(x)^* \:, \label{isospectral}
\eeq
where $\simeq$ means that the operators have the same non-trivial eigenvalues
with the same algebraic multiplicities. Introducing the {\em{kernel of the fermionic projector}} $P(x,y)$ by
\[ P(x,y) := \Psi(x)\, \Psi(y)^* \::\: S_y \rightarrow S_x \:, \]
we can write~\eqref{isospectral} as
\[ x y \simeq P(x,y)\, P(y,x) \::\: S_x \rightarrow S_x \:. \]
In this way, the eigenvalues of the operator product~$xy$ as needed for the computation of
the Lagrangian~\eqref{Lagrange} and the spectral weight~\eqref{sw} are recovered as
the eigenvalues of a $2n \times 2n$-matrix. Since~$P(y,x) = P(x,y)^*$,
the Lagrangian~$\L_\kappa(x,y)$ in~\eqref{Lkappa} can be expressed in terms of~$P(x,y)$.
Consequently, the first variation of the Lagrangian can be expressed in terms
of the first variation of this kernel. Being real-valued and real-linear in~$\delta P(x,y)$,
it can be written as
\beq \label{delLdef}
\delta \L_\kappa(x,y) = 2 \re \Tr_{S_x} \!\big( Q(x,y)\, \delta P(x,y)^* \big)
\eeq
with a kernel~$Q(x,y)$ which is again symmetric (with respect to the spin inner product), i.e.
\[ Q(x,y) \::\: S_y \rightarrow S_x \qquad \text{and} \qquad Q(x,y)^* = Q(y,x) \]
(more details on this method and many computations can be found in~\cite[Sections~1.4 and~2.6
as well as Chapters~3-5]{cfs}). Expressing the variation of~$P(x,y)$ in terms of~$\delta \Psi$,
the variations of the Lagrangian can be written as
\begin{align*}
D_{1,\bu} \L_\kappa(x,y) = 2\, \re \tr \big( \delta \Psi(x)^* \, Q(x,y)\, \Psi(y) \big) \\
D_{2,\bu} \L_\kappa(x,y) = 2\, \re \tr \big( \Psi(x)^* \, Q(x,y)\, \delta \Psi(y) \big)
\end{align*}
(where~$\tr$ denotes the trace of a finite-rank operator on~$\H$). Likewise, the
variation of~$\ell$ becomes
\[ D_\bu \ell(x) = 2\, \re \int_M \tr \big( \delta \Psi(x)^* \, Q(x,y)\, \Psi(y) \big)\: d\rho(y) \:. \]
The weak EL equations~\eqref{ELtest} imply that this expression vanishes for any~$\bu \in \Gfermi_{\rho,0}$.
Using that the variation can be arbitrary at every spacetime point (see~\eqref{Gnontriv}), one may be tempted
to conclude that
\[ 
\int_M Q(x,y)\, \Psi(y) \:\phi\: d\rho(y) = 0 \qquad \text{for all~$x \in M$ and~$\phi \in \H^\fermi$}\:. \]
However, we must take into account that the local trace must be
preserved in the variation~\eqref{trconst2}. This can be arranged by rescaling the
operator~$x$ in the variation (for details see~\cite[Section~6.2]{perturb}) or, equivalently,
by treating it with a Lagrange multiplier term (see~\cite[\S1.4.1]{cfs}). We thus obtain
the {\em{EL equation for the physical wave functions}}
\beq \label{ELQ}
\int_M Q(x,y)\, \Psi^\fermi(y) \:d\rho(y) = \mathfrak{r}\, \Psi^\fermi(x) \qquad \text{for all~$x \in M$}\:,
\eeq
where~$\mathfrak{r} \in \R$ is the Lagrange parameter of the trace constraint,
and~$\Psi^\fermi:=\Psi|_{\H^\fermi}$ denotes the restriction of the wave evaluation operator
to the finite-dimensional subspace~$\H^\fermi$.

Let us briefly discuss the structure of the obtained EL equation~\eqref{ELQ}.
Being a linear equation for every physical wave function, it has similarity with the Dirac equation.
The interaction is taken into
account because the kernel~$Q(x,y)$ also depends on the ensemble of wave functions.
However, as a major difference to the Dirac equation, the EL equation~\eqref{ELQ}
only describes the occupied states of the system.
More concretely, in the example of the Dirac sea vacuum
(see for example~\cite{oppio, neumann}), the physical wave functions correspond to
the negative-energy solutions of the Dirac equation. But the solutions of positive energy
are not described by~\eqref{ELQ}. It is the main goal of the present paper to extend~\eqref{ELQ}
in such a way that the solutions of positive energy are included.

We finally comment on the significance of the subspace~$\H^\fermi \subset \H$.
Choosing a finite-dimensional subspace is a technical simplification, made in agreement with the
method of restricting attention to locally compact variations discussed in Section~\ref{seccfscvp}.
The strategy is to choose~$\H^\fermi$ large enough to capture all the relevant physical effects.
Our physical picture is that~$\H^\fermi$ should contain all physical wave functions whose energies
are much smaller than the Planck energy and which are therefore accessible to measurements
as describing particle or anti-particle states.
If necessary, one could analyze the limit where the dimension~$f^\fermi$ of~$\H^\fermi$ tends to infinity.
In what follows, we leave~$\H^\fermi$ unspecified as being any finite-dimensional subspace of~$\H$.

\subsection{The Linearized Field Equations} \label{seclinfield}
In simple terms, the linearized field equations
describe variations of the universal measure which preserve the EL equations.
More precisely, we consider variations where we multiply~$\rho$ by a non-negative smooth
function and take the push-forward with respect to a smooth mapping from~$M$ to~$\F$.
Thus we consider families of measures~$(\tilde{\rho}_\tau)_{\tau \in (-\delta, \delta)}$ 
of the form
\beq \label{rhotau}
\tilde{\rho}_\tau = (F_\tau)_* \big( f_\tau \, \rho \big)
\eeq
with functions
\beq \label{fFsmooth}
f_\tau \in C^\infty\big(M, \R^+ \big) \qquad \text{and} \qquad
F_\tau \in C^\infty\big(M, \F \big) \:,
\eeq
which also depend smoothly on the parameter~$\tau$
and have the properties~$f_0(x)=1$ and~$F_0(x) = x$ for all~$x \in M$
(here the push-forward measure is defined
for a Borel subset~$\Omega \subset \F$ by~$((F_\tau)_*\mu)(\Omega)
= \mu ( F_\tau^{-1} (\Omega))$; see for example~\cite[Section~3.6]{bogachev}).
If we demand that~$(\tilde{\rho}_\tau)_{\tau \in (-\delta, \delta)}$ is a family of minimizers,
the EL equations~\eqref{EL} hold for all~$\tau$, i.e.
\beq \label{elltau}
\tilde{\ell}_\tau|_{M_\tau} \equiv \inf_\F \ell_\tau = 0 
\qquad \text{with} \qquad \tilde{\ell}_\tau(x) := \int_\F \L_\kappa(x,y)\: d\tilde{\rho}_\tau(y) - \s \:,
\eeq
where~$M_\tau$ is the support of the varied measure,
\[ M_\tau := \supp \tilde{\rho}_\tau = \overline{F_\tau(M)} \:. \]
In~\eqref{elltau} we can express~$\tilde{\rho}$ in terms of~$\rho$. Moreover,
it is convenient to rewrite this equation as an equation on~$M$ and to multiply it
by~$f_\tau(x)$. We thus obtain the equivalent equation
\[ \ell_\tau|_M \equiv \inf_\F \ell_\tau = 0 \]
with
\[ \ell_\tau(x) :=
\int_\F f_\tau(x) \,\L_\kappa\big(F_\tau(x),F_\tau(y) \big)\: f_\tau(y)\: d\tilde{\rho}_\tau(y) - f_\tau(x)\: \s \]
In analogy to~\eqref{ELtest} we write the corresponding weak EL equations as
\[ 
\nabla_{\u} \ell_\tau|_M = 0 \qquad \text{for all~$\u \in \Jtest_\rho$} \]
(for details on why the jet space does not depend on~$\tau$ we refer to~\cite[Section~4.1]{perturb}).
Since this equation holds by assumption for all~$\tau$, we can differentiate it with respect to~$\tau$.
Denoting the infinitesimal generator of the variation by~$\v$, i.e.
\[ \label{vvary}
\v(x) := \frac{d}{d\tau} \big( f_\tau(x), F_\tau(x) \big) \Big|_{\tau=0} \:, \]
we obtain the {\em{linearized field equations}}
\[ 
0 = \la \u, \Delta \v \ra(x) := 
\nabla_\u \bigg( \int_M \big( \nabla_{1, \v} + \nabla_{2, \v} \big) \L_\kappa(x,y)\: d\rho(y) - \nabla_\v \,\s \bigg) \:, \]
which are to be satisfied for all~$\u \in \Jtest_\rho$ and all~$x \in M$
(for details see~\cite[Section~3.3]{perturb}).
We denote the vector space of all solutions of the linearized field equations by~$\Jlin_\rho$.

\subsection{Surface Layer Integrals} \label{secosi}
{\em{Surface layer integrals}} were first introduced in~\cite{noether}
as double integrals of the general form
\beq \label{osi}
\int_\Omega \bigg( \int_{M \setminus \Omega} (\cdots)\: \L_\kappa(x,y)\: d\rho(y) \bigg)\, d\rho(x) \:,
\eeq
where $(\cdots)$ stands for a suitable differential operator formed of jets, and~$\Omega$
is a Borel subset of~$M$.
A surface layer integral generalizes the concept of a surface integral over~$\partial \Omega$
to the setting of causal fermion systems.
The connection can be understood most easily in the case when~$\L_\kappa(x,y)$ vanishes
unless~$x$ and~$y$ are close together. In this case, we only get a contribution to~\eqref{osi}
if both~$x$ and~$y$ are close to the boundary of~$\Omega$.
A more detailed explanation of the idea of a surface layer integrals is given in~\cite[Section~2.3]{noether}.

In~\cite{noether, jet, osi}, {\em{conservation laws}} for surface layer integrals were derived.
The statement is that if~$\v$ describes a symmetry of the system or if~$\v$ satisfies the
linearized field equations, then suitable surface layer integrals~\eqref{osi} vanish for every
compact~$\Omega \subset M$. The conserved surface layer integrals of relevance here are
\begin{align}
\gamma^\Omega_\rho \::\:\Jlin_\rho \cap& \Jtest_\rho \rightarrow \R \qquad \text{(conserved one-form)} \notag \\
\gamma^\Omega_\rho(\v) &:= \int_\Omega d\rho(x) \int_{M \setminus \Omega} d\rho(y)\:
\big( \nabla_{1,\v} - \nabla_{2,\v} \big) \L_\kappa(x,y) \label{Ilin0} \\
\sigma^\Omega_\rho \::\: \Jlin_\rho \cap& \Jtest_\rho \times \Jlin_\rho \cap \Jtest_\rho \rightarrow \R  \qquad \text{(symplectic form)} \notag \\
\sigma^\Omega_\rho(\u,\v) &:= \int_\Omega d\rho(x) \int_{M \setminus \Omega} d\rho(y)\:
\big( \nabla_{1, \u} \nabla_{2, \v} - \nabla_{1, \v} \nabla_{2, \u} \big) \, \L_\kappa(x,y) \:. \label{I2asymm}
\end{align}
For the conserved one-form, the proof of the conservation law will be repeated in
Lemma~\ref{lemmaequivalent}.

\section{The Commutator Inner Product} \label{seccomminner}
In~\cite[Section~5]{noether} it was shown that the unitary invariance of the causal action
gives rise to a conservation law which generalizes current conservation to the setting of causal fermion systems.
We now present this conservation law from a new perspective, which will serve as the
starting point for the constructions in Section~\ref{secHextend}.

\subsection{Unitary Invariance and Commutator Jets} \label{seccommute}
The causal action principle is {\em{unitarily invariant}} in the following sense. Let~$\scrU \in \U(\H)$ be a
unitary transformation. Given a measure~$\rho$ on~$\F$, we can unitarily transform the measure by
setting
\beq \label{Urhodef}
(\scrU \rho)(\Omega) := \rho \big( \scrU^{-1} \,\Omega\, \scrU \big) 
\qquad \text{for} \qquad \Omega \subset \F \:.
\eeq
Since the eigenvalues of an operator are invariant under unitary transformations,
the measure~$\rho$ is a minimizer or critical point of the causal action principle if and only
if~$\scrU \rho$ is. This makes it possible to construct solutions of the linearized field equations, 
as we now explain.
Let~$\rho$ be a critical measure. Moreover, let~$(\scrU_\tau)_{\tau \in [0, \tau_{\max}]}$ be a
smooth family of unitary transformations with generator
\beq \label{Agen}
\scrA := -i \frac{d}{d\tau}\, \scrU_\tau \big|_{\tau=0} \:,
\eeq
which we assume to have finite rank.
According to~\eqref{Urhodef}, the support of the measures~$\tilde{\rho}_\tau:= \scrU_\tau \rho$ is
given by
\beq \label{Munit}
\tilde{M}_\tau := \supp \tilde{\rho}_\tau = \scrU_\tau\, M \, \scrU_\tau^{-1} \:.
\eeq
Due to the unitary invariance of the Lagrangian, the measures~$\tilde{\rho}_\tau$ all satisfy the EL equations. 
As a consequence, the infinitesimal generator
of the family is a solution of the linearized field equations:
\begin{Lemma} Assume that the jet
\beq \label{jvdef}
\mathfrak{C} :=(0,\Comm) \qquad \text{with} \qquad \Comm(x) := i \big[\scrA, x \big]
\eeq
has the property
\beq \label{DuComm}
D_u \Comm \in \Jtest \qquad \text{for all~$\u \in \Jtest$}
\eeq
(where the directional derivatives are computed in the distinguished charts mentioned
in Section~\ref{secEL}). Then the jet~$\mathfrak{C}$ is a solution of the linearized field equations, i.e.
\beq \label{lincomm}
\nabla_\u \int_M (D_{1,\Comm} + D_{2,\Comm} )\, \L_\kappa(x,y)\: d\rho(y) = 0 \qquad \text{for all~$\u \in \Jtest_\rho$}\:.
\eeq
\end{Lemma}
\Proof One method of proof would be to differentiate through the EL equations.
However, this would involve a transformation of the space of test jets
(similar as explained for example in~\cite[Section~3.1]{osi}).
Here we prefer to show that the integrand of~\eqref{lincomm} vanishes identically. Indeed,
due to the unitary invariance of the Lagrangian,
\[ \L_\kappa\big( \scrU_\tau x \scrU_\tau^{-1},\: \scrU_\tau y \scrU_\tau^{-1} \big) = \L_\kappa(x,y)\:. \]
Differentiating with respect to~$\tau$ gives
\beq \label{D12zero}
(D_{1,\Comm} + D_{2,\Comm} )\, \L_\kappa(x,y)\: d\rho(y) = 0 \:.
\eeq
Hence the integrand in~\eqref{lincomm} vanishes for all~$x,y \in \F$. As a consequence,
the integral in~\eqref{lincomm} vanishes for all~$x \in \F$.
Consequently, also its derivative in the direction of~$u$ vanishes.
Using our convention that the
jet derivatives act only on the Lagrangian (see Section~\ref{secEL}),
the directional derivative differs from the derivative by the term~$D_{D_u \Comm} \ell_\kappa(x)$.
This term vanishes in view of~\eqref{DuComm} and the weak EL equations~\eqref{ELtest}.
\QED
Due to the commutator in~\eqref{jvdef}, we refer to jets of this form
as the {\bf{commutator jets}}~$\JComm_\rho \subset \Jlin_\rho$.
Restricting attention to the vector component, we also write~$\GComm_\rho \subset \Glin_\rho$.
We also point out that for the argument~$x$ in~\eqref{jvdef} we can choose any operator~$x \in \F$.
Therefore, every commutator jet extends to a vector field on~$\F$,
\[ \Comm \in \Gamma(\F, T\F)\:. \]

\subsection{Time Orientation and Past Sets} \label{secpastset}
In what follows, we shall restrict attention to  commutator jets
of the form~\eqref{lincomm} with~$\scrA$ a linear operator on~$\H^\fermi$, i.e.
\beq \label{Af}
\scrA|_{\H^\fermi} \::\: \H^\fermi \rightarrow \H^\fermi \qquad \text{and} \qquad 
\scrA|_{(\H^\fermi)^\perp} = 0 \:.
\eeq
We refer to these jets as the {\em{commutator jets on~$\H^\fermi$}}
and denote the space of all these jets by~$\J^{\fermi, \Comm}_\rho \subset \J^\Comm_\rho$,
and the corresponding vector fields by~$\Gamma^{\fermi, \Comm}_\rho \subset \Gamma^\Comm_\rho$.
Before beginning, we specify regularity assumptions for these commutator jets, which can
be understood as implicit assumptions on the fermionic subspace~$\H^\fermi$.
\begin{Def} The causal fermion system is {\bf{commutator regular}} if~$\J^{\fermi,\Comm}_\rho \subset \Jtest_\rho$
and if
\beq \label{commreg}
D_\Comm \int_M \L(x,y)\: d\rho(y) = \int_M D_{1,\Comm} \L(x,y)\: d\rho(y) \qquad \text{for
all~$\Comm \in \Gamma^{\fermi, \Comm}_\rho$}\:.
\eeq
\end{Def}

We first specify the subsets of spacetime for which the conserved one-form~\eqref{Ilin0}
is well-defined for all commutator jets on~$\H^\fermi$.
\begin{Def} A Borel subset~$\Omega \subset M$ is called {\bf{surface layer finite}}
if for all commutator jets~$\Comm \in \Gamma^{\fermi, \Comm}_\rho$,
\beq \label{slf}
\int_\Omega d\rho(x) \int_{M \setminus \Omega} d\rho(y)\:
\Big| (D_{1,\Comm} - D_{2,\Comm} )\, \L_\kappa(x,y) \Big| < \infty \:.
\eeq
\end{Def} \noindent
We note for clarity that in the last double integral the
order of integration may be interchanged in view of Tonelli's theorem.

For the conservation law to hold, we need an additional assumption:
\begin{Def} \label{defcequi} Two surface layer finite subsets~$\Omega, \tilde{\Omega} \subset M$ are
called {\bf{causally equivalent}} if for all commutator jets~$\Comm \in \Gamma^{\fermi, \Comm}_\rho$,
\beq \label{cequi}
\int_L d\rho(x) \int_L d\rho(y)\: \Big| (D_{1,\Comm} - D_{2,\Comm} )\, \L_\kappa(x,y) \Big| < \infty \:,
\eeq
where~$L:= (\Omega \setminus \Omega') \cup (\Omega' \setminus \Omega)$.
\end{Def}
The next lemma shows that the commutator inner product is conserved
for causally equivalent sets; the proof follows the idea in~\cite[proof of Theorem~3.3]{noether}.
\begin{Lemma} \label{lemmaequivalent}
Assume that the causal fermion system is commutator regular. Then
for two causally equivalent subsets~$\Omega, \Omega' \subset M$,
the conserved one-forms coincide for all commutator jets in~$\H^\fermi$, i.e.\
\[ \gamma^\Omega_\rho(\Comm) = \gamma^{\Omega'}_\rho(\Comm) \qquad \text{for all~$(0, \Comm) \in \J^{\fermi, \Comm}_\rho$} \:. \]
\end{Lemma}
\Proof Setting
\[ L_1 = \Omega \setminus \Omega' \:,\quad
L_2 = \Omega' \setminus \Omega\:,\quad
A=\Omega \setminus L_1\:,\quad
B=M \setminus (\Omega \cup L_2) \:, \]
we decompose the sets~$\Omega$, $\Omega'$ and their complements as
\[ \Omega = A \dot{\cup} L_1 \:,\qquad M \setminus \Omega = B \dot{\cup} L_2 \:,\quad
\Omega' = A \dot{\cup} L_2 \:,\qquad M \setminus \Omega' = B \dot{\cup} L_1 \:. \]
In view of~\eqref{slf}, the double integrals of the surface layer integrals are well-defined in the
Lebesgue sense. Using the above decompositions and using linearity of the integrals,
a straightforward computation shows that
\begin{align*}
\gamma^\Omega_\rho(\Comm) - \gamma^{\Omega'}_\rho(\Comm)
&= \int_{L_1} d\rho(x) \int_{M \setminus L_1} d\rho(y)\: (D_{1,\Comm} - D_{2,\Comm} )\, \L_\kappa(x,y) \\
&\quad\;
- \int_{L_2} d\rho(x) \int_{M \setminus L_2} d\rho(y)\: (D_{1,\Comm} - D_{2,\Comm} )\, \L_\kappa(x,y) \:,
\end{align*}
where all integrals are again well-defined in the Lebesgue sense.
Using~\eqref{cequi}, we can add the integrals over~$L_1 \times L_1$ and~$L_2 \times L_2$ to obtain
\begin{align*}
\gamma^\Omega_\rho(\Comm) - \gamma^{\Omega'}_\rho(\Comm)
&= \int_{L_1} d\rho(x) \int_{M} d\rho(y)\: (D_{1,\Comm} - D_{2,\Comm} )\, \L_\kappa(x,y) \\
&\quad\;
- \int_{L_2} d\rho(x) \int_{M} d\rho(y)\: (D_{1,\Comm} - D_{2,\Comm} )\, \L_\kappa(x,y) \:.
\end{align*}
Using~\eqref{D12zero} together with~\eqref{commreg}, we get
\begin{align*}
\gamma^\Omega_\rho(\Comm) - \gamma^{\Omega'}_\rho(\Comm)
&= 2 \int_{L_1} D_\Comm \ell(x) \:d\rho(x) - 2 \int_{L_2} d\rho(x) D_\Comm \ell(x) \:d\rho(x) \:,
\end{align*}
and applying the weak EL equation~\eqref{ELtest} gives the result.
\QED

We next show that causal equivalence is indeed an equivalence relation.
\begin{Lemma} \label{lemmaequiv}
The notion of causal equivalence defines an equivalence relation on the Borel subsets of~$M$.
\end{Lemma}
\Proof Since symmetry is obvious, it remains to prove transitivity.
Thus assume that~$\Omega$ is causally equivalent to~$\Omega'$, and that~$\Omega'$ is
causally equivalent to~$\Omega''$. Setting
\[ L = (\Omega' \setminus \Omega) \cup (\Omega \setminus \Omega') \:,\quad
\tilde{L} = (\Omega'' \setminus \Omega') \cup (\Omega' \setminus \Omega'') \:,\quad
\hat{L} = (\Omega'' \setminus \Omega) \cup (\Omega \setminus \Omega'') \:, \]
our task is to show that the function~$g:=| (D_{1,\Comm} - D_{2,\Comm} )\, \L_\kappa(.,.)|$
is integrable on~$\hat{L} \times \hat{L}$. Noting that~$\hat{L} \subset L \cup \tilde{L}$ and
\[ \hat{L} \times \hat{L} \;\subset\; (L \times L) \cup (\tilde{L} \times \tilde{L})
\cup \big(L \times (M \setminus L)\big) \cup \big(\hat{L} \times (M \setminus \hat{L})\big) \:, \]
the function~$g$ is integrable on~$L \times L$ because~$\Omega \sim \Omega'$
and on~$\tilde{L} \times \tilde{L}$ because~$\Omega' \sim \Omega''$.
On the other hand, it is is integrable on~$L \times (M \setminus L)$ because~$\Omega$ and~$\Omega'$
are commutator regular, and it is integrable on~$\hat{L} \times (M \setminus \hat{L})$ because~$\Omega'$ and~$\Omega''$ are commutator regular. This concludes the proof.
\QED
We denote this equivalence relation by~$\Omega \sim \tilde{\Omega}$.
Clearly, the condition~\eqref{cequi} is satisfied if the sets~$\Omega$ and~$\Omega'$
differ by a compact set,
\[ \Omega \sim \Omega \cup K \qquad \text{for all~$\Omega \in {\mathfrak{B}}(M)$ and compact~$K \subset M$}\:. \]
Therefore, the equivalence classes~$[\Omega]$ with~$\Omega \in {\mathfrak{B}}(M)$
give information on the non-compact causal structure of~$M$.
In general, this structure can be quite complicated.
In what follows, we restrict attention to causal fermion systems
for which the equivalence classes have a particularly simple form,
corresponding to the usual assumptions that spacetime is connected
and time-orientable, so that there are (up to global reversals of the time orientation)
unique notions of future and past.
This concept is implemented in the following definition:
\begin{Def} The causal fermion system~$(\H, \F, \rho)$ is called
{\bf{time-orientable}} if there are precisely four equivalence classes
\[ [\varnothing], \quad [M], \quad [\Omega] \quad \text{and} \quad  [M \setminus \Omega] \]
(where~$\Omega \subset M$ is a Borel subset representing a nontrivial equivalence class).
\end{Def} \noindent
Clearly, this definition involves the freedom of reversing the time direction by
replacing~$\Omega$ by~$M \setminus \Omega$.
Since this replacement also changes the sign of~$\gamma^\Omega_\rho(\Comm)$, one
could fix the direction for example by demanding that~$\gamma^\Omega_\rho(\Comm[\scrA])$
is non-negative for every positive operator~$\scrA$. But this procedure is not compelling.
Therefore, we prefer to fix the time orientation by distinguishing~$[\Omega]$.
\begin{Def} \label{deffuturepast}
A time-orientable causal fermion system together with a choice of~$[\Omega]$ is
called {\bf{time-oriented}}.
The sets in~$[\Omega]$ are referred to as {\bf{past sets}}, whereas the sets in~$[M \setminus \Omega]$
are {\bf{future sets}}.
\end{Def}

\subsection{Representing the Scalar Product~$\la .|. \ra_\H$ by a Surface Layer Integral}
In what follows, we always let~$\Omega$ be a past set.
Evaluating the surface layer integral~$\gamma^\Omega_\rho$ for commutator jets makes it possible
to represent the Hilbert space scalar product~$\la .|. \ra_\H$ as a surface layer integral,
as we now explain. To this end, we consider more specifically families of unitary
transformations with generators~$\scrA$ of rank one. Namely, given
a non-zero vector~$\psi \in \H$, we form the symmetric linear operator~$\scrA \in \Lin(\H)$ of rank at one by
\beq \label{Adef}
\scrA \psi := \la u | \psi\ra_\H \: u
\eeq
(thus in bra/ket notation, $\scrA = | u \ra \la u |$).
By exponentiating we obtain a corresponding family of unitary operators~$(\scrU_\tau)_{\tau \in \R}$,
\[ 
\scrU_\tau := \exp(i \tau \scrA) \:. \]
We again denote the corresponding commutator jet in~\eqref{jvdef} by~${\mathfrak{C}}=(0, \Comm)$.
It is usually most convenient to only consider the vector component, giving rise to a mapping
\[ \Comm \::\: \H \rightarrow \Glin_\rho \cap \Gtest_\rho \:. \]
In view of~\eqref{Adef} and~\eqref{jvdef}, this mapping
is {\em{positive homogeneous of rank two}} in the sense that
\[ {\mathscr{C}} \big( \alpha u \big) = |\alpha|^2\: {\mathscr{C}}(u) \qquad \text{for all~$\alpha \in \C$}\:. \]
Moreover, from Lemma~\ref{lemmaequivalent} we know that the surface layer
in\-te\-gral~$\gamma^\Omega({\mathscr{C}}(u))$ is conserved in the sense that it does not depend
on the choice of the past set~$\Omega$.
This surface layer integral
defines a functional on~$\H$ which is again positive homogeneous of degree two, i.e.
\[ \gamma^\Omega_\rho \big( {\mathscr{C}}(\alpha u) \big) = |\alpha|^2\: \gamma^\Omega_\rho \big( {\mathscr{C}}(u) \big)
\qquad \text{for all~$u \in \H$ and~$\alpha \in \C$}\:. \]
Therefore, we can use the polarization formula to define a sesquilinear form
\begin{align}
\la .|. \ra^\Omega_\rho &: \H \times \H \rightarrow \C \:, \label{Commosidef} \\
\la u | v \ra^\Omega_\rho &:=
\frac{1}{4} \:\Big( \gamma^\Omega_\rho \big( {\mathscr{C}}(u+v)\big) - \gamma^\Omega_\rho \big( {\mathscr{C}}(u-v) \big) \Big) \notag \\
&\quad\; -\frac{i}{4} \:\Big( \gamma^\Omega_\rho \big( {\mathscr{C}}(u+iv) \big) - \gamma^\Omega_\rho \big( {\mathscr{C}}(u-iv) \big) \Big) \:, \notag
\end{align}
referred to as the {\bf{commutator inner product}}.
In~\cite[Section~5]{noether} it was shown that, taking the continuum limit of the vacuum in Minkowski space,
this sesquilinear form coincides, up to a constant, with the scalar product~$\la u|v\ra_\H$.
We now give this property a useful name.
For the sake of larger generality, we only assume that this property holds for all vectors in the
finite-dimensional subspace~$\H^\fermi \subset \H$.
\begin{Def} \label{defSLrep}
Given a critical measure~$\rho$ and a past set~$\Omega \subset M$, the commutator inner product
is said to {\bf{represent the scalar product}} if
\beq \label{Ccond}
\la u|v \ra^\Omega_\rho = c\, \la u|v \ra_\H \qquad \text{for all~$u,v \in \H^\fermi$}
\eeq
with a suitable positive constant~$c$.
\end{Def} \noindent
Due to the conservation of the surface layer integral (see Lemma~\ref{lemmaequivalent}),
this condition does not depend on the choice of the past set~$\Omega$.
In physical applications, one chooses~$\Omega$ as the
past of a Cauchy surface. Then the condition holds automatically, provided that it holds one specific time.
Examples for such a choosing this specific time would be shortly after the big bang
before particles were created or any other time when the particle density was so small that the Minkowski
vacuum is a good approximation.

In what follows, we always assume that the commutator inner product represents the scalar product.
Moreover, by a rescaling of the Hilbert space scalar product 
and possibly a time reversal we always arrange that the constant~$c$ in~\eqref{Ccond}
is equal to one. Thus we always assume that
\[ 
\la u|v \ra^\Omega_\rho = \la u | v \ra_\H \qquad \text{for all~$u,v \in \H^\fermi$}\:. \]

We now rewrite this condition in an equivalent way which is sometimes more useful.
First, in view of the polarization formula, it suffices to satisfy~\eqref{Ccond}
in the case~$u=v=:\psi$. We thus obtain the equivalent condition
\[ \gamma^\Omega_\rho\big( {\mathscr{C}}(\psi) \big) =  \|\psi\|^2_\H \qquad \text{for all~$\psi \in \H^\fermi$}\:. \]
Second, in view of the definition of the operator~$\scrA$ in~\eqref{Adef}, its trace
is given by~$\tr \scrA = \|\psi\|^2_\H$. We have thus proven the following result:
\begin{Lemma} \label{lemmaSLrep}
The commutator inner product represents the
scalar product with constant~$c>0$
if and only if for every symmetric operator~$\scrA \in \Lin(\H)$ on~$\H^\fermi$ (see~\eqref{Af}), the
corresponding commutator jet~$(0, \Comm)$ with
\[ \Comm(x) = i [\scrA, x] \]
satisfies the relation
\beq \label{repc}
\gamma^\Omega_\rho \big( {\mathscr{C}}(u) \big) = c\, \tr \scrA \:.
\eeq
\end{Lemma} \noindent
By a rescaling, we shall always arrange that~$c=1$.

We finally bring the commutator inner product into a more explicit and convenient form:
\begin{Prp} \label{prpcomm}
Using the kernel~$Q(x,y)$ as defined by~\eqref{delLdef},
the commutator inner product can be written as
\[ \la u|v \ra^\Omega_\rho = -2i \,\bigg( \int_{\Omega} \!d\rho(x) \int_{M \setminus \Omega} \!\!\!\!\!\!\!d\rho(y) 
- \int_{M \setminus \Omega} \!\!\!\!\!\!\!d\rho(x) \int_{\Omega} \!d\rho(y) \bigg)\:
\Sl \psi^u(x) \:|\: Q(x,y)\, \psi^v(y) \Sr_x \:. \]
\end{Prp} 
\Proof A similar computation is given in~\cite[Section~5.2]{noether}.
We repeat it here for completeness, also using the present notation.
In view of the polarization formula, it suffices to consider the case~$u=v$.
In this case, we know from~\eqref{Commosidef} and~\eqref{Ilin0} that
\beq \label{uutosi}
\la u,u \ra^\Omega_\rho = \gamma^\Omega_\rho \big( {\mathscr{C}}(u) \big)
= \int_\Omega d\rho(x) \int_{M \setminus \Omega} d\rho(y)\:
\big( D_{1,\Comm} - D_{2,\Comm} \big) \L_\kappa(x,y) \:.
\eeq

The remaining task is to compute the derivatives of the Lagrangian.
The local correlation operator 
can be expressed in terms of the wave evaluation operator by
(see~\cite[Lemma~1.1.3]{cfs})
\[ x = - \Psi(x)^* \Psi(x) \:. \]
Hence the unitary transformation of the local correlation operators in~\eqref{Munit}
corresponds to the transformation~$\Psi(x) \rightarrow \Psi(x)\, \scrU_\tau^{-1}$.
Consequently, writing the kernel of the fermionic projector
as~$P(x,y) = -\Psi(x)\, \Psi(y)^*$ (see again~\cite[Lemma~1.1.3]{cfs}), its variation is computed by
\begin{align*}
D_{1,\Comm} P(x,y) &= \frac{d}{d\tau} \Big( -\Psi(x)\, \scrU_\tau^{-1}\, \Psi(y)^* \Big) \Big|_{\tau=0}
\overset{\eqref{Agen}}{=} -\Psi(x)\,(-i \scrA)\, \Psi(y)^* \\
D_{2,\Comm} P(x,y) &= -\Psi(x)\,(i \scrA)\, \Psi(y)^* \:.
\end{align*}
Using the form of the operator~$\scrA$ in~\eqref{Adef}, we conclude that
\[ D_{1,\Comm} P(x,y) =  i\, |\psi^u(x)\Sr \Sl \psi^u(y)| \qquad \text{and} \qquad
D_{2,\Comm} P(x,y) =  -i\, |\psi^u(x)\Sr \Sl \psi^u(y)| \]
(where~$\psi^u(x) := \pi_x u$ is the physical wave function of~$u \in \H$).

Using the last relations in~\eqref{delLdef} gives
\begin{align*}
&\big( D_{1,\Comm} - D_{2,\Comm} \big) \L_\kappa(x,y) \\
&= 2 i \:\Big( 
\Tr_{S_y} \big( Q(y,x)\, |\psi^u(x)\Sr \Sl \psi^u(y)| \big) - \Tr_{S_x} \!\big( Q(x,y)\, |\psi^u(y)\Sr \Sl \psi^u(x)| \big) \Big) \\
&= 2 i \:\Big( 
\Sl \psi^u(y) \,|\, Q(y,x)\, \psi^u(x)\Sr -  \Sl \psi^u(x) \,|\, Q(x,y)\, |\psi^u(y)\Sr \Big) \:.
\end{align*}
Substituting these equations into~\eqref{uutosi} gives the result.
\QED

For the following constructions, it is preferable to have the freedom to modify the kernel~$Q(x,y)$
while preserving the EL equations for the physical wave functions~\eqref{ELQ}.
The point is that these EL equations must be satisfied only for all the physical wave functions,
giving us the freedom to modify~$Q$ arbitrarily on the complement of the span of these wave functions.
Therefore, we decompose the kernel~$Q(x,y)$ as
\beq \label{Qreg}
Q(x,y) = Q^\reg(x,y) + Q^\sing(x,y) \:,
\eeq
where we choose~$Q^\sing(x,y)$ as a symmetric kernel (i.e.\ $Q^\sing(x,y)^* = Q^\sing(y,x)$)
such that
\beq \label{Qsing}
\int_M Q^\sing(x,y)\, \psi^u(x)\: d\rho(x) =0 \qquad \text{for all~$u \in \H^\fermi$} \:.
\eeq
In what follows, we always replace~$Q$ by~$Q^\reg$. In particular, the EL equations~\eqref{ELQ}
become
\beq \label{ELQreg}
\int_M Q^\reg(x,y)\, \Psi^\fermi(y) \:d\rho(y) = \mathfrak{r}\, \Psi^\fermi(x) \qquad \text{for all~$x \in M$}\:.
\eeq
These equations are again satisfied in view of~\eqref{Qreg} and~\eqref{Qsing}.
Moreover, the kernel~$Q^\reg$ is again symmetric,
\beq \label{Qregsymm}
Q^\reg(x,y)^* = Q^\reg(y,x) \:.
\eeq
The commutator inner product of Proposition~\ref{prpcomm} is modified to
\beq \label{OSIreg}
\begin{split}
\la u|v \ra^\Omega_\rho = -2i \,\bigg( \int_{\Omega} \!d\rho(x) \int_{M \setminus \Omega} \!\!\!\!\!\!\!d\rho(y) 
&- \int_{M \setminus \Omega} \!\!\!\!\!\!\!d\rho(x) \int_{\Omega} \!d\rho(y) \bigg)\\
&\times\:
\Sl \psi^u(x) \:|\: Q^\reg(x,y)\, \psi^v(y) \Sr_x \:.
\end{split}
\eeq
This is still conserved as a consequence of~\eqref{ELQreg}, but the inclusion of~$Q^\sing$
may change its value. In order to ensure that the commutator inner product remains unchanged,
we also demand that the contribution by~$Q^\sing$ vanish for all~$u, v \in \H^\fermi$,
\beq \label{Qsing2}
\bigg( \int_{\Omega} \!d\rho(x) \int_{M \setminus \Omega} \!\!\!\!\!\!\!d\rho(y) 
- \int_{M \setminus \Omega} \!\!\!\!\!\!\!d\rho(x) \int_{\Omega} \!d\rho(y) \bigg)
\Sl \psi^u(x) \:|\: Q^\sing(x,y)\, \psi^v(y) \Sr_x = 0 \:.
\eeq
Clearly, in view of the conservation law for~\eqref{OSIreg}, it suffices to verify~\eqref{Qsing2} for
a specific past set~$\Omega$.
Apart from the conditions~\eqref{Qsing} and~\eqref{Qsing2}, the kernel~$Q^\sing(x,y)$
can be chosen arbitrarily.

Before explaining this construction in a concrete example, we make a few general remarks.
We first point out that the above ``regularization'' of~$Q$ is optional.
All the constructions and results of this paper are valid no matter if or how this regularization
is performed.
We also point out that modifying~$Q^\reg$ merely amounts to changing the representation
of the physical wave functions, but it has no influence on any physical observables,
nor does it change the interaction or the dynamics of the system.
In order to illustrate the significance of the freedom in modifying~$Q^\reg$, we now
discuss a concrete example.

\subsection{Example: The Regularized Minkowski Vacuum} \label{secexmink}
The following example explains why the above decomposition of~$Q(x,y)$
is of advantage in the applications.
For the regularized Min\-kowski vacuum, the kernel~$Q(x,y)$
was analyzed in the so-called state stability analysis carried out in~\cite[Section~5.6]{pfp},
\cite{vacstab} and~\cite{reg}, as we now briefly recall.
The detailed analysis of the continuum
limit in~\cite[Chapter~3]{cfs} shows that in order to obtain well-defined field equations
in the continuum limit, the number of generations must be equal to three.
Therefore, we now consider an unregularized fermionic projector of the vacuum
involving a sum of three Dirac seas,
\[ 
P(x,y) = \sum_{\beta=1}^3 \int \frac{d^4k}{(2 \pi)^4}\: (\slashed{k}+m_\beta)\:
\delta \big(k^2-m_\beta^2 \big)\: e^{-ik(x-y)} \]
(this configuration is also referred to as three generations in a single sector;
see~\cite[Chapter~3]{cfs}).
The corresponding kernel~$Q(x,y)$ obtained in the continuum limit depends only on
the difference vector~$y-x$ and can thus be written as the Fourier transform
of a distribution~$\hat{Q}(k)$,
\[ Q(x,y) = \int \frac{d^4k}{(2 \pi)^4} \:\hat{Q}(k)\: e^{-ik (x-y)} \:. \]
The state stability analysis in~\cite[Section~5.6]{pfp} implies that the Fourier
transform~$\hat{Q}$ has the following form (cf.~\cite[Definition~5.6.2]{pfp}).
It is is well-defined inside the lower mass cone
\[ \mathcal{C}^\land := \{ k \in \R^4 \,|\, k^i k_i >0 \text{ and } k^0<0 \} \:, \]
where it can be written as
\[ 
\hat{Q} (k) = a\:\frac{k\slsh}{|k|} + b \]
with continuous real functions $a$ and $b$ on $\mathcal{C}^\land$ having
the following properties:
\bitem
\item[(i)] $a$ and $b$ are Lorentz invariant,
\[ a = a(k^2)\:,\qquad b = b(k^2) \:. \]
\item[(ii)] $a$ is non-negative.
\item[(iii)] The function $a+b$ is minimal on the mass shells,
\[ 
(a+b)(m^2_\beta) = \inf_{q \in {\mathcal{C}}^\land} (a+b)(q^2) \quad\mbox{for~$\beta=1,2,3$}\:. \]
\eitem
Making essential use of these properties, in~\cite[Section~5.2]{noether} it
was proven that the commutator inner product reduces to the usual current integral,
proving that the commutator inner product indeed represents the scalar product (see Definition~\ref{defSLrep}).
All the above formulas hold without the need for an ultraviolet regularization.
If an ultraviolet regularization on the scale~$\varepsilon$ is introduced, the formulas are all well-behaved in the limit~$\varepsilon \searrow 0$.

The motivation for the kernel~$Q^\sing$ in~\eqref{Qreg} comes from the fact that~$\hat{Q}$
is {\em{ill-defined outside}} the lower mass cone. This also means that, if a regularization is present,
the regularized kernel~$\hat{Q}^\varepsilon$ will in general diverge inside the upper mass cone
as~$\varepsilon \searrow 0$. This behavior can be understood from the specific structure of~$Q(x,y)$
as being a product in position space of the form
\beq \label{62j}
Q(x,y) = \frac{1}{2}\: {\mathcal{M}}(x,y)\, P(x,y)\:,
\eeq
where from symmetry considerations one knows that the Fourier transform~$\hat{\mathcal{M}}$
of~${\mathcal{M}}$ is supported inside the upper and lower mass cone. Rewriting the product in position
space in~\eqref{62j} as a convolution in momentum space,
\[ \hat{Q}(q) \;=\; \frac{1}{2} \:\int \frac{d^4p}{(2\pi)^4}\:\hat{\mathcal{M}} (p) \:\hat{P}(q-p)\:, \]
the integration range is compact if~$q$ lies in the lower mass shell, but it is unbounded
if~$q$ is in the upper mass shell (see Figure~\ref{fig4}).
\begin{figure}[tb]
\begin{center}
\scalebox{0.9}
{\includegraphics{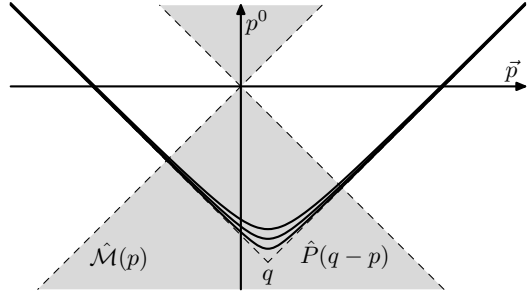}}
\caption{The convolution~${\hat{\mathcal{M}}} * \hat{P}$.}
\label{fig4}
\end{center}
\end{figure}

In position space, this divergence can be understood from the fact that both~${\mathcal{M}}(x,y)$
and~$P(x,y)$ are singular on the light cone, implying that the pointwise product in~\eqref{62j}
cannot be taken in a naive way. This consideration also shows that the problem can be cured
by subtracting suitable counter terms supported on the light cone.
The kernel~$Q^\sing(x,y)$ should precisely consist of these counter terms, thereby
arranging that~$Q^\text{reg}(x,y)$ is regular and well-defined even in the limit~$\varepsilon \searrow 0$.
In the next proposition we work out~$Q^\sing(x,y)$ more explicitly.

\begin{Prp} \label{prpQsing} There is a kernel~$Q^\sing$ of the form
\[ Q^\sing(x,y) = \frac{1}{\varepsilon^2}\: Q^{(2)}(x,y) + \frac{1}{\varepsilon}\: Q^{(1)}(x,y)
+ \log(\varepsilon) \:Q^{(0)}(x,y) \]
with tempered distributions~$Q^{(0)}$, $Q^{(1)}$ and~$Q^{(2)}$ such that the limit
\[ \lim_{\varepsilon \searrow 0} \big( Q^\varepsilon(x,y) - Q^\sing(x,y) \big) \]
exists in the distributional sense. Moreover, these distributions are supported on the light cone,
\beq \label{supp1}
\supp Q^{(0)}, \:\supp Q^{(1)}, \:\supp Q^{(2)} \subset \{ \xi \:|\: \la \xi, \xi \ra = 0 \} \:,
\eeq
and their Fourier transforms vanish inside the lower mass cone,
\beq \label{supp2}
\supp \hat{Q}^{(0)}, \:\supp \hat{Q}^{(1)}, \:\supp \hat{Q}^{(2)} \subset
\big\{ k \:\big|\: \la k,k \ra \leq 0 \text{ or } k^0 \geq 0 \big\} \:.
\eeq
\end{Prp}
\Proof The kernel~${\mathcal{M}}(x,y)$ was computed in~\cite{reg} away from the light cone
to be of the form (see~\cite[eqs~(2.20) and~(3.5)]{reg})
\[ 
{\mathcal{M}}(x,y) = \left\{ \begin{array}{cl} \xi\slsh\: \epsilon(\xi^0)\:f(\xi^2) &
{\mbox{if $\xi$ is timelike}} \\
0 & {\mbox{if $\xi$ is spacelike}} \end{array} \right. \]
where
\[ 
f(z) = \frac{c_3\,m^3}{z^2} + \frac{c_5\,m^5}{z} \:+\: {\mathcal{O}}(\log z) \:. \]
Multiplying pointwise by~$P(x,y)$ gives singularities on the light cone of the form
\[ Q(x,y) \sim c_3\, m^3\, (\deg=3) + c_3\, m^4\, \slashed{\xi}\, (\deg=3) + (\deg < 3) \]
(where~$\deg$ denotes the degree on the light cone as defined in the formalism
of the continuum in~\cite[\S2.4.4]{cfs}). These singular contributions can also be written as
\beq \label{Qform}
\begin{split}
Q(x,y) &\sim \frac{m^3}{\varepsilon^2\, t^2}\, K_0(\xi)
+ \frac{m^4}{\varepsilon^2\, t^2}\, \slashed{\xi}\, K_0(\xi) + 
\frac{m^5}{\varepsilon t}\, K_0(\xi)
+ \frac{m^6}{\varepsilon t}\, \slashed{\xi}\, K_0(\xi) \\
&\quad\: + m^7\, \log (\varepsilon t) \, K_0(\xi)
+ m^8\, \log (\varepsilon t) \, \slashed{\xi}\, K_0(\xi) + \O(\varepsilon^0) \:,
\end{split}
\eeq
where~$K_0$ is the causal fundamental solution of the scalar wave equation defined by
\[ 
K_0(\xi) := \frac{1}{2 \pi i}\: \big( S^\vee_0- S^\wedge_0 \big)
= \frac{i}{4 \pi^2}\: \frac{1}{2 t}\: \delta(|t|-r) \:. \]
These singular contributions have also been found and discussed in~\cite[Section~5]{action}, and we
refer for computational details to this paper, where the distributions in~\eqref{Qform}
were also computed in momentum space. The resulting formulas show explicitly that,
after choosing the distributional contribution at~$x=y$ appropriately, these distributions
vanish inside the lower mass shell (see the functions plotted in~\cite[Figure~3]{action} and note that
their restrictions to the lower mass shell are polynomials).
Therefore, we can compensate all the singular contributions in~\eqref{Qform} by a
kernel~$Q^\sing$ of the required form.
\QED
Clearly, this procedure leaves the freedom to modify~$Q^\sing(x,y)$ by contributions
having the support properties~\eqref{supp1} and~\eqref{supp2}, but which 
are finite in the sense that they do not depend on~$\varepsilon$.
This freedom will be discussed in Section~\ref{secexmink2}.

We close with two short remarks. We first point out that in the above example of Minkowski space,
the commutator inner product can be expressed in terms of the discontinuity of the derivative of~$\hat{Q}$
on the lower mass shell (for details see~\cite[Section~5.2]{noether}).
Therefore, the support property~\eqref{supp2} implies that the conditions~\eqref{Qsing}
and~\eqref{Qsing2} are both satisfied.

Second, for clarity we remark why it is preferable to work with the regularized kernels.
Before proceeding, we point out that the above divergence is unproblematic
in the EL equations~\eqref{ELQ}, because in these equations, $\hat{Q}$ is evaluated only on the
lower mass shell, where it is finite and well-behaved as~$\varepsilon \searrow 0$.
Therefore, at this stage, subtracting the singular contribution as in~\eqref{ELQreg}
is unnecessary. However, the situation becomes more subtle when the system is perturbed,
for example by introducing an external potential or, more generally, by considering variations
as in Section~\ref{secvary} below. Then both~$\Psi^\fermi$ and~$Q$ are perturbed, in such a way
that the EL equations~\eqref{ELQ} are preserved. This means that the singular contribution to~$Q(x,y)$
is perturbed in a fine-tuned way, such that the image of~$\Psi^\fermi$ remains in the kernel
of~$Q^\sing(x,y)$. As a consequence, the relation~\eqref{Qsing} is preserved by the perturbation.
Again, this causes no problem in the EL equations~\eqref{ELQ}, where the integral remains well-defined.
Thus there there is no necessity to modify the EL equations to~\eqref{ELQ}.
However, the perturbation expansion for~$Q$ involves perturbations of the divergent contributions
on the light cone. Analyzing these contributions in detail is a difficult and laborious task,
because the order of the divergence (the so-called degree on the light cone) is lower than
that of the contributions analyzed in the continuum analysis in~\cite[Chapters~4-6]{cfs}.
Working instead with~$Q^\text{reg}$, these divergent contributions no longer appear,
making it unnecessary to analyze them.

\section{Extending the Hilbert Space in a Surface Layer} \label{secHextend}
\subsection{The Adapted $L^2$-Scalar Product in the Surface Layer}
In preparation of extending the surface layer inner product of Proposition~\ref{prpcomm}
to more general wave functions, we now introduce a Hilbert space of wave functions
endowed with an $L^2$-scalar product involving a measure~$\mu$ which can be thought of
as being supported in the surface layer. To this end, similar as in the construction of the
Krein space structures in~\cite[\S1.1.5]{cfs}, on the spin spaces we introduce the scalar product
\beq \label{spinscalar}
\lla .|. \rra_x \::\: S_x \times S_x \rightarrow \C\:,\qquad 
\lla \psi | \phi \rra_x := \la \psi \,|\, |x|\, \phi \ra_\H
\eeq
(where we make use of the fact that~$S_x \subset \H$) and denote the corresponding norm
on $S_x$ by~$\norm . \norm_x$.
This makes it possible to introduce the norm of the kernel~$Q^\reg(x,y)$ in the usual way by
\[ \norm Q^\reg(x,y) \norm = \sup_{\psi \in S_y \text{ with } \norm \!\psi\! \norm_y=1} \norm Q^\reg(x,y) \psi \norm_x \:. \]
Next, we define the {\em{surface layer measure}} $\mu^\Omega_\rho$ as the Borel measure
given by
\beq \label{muOrdef}
\begin{split}
\mu^\Omega_\rho(U) &:= \int_{U \cap \Omega} d\rho(x)\int_{M \setminus \Omega} d\rho(y) \: \norm Q^\reg(x,y) \norm \\
&\qquad + \int_{U \cap (M \setminus \Omega)} d\rho(x) \int_{\Omega} d\rho(y) \: \norm Q^\reg(x,y) \norm \:.
\end{split}
\eeq
Note that~$x$ is integrated over a subset of~$\Omega$, whereas the $y$-integration
is over a subset of~$M \setminus \Omega$. In this sense, the double integrals have the form of
a surface layer integral. But there is the obvious difference that one of the integrals is ``localized''
to the domain~$U$.

In order to ensure that the integrals in~\eqref{muOrdef} are well-defined, in what follows we always assume
that the integral
\beq \label{Qintegrate}
\int_M \norm Q^\reg(x,y) \norm \:d\rho(y) \qquad \text{is finite for all~$x \in M$
and continuous in~$x$}\:.
\eeq
Under this assumption, on the compactly supported wave
functions~$C^0_0(M, SM)$ we can introduce the scalar product
\[ \lla \psi | \phi \rra^\Omega_\rho := \int_M \lla \psi(x) \,|\, \phi(x) \rra_x \: d\mu^\Omega_\rho(x) \:, \]
referred to as the {\em{adapted $L^2$-scalar product in the surface layer}}.
The corresponding norm is denoted by~$\norm . \norm^\Omega_\rho$.
Forming the completion gives the Hilbert space of wave functions~$(\scrW^\Omega_\rho, \lla \psi | \phi \rra^\Omega_\rho)$.

\subsection{Extending the Surface Layer Integral}
The next step in our construction is to extend the commutator inner product
in Proposition~\ref{prpcomm} to more general wave functions.
In this section, we allow for a general class of wave functions for which the surface layer integral is well-defined.
The resulting space is too large for the applications. In order to obtain a space which
can be thought of as being a generalization of the Hilbert space of all Dirac solutions
(also including all the solutions of positive energy), we want to extend~$\H$ only by
those physical wave functions obtained when the physical system is varied while preserving
the Euler-Lagrange equations. This idea will be implemented in Sections~\ref{secvary}--\ref{secextend},
giving rise to the Hilbert space~$\H^{\fermi, \Omega}_\rho$.

In order to extend the commutator inner product to more general wave functions,
on the Hilbert space~$(\scrW^\Omega_\rho, \lla \psi | \phi \rra^\Omega_\rho)$ we introduce the sesquilinear form
\begin{align}
\la \psi | \phi \ra^\Omega_\rho &\::\: \scrW^\Omega_\rho \times \scrW^\Omega_\rho \rightarrow \C \label{Vtinner} \\
\la \psi | \phi \ra^\Omega_\rho &:= -2i \,\bigg( \int_{\Omega} \!d\rho(x) \int_{M \setminus \Omega} \!\!\!\!\!\!\!d\rho(y) 
- \int_{M \setminus \Omega} \!\!\!\!\!\!\!d\rho(x) \int_{\Omega} \!d\rho(y) \bigg)\:
\Sl \psi(x) \:|\: Q^\reg(x,y)\, \phi(y) \Sr_x \:. \notag
\end{align}
In the next lemma it is shown that this sesquilinear form is well-defined.

\begin{Lemma} For any~$\psi, \phi \in \scrW^\Omega_\rho$, the integrals in~\eqref{Vtinner}
are well-defined and
\[ |\la \psi | \phi \ra^\Omega_\rho| \leq 2\, \norm \!\psi\! \norm^\Omega_\rho\: \norm \!\phi\! \norm^\Omega_\rho \:. \]
\end{Lemma}
\Proof Using a standard denseness argument, it suffices to consider compactly
supported wave functions~$\psi, \phi \in C^0_0(M, SM)$. Then the integral can be estimated by
\begin{align*}
\big| &\la \psi | \phi \ra^\Omega_\rho \big| \\
&\leq 2 \,\bigg( \int_{\Omega} \!d\rho(x) \int_{M \setminus \Omega} \!\!\!\!\!\!\!d\rho(y) 
+ \int_{M \setminus \Omega} \!\!\!\!\!\!\!d\rho(x) \int_{\Omega} \!d\rho(y) \!\bigg)
\norm \!\psi(x)\! \norm_x \, \norm Q^\reg(x,y) \norm\, \norm \!\phi(y)\! \norm_y \:.
\end{align*}
Employing the inequality
\[ \norm \!\psi(x)\! \norm_x \:\norm \!\phi(y)\! \norm_y
\leq \frac{1}{2} \: \Big( \norm \!\psi(x)\! \norm_x^2 + \norm \!\phi(y)\! \norm_y^2 \Big) \]
and using the definition of the surface layer measure~\eqref{muOrdef}, we obtain
\[ \big|\la \psi | \phi \ra^\Omega_\rho \big| \leq \big( \norm \!\psi\! \norm^\Omega_\rho \big)^2 +
\big( \norm \!\phi\! \norm^\Omega_\rho \big)^2 \:. \]
Using the freedom to scale the vectors~$\psi$ and~$\phi$ gives the result.
\QED

The last lemma also shows that the sesquilinear form~$\la .|. \ra^\Omega_\rho$ is bounded
in the Hilbert space~$(\scrW^\Omega_\rho, \lla \psi | \phi \rra^\Omega_\rho)$.
Therefore, applying the Fr{\'e}chet-Riesz theorem, we obtain a 
bounded symmetric operator $\Sig \in \Lin(\scrW^\Omega_\rho)$ with the property that
\beq \label{Sigdef}
\la \psi | \phi \ra^\Omega_\rho = \lla \psi \,|\, \Sig \phi \rra^\Omega_\rho \:.
\eeq
For technical simplicity, we shall assume that the operator~$\Sig$ is injective (i.e.\ that it has
a trivial kernel). Clearly, this is equivalent to assuming that the sesquilinear form~$\la .|. \ra^\Omega_\rho$
is {\em{non-degenerate}} on~$\scrW^\Omega_\rho$.
We remark that all our constructions could be extended
to the case with degeneracies by restricting the operators and sesquilinear
forms to the kernel's orthogonal complement. However, we shall not enter these generalizations here to avoid lengthening the exposition. Thus assuming non-degeneracy, we may use the spectral calculus to form other (in general unbounded)
selfadjoint operators which come with a corresponding dense domain.
More precisely, the spectral theorem for bounded 
selfadjoint operators allows us to represent~$\Sig$ as
\[ \Sig = \int_{\sigma(\Sig)} \lambda\: dE_\lambda \]
with a spectral measure~$E$ on the Borel algebra of the compact set~$\sigma(\Sig) \subset \R$.
Given a real-valued Borel function
\[ 
g \::\: \sigma(\Sig) \setminus \{0\} \rightarrow \R \]
(note that this function is finite a.e.\ with respect to $E$ because~$E_{\{0\}}=0$), the operator~$g(\Sig)$ defined by
\beq \label{gspectral}
\begin{split}
g(\Sig) &= \int_{\sigma(\Sig)} g(\lambda)\: dE_\lambda \::\: \D\big( g(\Sig) \big) \subset \scrW^\Omega_\rho \rightarrow \scrW^\Omega_\rho \qquad \text{with} \\
\D\big(g(\Sig) \big) &= \Big\{ \psi \in \subset \scrW^\Omega_\rho \:\Big|\:
\int_{\sigma(\Sig)} \big|g(\lambda) \big|^2 \: d\lla \psi | E_\lambda \psi \rra^\Omega_\rho < \infty \Big\}
\end{split}
\eeq
is densely defined and selfadjoint (see for example~\cite[Theorem VIII.6]{reed+simon}). In particular,
we can introduce the inverse of~$\Sig$ as a selfadjoint operator
\begin{align*}
\Sig^{-1} &= \int_{\sigma(\Sig)} \lambda^{-1}\: dE_\lambda \::\: \D\big( \Sig^{-1} \big) \subset \scrW^\Omega_\rho \rightarrow \scrW^\Omega_\rho \qquad \text{with} \\
\D\big(\Sig^{-1} \big) &= \Big\{ \psi \in \subset \scrW^\Omega_\rho \:\Big|\:
\int_{\sigma(\Sig)} |\lambda|^{-2} \: d\la \psi | E_\lambda \psi \ra^\Omega_\rho < \infty \Big\} \:.
\end{align*}

\subsection{Varying the Surface Layer Integral} \label{secvary}
In the second step of the construction, we want to describe the situation in a scattering process
where particle/anti-particle pairs are created starting from the fermionic vacuum
(see the left of Figure~\ref{figscatter}). 
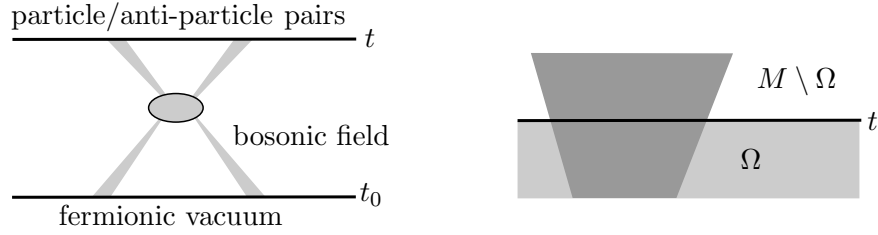
\begin{figure}[tb]
\psscalebox{1.0 1.0} 
{
\begin{pspicture}(-1.5,-1.4659961)(14.455,1.4659961)
\definecolor{colour0}{rgb}{0.8,0.8,0.8}
\definecolor{colour1}{rgb}{0.6,0.6,0.6}
\psframe[linecolor=colour0, linewidth=0.02, fillstyle=solid,fillcolor=colour0, dimen=outer](11.277223,-0.0039474824)(6.732222,-1.0339475)
\pspolygon[linecolor=colour1, linewidth=0.02, fillstyle=solid,fillcolor=colour1](6.922222,0.89105254)(7.4722223,-1.0189475)(8.837222,-1.0189475)(9.582222,0.8860525)
\pspolygon[linecolor=colour0, linewidth=0.02, fillstyle=solid,fillcolor=colour0](3.202222,1.0860525)(2.2522223,0.09605252)(2.9772222,1.0810525)
\pspolygon[linecolor=colour0, linewidth=0.02, fillstyle=solid,fillcolor=colour0](1.2922223,1.0710526)(2.2422223,0.08105252)(1.5172222,1.0660526)
\pspolygon[linecolor=colour0, linewidth=0.02, fillstyle=solid,fillcolor=colour0](3.3722222,-1.0189475)(2.3872223,0.12605251)(3.1472223,-1.0139475)
\pspolygon[linecolor=colour0, linewidth=0.02, fillstyle=solid,fillcolor=colour0](1.0822222,-1.0189475)(2.067222,0.12605251)(1.3072222,-1.0139475)
\psellipse[linecolor=black, linewidth=0.02, fillstyle=solid,fillcolor=colour0, dimen=outer](2.1897223,0.17105252)(0.3725,0.2)
\psline[linecolor=black, linewidth=0.04](0.017222222,1.0810525)(4.587222,1.0860525)
\rput[bl](0.635,-1.3900586){fermionic vacuum}
\rput[bl](0.0,1.1949414){particle/anti-particle pairs}
\psline[linecolor=black, linewidth=0.04](0.012222222,-1.0189475)(4.582222,-1.0139475)
\rput[bl](4.7,0.97494143){$t$}
\rput[bl](4.635,-1.1300586){$t_0$}
\rput[bl](2.945,-0.3300586){bosonic field}
\psline[linecolor=black, linewidth=0.04](6.732222,0.0010525173)(11.302222,0.0060525173)
\rput[bl](11.385,-0.12505859){$t$}
\rput[bl](9.695,-0.6450586){$\Omega$}
\rput[bl](9.915,0.33494142){$M \setminus \Omega$}
\end{pspicture}
}
\caption{A scattering process (left) and a description in terms of linearized solutions (right).}
\label{figscatter}
\end{figure}%
We want to build up the extended Hilbert space~$\H^{\fermi,\Omega}_\rho$
by all physical wave functions which can be generated with this procedure by considering all
possible scattering processes. In this scenario, the pair creation is triggered by bosonic fields
which typically are also present at initial and final times.
However, at least at initial time, the bosonic field should be so weak and/or so spread out that
it has no effect on the commutator inner product. Then, due to the conservation law, the
commutator inner product at time~$t$ again represents the scalar product (see Definition~\ref{defSLrep}).
But we need to take into account that both the physical wave functions and the kernel~$Q^\reg(x,y)$ in Proposition~\ref{prpcomm} will in general change.

This situation can be modelled in a simpler way without referring to an initial state as follows
(see the right of Figure~\ref{figscatter}).
We consider variations~$(\tilde{\rho}_\tau)_{\tau \in [0, \delta]}$
of the measure~$\rho$. For technical simplicity, as in~\cite[Section~2.3]{fockbosonic} 
and in the derivation of the linearized field equations (see Section~\ref{seclinfield}),
we assume that the variation can be written in the form~\eqref{rhotau} with smooth mapping~$f_\tau$
and~$F_\tau$~\eqref{fFsmooth} which depend smoothly on~$\tau$. 
Moreover, we assume that each~$F_\tau$ is injective, closed, and that the inverse on its image
is continuous, implying that
\beq \label{Fbijective}
F_\tau \::\: M \rightarrow \tilde{M}_\tau := \supp \tilde{\rho}_\tau \quad \text{is a homeomorphism}
\eeq
for every~$\tau \in [0, \delta]$.
Next, we assume that the EL equations hold for all~$\tau$. This implies that the jets tangential to the curve
satisfy the linearized field equations, i.e.\
\[ \v_\tau := \frac{d}{d\tau} \big( f_\tau, F_\tau \big) \in \Jlin_{\tilde{\rho}_\tau} \qquad
\text{for all~$\tau \in [0, \delta]$}\:. \]
Again for technical simplicity, we assume that~$\v_\tau$ has spatially compact support, $\v_\tau \in \Jlin_{\tilde{\rho}_\tau, \sc}$.
Moreover, we assume that the commutator inner product does not depend on~$\tau$, i.e.
\beq \label{conscond}
\la u | v \ra^\Omega_{\tilde{\rho}_\tau} = \la u | v \ra^\Omega_\rho \qquad \text{for all~$u,v \in \H^\fermi$ and~$\tau \in [0, \delta]$}\:.
\eeq
This condition is satisfied for a scattering process as discussed on the left of Figure~\ref{figscatter}.

More mathematically, the condition~\ref{conscond} can be understood as follows.
\begin{Lemma} \label{lemmagpreserve}
A variation~$(\tilde{\rho}_\tau)_{\tau \in [0, \delta]}$ of the form~\eqref{rhotau}
with tangential jets~$\v_\tau \in \Jlin_{\tilde{\rho}_\tau, \sc}$ preserves the commutator inner product~\eqref{conscond}
if and only if for every~$u \in \H$, the corresponding commutator jet~$\Comm$
given by~\eqref{jvdef} with~$\scrA$ according to~\eqref{Adef} satisfies the condition
\beq \label{sigmapreserve}
\sigma^\Omega_{\tilde{\rho}_\tau}\big( \Comm, \v_\tau \big) = 0 \qquad \text{for all~$\tau \in [0,\delta]$}\:.
\eeq
\end{Lemma}
\Proof In view of the polarization identity, in~\eqref{conscond} it suffices to consider the case~$v=u$
for any given~$u \in \H$. Then, by the fundamental theorem of calculus, \eqref{conscond} is equivalent to the condition
\[ \label{conscond2}
\frac{d}{d\tau} \la u | u \ra^\Omega_{\tilde{\rho}_\tau} = 0 \qquad \text{for all~$\tau \in [0, \delta]$}\:. \]

Setting $\scrU_s = e^{i s \scrA}$, the commutator jet can be written as
\begin{align*}
\la u | u \ra^\Omega_{\tilde{\rho}_\tau} &= \int_{F_\tau(\Omega)} d\tilde{\rho}_\tau(x)
\int_{F_\tau(M \setminus \Omega)} d\tilde{\rho}_\tau(y)\: \big(D_{1,\Comm} - D_{2,\Comm} \big) \L_\kappa(x,y) \\
&= \frac{d}{ds} \int_{F_\tau(\Omega)} d\tilde{\rho}_\tau(x)
\int_{F_\tau(M \setminus \Omega)} d\tilde{\rho}_\tau(y)\: \L_\kappa\big( \scrU_s x \scrU_s^{-1},
\scrU_s^{-1} y \scrU_s \big) \bigg|_{s=0} \:.
\end{align*}
Using the definition of the push-forward measure, we obtain
\beq \label{uufF}
\begin{split}
&\la u | u \ra^\Omega_{\tilde{\rho}_\tau} \\
&= \frac{d}{ds} \int_{\Omega} d\rho(x)
\int_{M \setminus \Omega} d\rho(y)\: f_\tau(x)\: f_\tau(y)\: \L_\kappa\big( \scrU_s \,F_\tau(x)\, \scrU_s^{-1},
\scrU_s^{-1} \,F_\tau(y)\, \scrU_s \big) \bigg|_{s=0} \:.
\end{split}
\eeq
Moreover, the unitary invariance of the Lagrangian implies that for all~$\hat{x}, \hat{y} \in \F$,
\[ 
\L_\kappa(\hat{x},\hat{y}) = \L_\kappa\big( \scrU_s \hat{x} \scrU_s^{-1}, \scrU_s \hat{y} \scrU_s^{-1} \big) \:. \]
Choosing~$\hat{x}=F_\tau(x)$ and~$\hat{y}=F_{2 \tau_0-\tau}(y)$ (for~$\tau_0 \in [0, \delta]$),
multiplying by the functions~$f_\tau(x)$ and~$f_{2 \tau_0-\tau}(y)$
and subtracting the resulting expression from the integrand in~\eqref{uufF}, we conclude that
\begin{align*}
\la u | u \ra^\Omega_{\tilde{\rho}_\tau} &= \frac{d}{ds} \int_{\Omega} d\rho(x)
\int_{M \setminus \Omega} d\rho(y)\:  \\
&\qquad \times
\Big( f_\tau(x)\: f_\tau(y)\: \L_\kappa\big( \scrU_s \,F_\tau(x)\, \scrU_s^{-1}, \scrU_s^{-1} \,F_\tau(y)\, \scrU_s \big) \\
&\qquad \quad - f_\tau(x)\: f_{2 \tau_0-\tau}(y)\: \L_\kappa\big( \scrU_s \,F_\tau(x)\, \scrU_s^{-1}, \scrU_s \,F_{2 \tau_0-\tau}(y)\, \scrU_s^{-1} \big) \Big) \bigg|_{s=0} \:.
\end{align*}
We now take the $\tau$-derivative at~$\tau=\tau_0$. Using that
\[ \frac{d}{d\tau} \big( f_{2 \tau_0-\tau}(y), F_{2 \tau_0-\tau}(y) \big|_{\tau=\tau_0} \big) = - \frac{d}{d\tau} 
\big(f_\tau, F_\tau(y) \big) \:, \]
the terms cancel whenever derivatives act either both at~$x$ or both at~$y$.
A straightforward computation yields
\begin{align*}
&\frac{d}{d\tau} \la u | u \ra^\Omega_{\tilde{\rho}_\tau} \Big|_{\tau=\tau_0} \\
&= \int_{\Omega} d\rho(x) \int_{M \setminus \Omega} \!\!\!\!\!\!d\rho(y)\:
f_{\tau_0}(x)\: f_{\tau_0}(y) \: \big(D_{1,\Comm} \nabla_{2,\v} - D_{2,\Comm} \nabla_{1,\v} \big)
\L_\kappa\big(F_{\tau_0}(x),F_{\tau_0}(y) \big) \\
&=\sigma^\Omega_{\tilde{\rho}_{\tau_0}}\big( \v_{\tau_0}, \Comm \big) \:.
\end{align*}
This concludes the proof.
\QED

\subsection{Transformation of the Varied Surface Layer Integral} \label{secositransform}
In view of our assumption~\eqref{conscond}, the commutator inner product is preserved in~$\tau$.
However, when~$\tau$ is varied, both the physical wave functions and the surface layer inner product
change (note that the kernel~$Q^\reg(x,y)$ and the measure~$\rho$ in~\eqref{Vtinner}
also depend on~$\tau$). In order to work for all~$\tau$ in the same inner product space~$\scrW^\Omega_\rho$,
we need to transform the surface layer inner product for~$\tau \neq 0$ back to the surface layer
inner product for~$\tau=0$. To this end, given a measure~$\tilde{\rho}$,
in this section we shall construct a densely defined
isometry between the indefinite inner product spaces, i.e.\
\[ {\mathscr{I}}^\Omega_{\rho, \tilde{\rho}} : \D\big( {\mathscr{I}}^\Omega_{\rho, \tilde{\rho}} \big) \subset \scrW^\Omega_{\tilde{\rho}} \rightarrow \scrW^\Omega_\rho \qquad \text{with} \qquad
\la \psi | \phi \ra^\Omega_{\tilde{\rho}} = \la {\mathscr{I}}^\Omega_{\rho, \tilde{\rho}}\, \psi \,|\, {\mathscr{I}}^\Omega_{\rho, \tilde{\rho}}\, \phi \ra^\Omega_\rho
\quad \forall\;\psi, \phi \in \D( {\mathscr{I}}^\Omega_{\rho, \tilde{\rho}} ) \:. \]
This mapping exists only under additional assumptions which we 
shall work out on the way. We collect all these assumptions in a condition for~$\tilde{\rho}$ (see Definition~\ref{defadmissible}).

A basic difficulty is that we need to relate wave functions
in the varied spacetime~$\tilde{M}$ to wave functions in the original
spacetime~$M$, making it necessary to identify the corresponding spin spaces
$S_x \leftrightarrow S_{F(x)}$. One way of making this identification canonical
is to fix the gauge as described in~\cite{gaugefix}. For our purposes, it seems sufficient to
use the somewhat simpler method of working with the orthogonal projection $S_x \rightarrow S_{F(x)}$
in~$\H$. Here we make essential use of the fact that the spin spaces are subspaces of~$\H$,
an observation which also lies at the heart of the gauge-fixing procedure in~\cite{gaugefix}.
In order for the construction to fit together with the norm on the respective spin spaces~\eqref{spinnorm}, we also
insert suitable factors of the operators
\[ |x| \big|_{S_x}^{\frac{1}{2}},\; |x| \big|_{S_x}^{-\frac{1}{2}} : S_x \rightarrow S_x \:. \]
This leads us to introduce the mapping
\beq \label{pirhodef}
\begin{split}
&\pi_{\rho, \tilde{\rho}} \::\: C^0_0(\tilde{M}, S\tilde{M}) \rightarrow C^0_0(M, SM) \:, \\
&\big(\pi_{\rho, \tilde{\rho}}\, \psi \big)(x) := |x| \big|_{S_x}^{-\frac{1}{2}} \,\pi_x
\, |F(x)| \big|_{S_{F(x)}}^{\frac{1}{2}} \, \psi \big(F(x) \big) \:.
\end{split}
\eeq
This mapping also gives rise to the positive definite sesquilinear
form~$\lla \pi_{\rho, \tilde{\rho}} \,.\, | \,\pi_{\rho, \tilde{\rho}}
\,.\, \rra^\Omega_\rho$ on~$C^0_0(\tilde{M}, S\tilde{M})$. 
In the next lemma we express this sesquilinear form with respect to the $L^2$-scalar product
in the surface layer.

\begin{Lemma} Let~$\D(\tilde{\T})$ be the space of all
bounded Borel wave functions~$\phi$ on~$\tilde{M}$ which are compactly supported
in the interior of the support of the surface layer measure~\eqref{muOrdef},
\beq \label{suppphi}
\supp \phi \Subset \text{\rm{int}}  \supp \mu^\Omega_{\tilde{\rho}} \:.
\eeq
Then there is a unique linear operator
\[ \tilde{\T} \::\: \D(\tilde{\T}) \rightarrow \scrW^\Omega_{\tilde{\rho}} \]
with the property that
\beq \label{Tdefine}
\lla \pi_{\rho, \tilde{\rho}} \,\psi \, | \,\pi_{\rho, \tilde{\rho}} \,\phi \rra^\Omega_\rho
= \lla \psi \,|\, \tilde{\T} \,\phi \rra^\Omega_{\tilde{\rho}}
\qquad \text{for all~$\psi \in \scrW^\Omega_{\tilde{\rho}}$
and~$\phi \in \D(\tilde{\T})$}\:.
\eeq
This operator is densely defined, symmetric and positive semi-definite.
\end{Lemma}
\Proof Let~$\phi \in \D(\tilde{\T})$. Consider the linear functional
\beq \label{functional}
\psi \mapsto \lla \pi_{\rho, \tilde{\rho}} \,\phi \, | \,\pi_{\rho, \tilde{\rho}}\, \psi \rra^\Omega_\rho \:.
\eeq
Let us verify that this functional is bounded. To this end, we begin with the estimate
\begin{align}
&\big| \lla \pi_{\rho, \tilde{\rho}} \,\phi \, | \,\pi_{\rho, \tilde{\rho}} \,\psi \rra^\Omega_\rho \big| \leq
\int_{M} \norm ( \pi_{\rho, \tilde{\rho}} \phi)(x) \norm_x\:
\norm ( \pi_{\rho, \tilde{\rho}} \psi)(x) \norm_x\: d\mu^\Omega_{\rho}(x) \notag \\
&= \int_{M} \Big\| \pi_x \,|F(x)| \big|_{S_{F(x)}}^{\frac{1}{2}}\,\phi \big(F(x)\big) \Big\|_\H\:
\Big\| \pi_x \,|F(x)| \big|_{S_{F(x)}}^{\frac{1}{2}}\,\psi \big(F(x)\big) \Big\|_\H
\: d\mu^\Omega_{\rho}(x) \notag \\
&\leq \int_{M} \Big\| |F(x)| \big|_{S_{F(x)}}^{\frac{1}{2}}\,\phi \big(F(x)\big) \Big\|_\H\:
\Big\| |F(x)| \big|_{S_{F(x)}}^{\frac{1}{2}}\,\psi \big(F(x)\big) \Big\|_\H
\: d\mu^\Omega_{\rho}(x) \notag \\
&= \int_{M} \norm \phi \big(F(x)\big) \norm_{F(x)}\:
\norm \psi \big(F(x)\big) \norm_{F(x)}\: d\mu^\Omega_{\rho}(x) \notag \\
&= \int_{\tilde{M}} \norm \phi(x)\norm_{x}\:
\norm \psi(x) \norm_{x}\: d\big(F_* \mu^\Omega_{\rho}\big)(x) \:. \label{step1}
\end{align}
In order to relate this integral to the norm~$\norm . \norm^\Omega_{\tilde{\rho}}$, we need to
compare the integration measure~$F_* \mu^\Omega_{\rho}$ with~$\mu^\Omega_{\tilde{\rho}}$.
Using the definition~\eqref{muOrdef}, we can write the measures as
\[ d(F_* \mu^\Omega_\rho)(x) = \mathfrak{n}(x)\: d\tilde{\rho}(x) \qquad \text{and} \qquad
d\mu^\Omega_{\tilde{\rho}}(x) = \tilde{\mathfrak{n}}(x)\: d\tilde{\rho}(x) \:, \]
where the functions~$\mathfrak{n}$ and~$\tilde{\mathfrak{n}}$ are given by
\begin{align*}
\mathfrak{n}(x) &= \chi_{\tilde{\Omega}}(x) \int_{M \setminus \Omega} \norm Q^\reg\big(F^{-1}(x),y \big) \norm \:d\rho(y)
+  \chi_{\tilde{M} \setminus \tilde{\Omega}}(x) \int_{\Omega} \norm Q^\reg\big(F^{-1}(x),y \big) \norm\:d\rho(y) \\
\tilde{\mathfrak{n}}(x) &= \chi_{\tilde{\Omega}}(x) \int_{\tilde{M} \setminus \tilde{\Omega}} \norm Q^\reg(x,y) \norm\:d\tilde{\rho}(y)
+  \chi_{\tilde{M} \setminus \tilde{\Omega}}(x) \int_{\tilde{\Omega}} \norm Q^\reg(x,y) \norm\:d\tilde{\rho}(y) \:.
\end{align*}
(here we used that~$F$ is bijective~\eqref{Fbijective}).
Using~\eqref{Qintegrate} and again~\eqref{Fbijective}, the function~$\mathfrak{n}$ is continuous.
Therefore, denoting the support of~$\phi$ by~$K \Subset \tilde{M}$, this function is bounded on~$K$. The function~$\tilde{\mathfrak{n}}$, on the other hand, is also continuous and strictly positive on~$K$
in view of~\eqref{suppphi}. Therefore, there is a constant~$c=c(\phi)>0$ such that
\[ \mathfrak{n}(x) \leq c(\phi)\: \tilde{\mathfrak{n}}(x) \qquad \text{for all~$x \in K$} \:. \]
Using this inequality in~\eqref{step1}, we conclude that
\beq \label{boundfunc}
\big| \lla \pi_{\rho, \tilde{\rho}} \,\phi \, | \,\pi_{\rho, \tilde{\rho}} \,\psi \rra^\Omega_\rho \big|
\leq c(\phi) \:\int_{K} \norm \phi(x)\norm_{x}\:
\norm \psi(x) \norm_{x}\: d\mu^\Omega_{\tilde{\rho}}(x)
\leq C(\phi) \:\|\psi\|^\Omega_{\tilde{\rho}}
\eeq
for a new constant~$C(\phi)>0$, where in the last step we used the Schwarz inequality.
The inequality~\eqref{boundfunc} shows that the functional~\eqref{functional} is indeed bounded.

Applying the Fr{\'e}chet-Riesz theorem, the functional~\eqref{functional} can be represented
unique\-ly by a vector~$\chi \in \scrW^\Omega_{\tilde{\rho}}$, i.e.
\[ \lla \pi_{\rho, \tilde{\rho}}\, \phi \, | \,\pi_{\rho, \tilde{\rho}}\, \psi \rra^\Omega_\rho
= \lla \chi | \psi \rra^\Omega_{\tilde{\rho}} \qquad \text{for all~$\psi \in \scrW^\Omega_{\tilde{\rho}}$}\:. \]
Setting~$\tilde{\T}(\phi)=\chi$ uniquely defines an operator with the property~\eqref{Tdefine}.

The denseness of~$\D(\tilde{\T})$ follows immediately by approximating
a wave function~$\psi \in C^0_0(\tilde{M}, S\tilde{M})$ by multiplication with characteristic functions,
\[ \psi_k(x) := \chi_{(\frac{1}{k}, \infty)}\big( \tilde{\mathfrak{n}}(x) \big)\: \psi(x) \:, \]
and taking the limit $k \rightarrow \infty$ with the help of Lebesgue's dominated convergence theorem.
The fact that the left side of~\eqref{Tdefine} is real and non-negative
for any~$\psi =\phi \in \D(\tilde{\T})$
yields that~$\tilde{\T}$ is symmetric and positive semi-definite.
\QED
Taking the Friedrichs extension for semi-bounded operators (see for example~\cite[Section~33.3]{lax}), we obtain a
\beq \label{hatTdef}
\text{selfadjoint operator} \qquad\hT \::\: \D\big( \hT \big) \subset 
\scrW^\Omega_{\tilde{\rho}} \rightarrow \scrW^\Omega_{\tilde{\rho}} \:.
\eeq
We assume that this operator is {\em{injective}}. Then, similar as explained for the operator~$\Sig$
after~\eqref{gspectral}, denoting the spectral measure of~$\hT$ by~$\hat{E}$,
given any real-valued Borel function~$g$ on~$\sigma(\hT) \setminus \{0\}$ we
can again apply the spectral theorem to define the selfadjoint operator
\begin{align*}
g(\hT) &= \int_{\sigma(\hT)} g(\lambda)\: d\hat{E}_\lambda \::\: 
D\big( g(\hT) \big) \subset \scrW^\Omega_{\tilde{\rho}} \rightarrow \scrW^\Omega_{\tilde{\rho}} \qquad \text{with}
\\
\D\big( g(\hT) \big) &= \Big\{ \psi \in \scrW^\Omega_{\tilde{\rho}} \:\Big|\:
\int_{\sigma(\Sig)} |g(\lambda)|^2 \: d\lla \psi | \hat{E}_\lambda \psi \rra^\Omega_{\tilde{\rho}} < \infty \Big\} \:.
\end{align*}

Finally, we extend the domain of the operator~$\pi_{\rho, \tilde{\rho}}$ and define its inverse.
In preparation, we note that the
operator~$U_{\rho, \tilde{\rho}} := \pi_{\rho, \tilde{\rho}} \, \hT^{-\frac{1}{2}}$ is an isometry because
\[ \lla U_{\rho, \tilde{\rho}} \,\psi \,|\, U_{\rho, \tilde{\rho}} \,\phi \rra^\Omega_\rho = 
\lla \pi_{\rho, \tilde{\rho}} \, \hT^{-\frac{1}{2}} \,\psi | \pi_{\rho, \tilde{\rho}} \, \hT^{-\frac{1}{2}}\, \phi \rra^\Omega_\rho = \lla \hT^{-\frac{1}{2}} \,\psi | \, \hT \,\hT^{-\frac{1}{2}}\, \phi \rra^\Omega_{\tilde{\rho}} = \lla \psi | \phi \rra^\Omega_{\tilde{\rho}} \:. \]
Hence this operator can be extended continuously to all of~$\scrW^\Omega_{\tilde{\rho}}$.
This makes it possible to introduce the extension
\begin{align}
\pi_{\rho, \tilde{\rho}} = U_{\rho, \tilde{\rho}}\, \hT^{\frac{1}{2}} \::\: \D\big( \pi_{\rho, \tilde{\rho}} \big) \subset
\scrW^\Omega_{\tilde{\rho}} \rightarrow \scrW^\Omega_\rho \:,&&
\D\big( \pi_{\rho, \tilde{\rho}} \big) &:= \D\big( \hT^{\frac{1}{2}} \big) \:.
\label{piext}
\intertext{For technical simplicity, we assume that this operator is {\em{surjective}}. Then~$U_{\rho, \tilde{\rho}}$ is a unitary
operator, making it possible to introduce the operator}
\pi_{\rho, \tilde{\rho}}^{-1} = \hT^{-\frac{1}{2}}\, U_{\rho, \tilde{\rho}}^{-1}\::\: \D\big( \pi_{\rho, \tilde{\rho}}^{-1} \big) \subset \scrW^\Omega_\rho \rightarrow \scrW^\Omega_{\tilde{\rho}} \:,&&
\D\big( \pi_{\rho, \tilde{\rho}}^{-1} \big) &:= U_{\rho, \tilde{\rho}}\, \D\big( \hT^{-\frac{1}{2}} \big)\:. \notag
\end{align}

We now come to the main construction. For simplicity of presentation, we begin with the following formal computation,
\begin{align}
\la \psi \,|\, \phi \ra^\Omega_{\tilde{\rho}} &=
\lla \psi \,|\, \tilde{\Sig} \,\phi \rra^\Omega_{\tilde{\rho}}
= \lla \psi \,|\, \hT\:\hT^{-1}\, \tilde{\Sig} \,\phi \rra^\Omega_{\tilde{\rho}}
= \lla \pi_{\rho, \tilde{\rho}} \, \psi \,|\, \pi_{\rho, \tilde{\rho}} \,\hT^{-1}\, \tilde{\Sig}\,\phi \rra^\Omega_{\rho} \notag \\
&= \lla \pi_{\rho, \tilde{\rho}} \, \psi \,|\, \Sig\, \Sig^{-1}\,\pi_{\rho, \tilde{\rho}} \,\hT^{-1}\, \tilde{\Sig}\,\phi \rra^\Omega_{\rho} 
= \la \pi_{\rho, \tilde{\rho}} \,\psi \,|\, \Sig^{-1}\,\pi_{\rho, \tilde{\rho}} \,\hT^{-1}\, \tilde{\Sig}\,
\phi \ra^\Omega_{\rho} \notag \\
&= \la \pi_{\rho, \tilde{\rho}} \,\psi \,|\, B\,\pi_{\rho, \tilde{\rho}} \, \phi \ra^\Omega_{\rho} \:, \label{kreinrep}
\end{align}
where we introduced the abbreviation
\beq \label{Bdef}
B := \Sig^{-1}\,\pi_{\rho, \tilde{\rho}} \,\hT^{-1}\, \tilde{\Sig}\, \pi_{\rho, \tilde{\rho}}^{-1} \:.
\eeq
In order to give this computation a proper mathematical meaning, we must make sure that
when taking products of operators, the range of each operator is contained in the domain of the factor to its left.
We summarize all the resulting conditions in the following definition.

\begin{Def} \label{defadmissible} The measure~$\tilde{\rho}$
is {\bf{variation-admissible}} if the following conditions hold:
\bitem
\item[\rm{(i)}] There is a one-parameter family of measures~$(\tilde{\rho}_\tau)_{\tau \in [0,\delta]}$
of the form~\eqref{rhotau} with~$\tilde{\rho}_0=\rho$ and~$\tilde{\rho}_\delta=\tilde{\rho}$
which satisfies the EL equations for all~$\tau$ and preserves the commutator inner product~\eqref{conscond}.
\item[\rm{(ii)}] The operators~$\tilde{\Sig}$, $\hT$ (see~\eqref{Sigdef} and~\eqref{hatTdef})
are injective, and the operator~$\pi_{\rho, \tilde{\rho}}$ (see~\eqref{piext}) is surjective.
\item[\rm{(iii)}] The operator product in~\eqref{Bdef} is well-defined in the sense that
\[ \D \big(\hT^{-1} \big) \subset \text{\rm{Rg}} \big( \tilde{\Sig} \big)
\qquad \text{and} \qquad
\D \big( \Sig^{-1} \big) \subset \text{\rm{Rg}} \big( \pi_{\rho, \tilde{\rho}} \,\hT^{-1} \,\tilde{\Sig} \big) \:. \]
\item[\rm{(iv)}] The resulting operator~$B$ in~\eqref{Bdef} is bounded.
The spectrum of its extension~$B \in \Lin(\scrW^\Omega_{\rho})$ does not intersect the
negative real axis,
\[ \sigma(B) \cap \R^-_0 = \varnothing \:. \]
\eitem
\end{Def}
The last assumption is needed for the next step of the construction, where we choose a
closed contour~$\Gamma$ which does not intersect the negative real axis and encloses the spectrum of~$B$ 
with positive orientation and introduce the square root of~$B$ as the contour integral
\[ \sqrt{B} := -\frac{1}{2 \pi i} \ointctrclockwise_\Gamma \sqrt{z}\: \big(B-z)^{-1}\: dz \]
(with the sign of the square root chosen such that $\re \sqrt{z} >0$).
Note that the operator~$B$ is symmetric with respect to the inner product~$\la .|. \ra^\Omega_\rho$ because
\begin{align*}
\la \pi_{\rho, \tilde{\rho}} \,\psi \,|\, B\, \pi_{\rho, \tilde{\rho}} \,\phi \ra^\Omega_\rho &=  
\lla \pi_{\rho, \tilde{\rho}} \,\psi \,|\, \Sig B\, \pi_{\rho, \tilde{\rho}} \,\phi \rra^\Omega_\rho = 
\lla \pi_{\rho, \tilde{\rho}} \,\psi \,|\, \pi_{\rho, \tilde{\rho}} \,\hT^{-1}\, \tilde{\Sig}\, \phi \rra^\Omega_\rho = \lla \psi \,|\, \tilde{\Sig}\, \phi \rra^\Omega_\rho \\
&= \lla \tilde{\Sig}\,\psi \,|\, \phi \rra^\Omega_\rho
= \cdots = \la B\,\pi_{\rho, \tilde{\rho}} \,\psi \,|\, \pi_{\rho, \tilde{\rho}} \,\phi \ra^\Omega_\rho \:.
\end{align*}
As a consequence,
\begin{align*}
&\la \pi_{\rho, \tilde{\rho}}\, \psi \,|\, \sqrt{B}\, \pi_{\rho, \tilde{\rho}}\, \phi \ra^\Omega_\rho
= -\frac{1}{2 \pi i} \ointctrclockwise_\Gamma \sqrt{z}\; \la \pi_{\rho, \tilde{\rho}}\, \psi \,|\, \big(B-z)^{-1} \,\pi_{\rho, \tilde{\rho}}\, \phi \ra^\Omega_\rho \: dz \\
&= -\frac{1}{2 \pi i} \ointctrclockwise_\Gamma \sqrt{z}\; \la \big(B-\overline{z})^{-1} \,\pi_{\rho, \tilde{\rho}} \,\psi \,|\, \pi_{\rho, \tilde{\rho}} \,\phi \ra^\Omega_\rho \: dz \\
&= \la \bigg( \frac{1}{2 \pi i} \ointctrclockwise_{\Gamma} \sqrt{\overline{z}} \:\big(B-\overline{z})^{-1} \: d\overline{z} \bigg)\, \pi_{\rho, \tilde{\rho}}\, \psi \,|\, \pi_{\rho, \tilde{\rho}} \,\phi \ra^\Omega_\rho  \\
&= \la \bigg( -\frac{1}{2 \pi i} \ointctrclockwise_{\overline{\Gamma}} \sqrt{z} \:\big(B-z)^{-1}  \: dz \bigg)\, \pi_{\rho, \tilde{\rho}}\, \psi \,|\, \pi_{\rho, \tilde{\rho}}\, \phi \ra^\Omega_\rho
= \la \sqrt{B}\, \pi_{\rho, \tilde{\rho}}\,\psi \,|\, \pi_{\rho, \tilde{\rho}}\,\phi \ra^\Omega_\rho \:,
\end{align*}
where the minus sign in the last line comes about because the complex conjugate contour has the
opposite orientation. This shows that also the operator~$\sqrt{B}$ is symmetric with respect to~$\la .|. \ra^\Omega_\rho$. A similar computation yields
\begin{align*}
&\la \sqrt{B}\, \pi_{\rho, \tilde{\rho}} \,\psi \,|\, \sqrt{B}\, \pi_{\rho, \tilde{\rho}} \,\phi \ra^\Omega_\rho \\
&= -\frac{1}{4 \pi^2} \ointctrclockwise_{\Gamma} dz \ointctrclockwise_{\Gamma'} dz'\;
\sqrt{z}\:\sqrt{z'} \:\la \pi_{\rho, \tilde{\rho}} \,\psi \,|\, (B-z)^{-1}\: (B-z')^{-1}\, \pi_{\rho, \tilde{\rho}} \,\phi \ra^\Omega_\rho \:.
\end{align*}
In order to simplify the double integral, it is convenient to choose the contours such that~$\Gamma$
encloses~$\Gamma'$. Then~$z$ and~$z'$ are never equal. Applying the resolvent identity, we can carry out the integrals with residues,
\begin{align*}
&\la \sqrt{B}\, \pi_{\rho, \tilde{\rho}} \,\psi \,|\, \sqrt{B}\, \pi_{\rho, \tilde{\rho}} \,\phi \ra^\Omega_\rho \\
&= -\frac{1}{4 \pi^2} \ointctrclockwise_{\Gamma} dz \ointctrclockwise_{\Gamma'} dz'\;
\frac{\sqrt{z}\:\sqrt{z'}}{z-z'} \:\la \pi_{\rho, \tilde{\rho}} \,\psi \,|\, \Big( (B-z)^{-1} - (B-z')^{-1}\Big)\,
\pi_{\rho, \tilde{\rho}} \,\phi \ra^\Omega_\rho \\
&= \frac{1}{4 \pi^2} \ointctrclockwise_{\Gamma} dz \ointctrclockwise_{\Gamma'} dz'\;
\frac{\sqrt{z}\:\sqrt{z'}}{z-z'} \:\la \pi_{\rho, \tilde{\rho}} \,\psi \,|\, (B-z')^{-1}\,
\pi_{\rho, \tilde{\rho}} \,\phi \ra^\Omega_\rho \\
&= -\frac{1}{2 \pi i} \ointctrclockwise_{\Gamma'} dz'\; z'\:\la \pi_{\rho, \tilde{\rho}} \,\psi \,|\, (B-z')^{-1}\,
\pi_{\rho, \tilde{\rho}} \,\phi \ra^\Omega_\rho = 
\la \pi_{\rho, \tilde{\rho}} \,\psi \,|\, B\, \pi_{\rho, \tilde{\rho}} \,\phi \ra^\Omega_\rho \:.
\end{align*}
Using~\eqref{kreinrep}, it follows that
\[ \la \sqrt{B}\, \pi_{\rho, \tilde{\rho}} \,\psi \,|\, \sqrt{B}\, \pi_{\rho, \tilde{\rho}} \,\phi \ra^\Omega_\rho
= \la \psi | \phi \ra^\Omega_{\tilde{\rho}} \:. \]
Therefore, the densely defined operator
\beq \label{Idef}
{\mathscr{I}}^\Omega_{\rho, \tilde{\rho}} = \sqrt{B} \, \pi_{\rho, \tilde{\rho}} \::\: 
\D\big( {\mathscr{I}}^\Omega_{\rho, \tilde{\rho}} \big) \subset \scrW^\Omega_{\tilde{\rho}}
\rightarrow \scrW^\Omega_\rho
\eeq
with~$\D( {\mathscr{I}}^\Omega_{\rho, \tilde{\rho}} ) := \D(\hT^{\frac{1}{2}})$
has the property that
\beq \label{itrans}
\la \psi | \phi \ra^\Omega_{\tilde{\rho}} = \la {\mathscr{I}}^\Omega_{\rho, \tilde{\rho}}\, \psi \,|\, {\mathscr{I}}^\Omega_{\rho, \tilde{\rho}}\, \phi \ra^\Omega_\rho
\qquad \text{for all~$\psi, \phi \in \D( {\mathscr{I}}^\Omega_{\rho, \tilde{\rho}} )$} \:.
\eeq

For clarity, we finally explain how the above construction simplifies in the perturbative description.
For brevity, we do this to first order in an expansion parameter~$\tau \in [0,\delta]$ (the higher orders can be
worked out similarly).
If~$\tau=0$, the measures~$\tilde{\rho}$ and~$\rho$ coincide, implying that~$\pi_{\rho, \tilde{\rho}}$
and~$\hat{T}$ are the identity. Expanding to first order, after a suitable identification of the
spin spaces we obtain
\begin{align*}
\pi_{\rho, \tilde{\rho}} &= \1 + \tau\, \pi_{\rho, \tilde{\rho}}^{(1)} + \O\big(\tau^2 \big) \:,\qquad
\hT = \1 + \tau\, \hT^{(1)} + \O\big(\tau^2 \big) \\
\tilde{\Sig} &= \Sig + \tau\, \Sig^{(1)} + \O\big(\tau^2 \big) \:,\qquad
B = \1 + \tau\, B^{(1)} + \O\big(\tau^2 \big)\quad \text{with} \notag \\
B^{(1)} &= \Sig^{-1}\: \big(\pi_{\rho, \tilde{\rho}}^{(1)} - \hT^{(1)} \big)\: \Sig + \Sig^{-1}\, \Sig^{(1)} - \pi_{\rho, \tilde{\rho}}^{(1)} \notag \\
\sqrt{B} &= \1 + \frac{\tau}{2}\: B^{(1)} + \O\big(\tau^2 \big) \notag \\
{\mathscr{I}}^\Omega_{\rho, \tilde{\rho}} &= \sqrt{B} \, \pi_{\rho, \tilde{\rho}}
= \1 + \tau\,\pi_{\rho, \tilde{\rho}}^{(1)} + \frac{\tau}{2}\: B^{(1)} + \O\big(\tau^2 \big) \notag \\
&= \1 + \frac{\tau}{2} \, \Big ( \pi_{\rho, \tilde{\rho}}^{(1)} + \Sig^{-1}\: \big(\pi_{\rho, \tilde{\rho}}^{(1)} - \hT^{(1)} \big)\: \Sig + \Sig^{(-1)}\, \Sig^{(1)} \Big) + \O\big(\tau^2 \big)  \:.
\end{align*}
For these expression to be well-defined, we need the following assumptions:
\bitem
\item[(i)] $\Sig$ is injective, so that the operator~$\Sig^{-1}$ exists (as a densely defined selfadjoint operator).
\item[(ii)] The operator~$\pi_{\rho, \tilde{\rho}}^{(1)}-\hT^{(1)}$ maps the image of
the operator~$\Sig$ to the domain of~$\Sig^{-1}$.
\item[(iii)] The image of the operator~$\Sig^{(1)}$ lies in the domain of~$\Sig^{-1}$.
\eitem
These conditions need to be verified in the applications.

We finally remark that in Appendix~\ref{appinner} the connection between admissible variations
(see Definition~\ref{defadmissible}) and {\em{inner solutions}} is explained.

\subsection{The Extended Hilbert Space $\H^{\fermi, \Omega}_\rho$} \label{secextend}
We now let~${\mathfrak{M}}$ be a set of variation-admissible measures (see Definition~\ref{defadmissible}).
Then for every~$\tilde{\rho} \in {\mathfrak{M}}$, by combining~\eqref{itrans} with~\eqref{conscond},
we conclude that for any~$u,v \in \H^\fermi$
with~$\tilde{\psi}^u, \tilde{\psi}^v \in \D({\mathscr{I}}^\Omega_{\rho, \tilde{\rho}})$,
\beq \label{conserve}
\la u | v \ra_\H = \la {\mathscr{I}}^\Omega_{\rho, \tilde{\rho}} \,\tilde{\psi}^u \,|\, {\mathscr{I}}^\Omega_{\rho, \tilde{\rho}} \,\tilde{\psi}^v \ra^\Omega_\rho \:.
\eeq
We form the set of wave functions generated by all these measures,
\beq \label{HFextend}
{\mathfrak{G}}^\Omega_\rho := \bigcup_{\tilde{\rho} \in \mathfrak{M}}
\bigg\{ {\mathscr{I}}^\Omega_{\rho, \tilde{\rho}} \,\tilde{\psi}^u \:\bigg|\: u \in \H^\fermi \text{ with } \tilde{\psi}^u \in \D({\mathscr{I}}^\Omega_{\rho, \tilde{\rho}}) \bigg\}  \; \subset\; \scrW^\Omega_\rho
\eeq
(note that~${\mathfrak{G}}^\Omega_\rho$ is in general {\em{not}} a vector space).
We endow~$\text{span} \,{\mathfrak{G}}^\Omega_\rho$ with the inner product induced by~$\la .|. \ra^\Omega_\rho$.
A-priori, this
restriction merely is a sesquilinear form, but it need not be non-degenerate or even positive definite.
We take it as an additional assumption that the restriction is positive definite.

\begin{Def} Assume that~$(\text{\rm{span}}\, {\mathfrak{G}}^\Omega_\rho, \la .|. \ra^\Omega_\rho)$ is a scalar product space.
We refer to its completion as the {\bf{extended Hilbert space}}~$(\H^{\fermi, \Omega}_\rho,
\la .|. \ra^\Omega_\rho)$.
\end{Def} \noindent
Keeping in mind that the choices of~$Q^\reg(x,y)$ in~\eqref{Qreg} as well as~${\mathfrak{M}}$ leave us a lot of freedom, one can take the point of view that these choices should be made in such a way that the
resulting inner product~$\la .|. \ra^\Omega_\rho$ becomes positive definite.

The extended Hilbert space can be thought of as the generalization of the Hilbert space of all Dirac solutions
to the abstract setting of causal fermion systems.
According to~\eqref{conserve}, for any measure~$\tilde{\rho} \in {\mathfrak{M}}$
we have an isometric embedding
\[ \iota^\Omega_{\rho, \tilde{\rho}} := {\mathscr{I}}^\Omega_{\rho, \tilde{\rho}} \circ \Psi_{\tilde{\rho}}|_{\H^\fermi} \::\: \H^\fermi \hookrightarrow \H^{\fermi, \Omega}_\rho \:. \]

\section{A Linear Dynamics on the Extended Hilbert Space} \label{secdynamics}
\subsection{General Idea and Basic Construction}
We now consider two past sets~$\Omega, \Omega' \subset M$ with~$\Omega \subset \Omega'$
such that for both sets, the constructions in Section~\ref{secHextend} apply,
giving rise to the extended Hilbert spaces~$\H^{\fermi, \Omega}_\rho$ and~$\H^{\fermi, \Omega'}_\rho$.
In generalization of the linear dynamics described by the Dirac equation, it would be desirable
to have a linear time evolution on the extended Hilbert spaces, i.e.\ a
\beq \label{Utime}
\text{unitary mapping} \qquad U_\Omega^{\Omega'} : \H^{\fermi, \Omega}_\rho
\rightarrow \H^{\fermi, \Omega'}_\rho \:.
\eeq
We cannot expect to obtain such a linear time evolution in the setting of the previous section,
because the perturbed measure~$\tilde{\rho}$ may involve bosonic fields which
influence the dynamics. But we can hope that by restricting attention to specific
measures~$\mathfrak{M}^\gen \subset {\mathfrak{M}}$, the time evolution becomes unique.
The physical picture for the regularized Dirac sea vacuum
is that these variations describe the creation of
particle/anti-particle pairs surrounded by the linear electromagnetic field generated by them,
plus possible contributions to~$Q$ localized on the light cone which change the behavior of~$\hat{Q}$
on the upper mass shells (as explained in words in Section~\ref{secexmink}).

Given a measure~$\tilde{\rho} \in {\mathfrak{M}}$ and a vector~$u \in \H$, we consider the corresponding
wave functions in the extended Hilbert spaces denoted by
\[ \tilde{\psi}^{u,\Omega} := \iota^\Omega_{\rho, \tilde{\rho}} \,u \in \H^{\fermi, \Omega}_\rho \qquad \text{and} \qquad
\tilde{\psi}^{u,\Omega'} := \iota^{\Omega'}_{\rho, \tilde{\rho}} \,u \in \H^{\fermi, \Omega'}_\rho \:. \]
Note that these two wave functions are not defined globally in spacetime, but merely
as equivalence classes of wave functions in the corresponding surface layers.
We have the situation in mind that the two surface layers are separated by a sufficiently ``thick'' time strip
(see Figure~\ref{figevolve}).
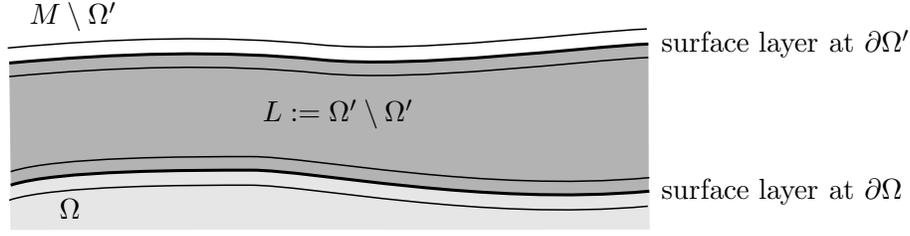
\begin{figure}
\psscalebox{1.0 1.0} 
{
\begin{pspicture}(3,28.14666)(8.5560255,30.84999)
\definecolor{colour0}{rgb}{0.9019608,0.9019608,0.9019608}
\definecolor{colour1}{rgb}{0.7019608,0.7019608,0.7019608}
\pspolygon[linecolor=colour0, linewidth=0.02, fillstyle=solid,fillcolor=colour0](0.050922852,28.754166)(0.47592285,28.834166)(1.2209228,28.899164)(2.1859229,28.939165)(3.125923,28.944164)(3.7759228,28.929165)(5.120923,28.784164)(5.948423,28.699165)(6.8559227,28.624165)(7.680923,28.614164)(8.530923,28.659164)(8.535923,28.171665)(0.05592285,28.166664)
\pspolygon[linecolor=colour1, linewidth=0.02, fillstyle=solid,fillcolor=colour1](0.04092285,30.371666)(0.48592284,30.401665)(0.99592286,30.431665)(1.5859228,30.471664)(2.0459228,30.481665)(2.8509228,30.486666)(3.4159229,30.466665)(4.0659227,30.421665)(4.5409226,30.386665)(5.030923,30.381664)(5.5859227,30.396666)(6.100923,30.421665)(7.1484227,30.506664)(7.823423,30.581665)(8.278423,30.616665)(8.518423,30.631664)(8.523423,28.681665)(8.138423,28.651665)(7.343423,28.636665)(6.7734227,28.651665)(6.0434227,28.686665)(5.3109226,28.756664)(4.6859226,28.821665)(3.918423,28.906666)(3.468423,28.951666)(2.8634229,28.956665)(1.4259229,28.931665)(0.81592286,28.891665)(0.37092286,28.836664)(0.06092285,28.771666)
\psbezier[linecolor=black, linewidth=0.04](0.030922852,30.376665)(1.4753677,30.510399)(2.8820548,30.54923)(4.160923,30.4216650390625)(5.439791,30.2941)(7.7077,30.596653)(8.52759,30.63)
\psbezier[linecolor=black, linewidth=0.04](0.050922852,28.751665)(0.65092283,28.951666)(2.450923,28.951666)(3.250923,28.9516650390625)(4.050923,28.951666)(6.340923,28.471664)(8.540923,28.671665)
\psbezier[linecolor=black, linewidth=0.02](0.010922852,30.576666)(1.4553677,30.7104)(2.8620548,30.74923)(4.140923,30.6216650390625)(5.419791,30.4941)(7.6877,30.796654)(8.507589,30.829998)
\psbezier[linecolor=black, linewidth=0.02](0.030922852,30.196665)(1.4753677,30.330399)(2.8820548,30.36923)(4.160923,30.2416650390625)(5.439791,30.1141)(7.7077,30.416653)(8.52759,30.449999)
\psbezier[linecolor=black, linewidth=0.02](0.050922852,28.931665)(0.65092283,29.131664)(2.450923,29.131664)(3.250923,29.1316650390625)(4.050923,29.131664)(6.340923,28.651665)(8.540923,28.851665)
\psbezier[linecolor=black, linewidth=0.02](0.030922852,28.551664)(0.63092285,28.751665)(2.4309227,28.751665)(3.230923,28.7516650390625)(4.030923,28.751665)(6.320923,28.271666)(8.520923,28.471664)
\rput[bl](0.7,28.3){$\Omega$}
\rput[bl](0.3,30.8){$M \setminus \Omega'$}
\rput[bl](3.4,29.5){$L:= \Omega' \setminus \Omega'$}
\rput[bl](8.7,30.45){surface layer at~$\partial \Omega'$}
\rput[bl](8.7,28.5){surface layer at~$\partial \Omega$}
\end{pspicture}
}
\caption{Time evolution from~$\H^{\fermi, \Omega}_\rho$ to~$\H^{\fermi, \Omega'}_\rho$.}
\label{figevolve}
\end{figure}
In this case, it is sensible to assume that there is a global wave function~$\tilde{\psi}^u \in C^0(M, SM)$ which 
in the respective surface layers coincides with~$\tilde{\psi}^{u,\Omega}$ and~$\tilde{\psi}^{u,\Omega'}$
(but clearly, the wave function~$\tilde{\psi}^u$ is far from unique, because it can be changed arbitrarily
away from the surface layers).
We set~$L:= \Omega' \setminus \Omega$ and
introduce~$\K_L$ as the Krein space with indefinite inner product
\beq
\bra \psi | \phi \ket_{\K_L} := \int_L \Sl \psi(x) | \phi(x) \Sr_x\: d\mu_t(x) \label{kreinL}
\eeq
and the topology induced by the scalar product
\[ \lla .|. \rra_{L} := \int_L \lla \psi(x) | \phi(x) \rra_x\: d\mu_t(x) \]
(where we used again the notation~\eqref{spinscalar}).
Moreover, we denote the restriction operator to the set~$L$ by
\[ \chi_L : C^0(M, SM) \rightarrow \K_L\:,\qquad \psi \mapsto \psi|_L \]
(extending this restriction by zero to all of~$M$, the operator~$\chi_L$ can also be
regarded as the multiplication by a characteristic function).
Since the commutator inner product is preserved by the variation of the measure, we know that
\beq \label{omcons}
\la \tilde{\psi}^{u,\Omega} | \tilde{\psi}^{u,\Omega} \ra^{\Omega}_\rho
= \la \tilde{\psi}^{u,\Omega'} | \tilde{\psi}^{u,\Omega'} \ra^{\Omega'}_\rho = \la u | u \ra_\H \:.
\eeq
As a consequence, a direct computation using~\eqref{Vtinner} gives
\begin{align}
0 &= \la \tilde{\psi}^u | \tilde{\psi}^u \ra^{\Omega'}_\rho - \la \tilde{\psi}^u | \tilde{\psi}^u \ra^{\Omega}_\rho \notag \\
&= -2i \int_L \Big( \Sl (Q^\reg\, \tilde{\psi}^u)(x) \:|\: \tilde{\psi}^u(x) \Sr_x  - \Sl \tilde{\psi}^u(x) \:|\: (Q^\reg\, \tilde{\psi}^u)(x) \Sr_x \Big)\: d\mu_t \notag \\
&= -2i \,\Big( \bra \chi_L \,(Q^\reg\, \tilde{\psi}^u) \:|\: \chi_L \,\tilde{\psi}^u \ket_{\K_t}  - \bra \chi_L \,\tilde{\psi}^u \:|\: \chi_L \,(Q^\reg\, \tilde{\psi}^u) \ket_{\K_L} \Big) \:. \label{consL}
\end{align}

Clearly, this computation can be carried out for any~$u \in \H$.
Polarizing, we conclude that
\beq \label{Qsymm}
\bra \chi_L \,(Q^\reg\, \tilde{\psi}^u) \:|\: \chi_L \,\tilde{\psi}^v \ket_{\K_t}  = \bra \chi_L \,\tilde{\psi}^u \:|\: \chi_L \,(Q^\reg\, \tilde{\psi}^v) \ket_{\K_L} \qquad \text{for all~$u, v \in \H^\fermi$}\:.
\eeq
In order to bring this equation into a more familiar form, we make the simplifying assumption
\beq \label{kerL}
\ker \chi_L \big|_{\iota_{\rho, \tilde{\rho}}(\H^\fermi)} \subset \ker \chi_L \,Q^\reg \big|_{\iota_{\rho, \tilde{\rho}}(\H^\fermi)} \:,
\eeq
where~$\iota_{\rho, \tilde{\rho}}$ is.a mapping which to every~$u \in \H^\fermi$ associates the
corresponding extended wave function in spacetime,
\[ \iota_{\rho, \tilde{\rho}} \::\: \H^\fermi \rightarrow C^0(M, SM)\:,  \quad u \mapsto \tilde{\psi}^u \:. \]

\begin{Lemma} Under the assumption~\eqref{kerL}, there is a linear operator
\[ R_L \::\: \text{\rm{span}}\, \iota_{\rho, \tilde{\rho}}(\H^\fermi) \subset \K_L \rightarrow \K_L \]
with the property that
\beq \label{TQR}
\chi_L \,Q^\reg \big|_{\iota_{\rho, \tilde{\rho}}(\H^\fermi)} = R_L\, \chi_L \big|_{\iota_{\rho, \tilde{\rho}}(\H^\fermi)} \:.
\eeq
\end{Lemma}
\Proof On the kernel of~$\chi_L$, the equation~\eqref{TQR} is trivially satisfied in view of~\eqref{kerL}.
Therefore, it suffices to arrange~\eqref{TQR} on a complement~$E$ of~$\ker \chi_L \big|_{\iota_{\rho, \tilde{\rho}}(\H^\fermi)}$ in~$\iota_{\rho, \tilde{\rho}}(\H^\fermi)$. On this complement, the operator~$\chi_L$
is injective by construction, and thus the mapping~$\chi_L|_E \::\: E \rightarrow \chi_L(E)$ is a bijection.
We denote its inverse by~$\chi_L^{-1} \::\: \chi_L(E) \rightarrow E$. Multiplying~\eqref{TQR} 
from the right by this inverse gives
\[ R_L \big|_{\chi_L(E)} = \chi_L \,Q^\reg\, \chi_L^{-1} \big|_{\chi_L(E)}\:. \]
Taking this equation as the definition of~$R_L$ concludes the proof.
\QED
Using~\eqref{TQR} in~\eqref{Qsymm}, we obtain the equation
\[ 
\bra R_L\, \chi_L \, \tilde{\psi}^u \:|\: \chi_L \,\tilde{\psi}^v \ket_{\K_t}  = \bra \chi_L \,\tilde{\psi}^u \:|\: R_L\, \chi_L\, \tilde{\psi}^v \ket_{\K_L} \qquad \text{for all~$u, v \in \H^\fermi$}\:, \]
which states that~$R_L$ is a symmetric operator on the subspace~$\chi_L \iota_{\rho, \tilde{\rho}} \H^\fermi$
of the Krein space~$\K_L$.

The basic idea is to take~\eqref{TQR} as the starting point for formulating
\beq \label{diracgen}
\chi_L \,Q^\reg \,\psi = R_L\, \chi_L \, \psi
\eeq
as the equation describing the dynamics on the extended Hilbert space. Being a linear equation for~$\psi$,
it can be viewed as a generalization of the Dirac equation to the setting of causal fermion systems.
The symmetric operator~$R_L$ on~$\K_L$ plays the role of an (in general nonlocal) potential.

Before implementing this idea, we must overcome the following difficulties:
In~\eqref{TQR}, the equation~\eqref{diracgen} holds only on the subspace of wave
functions~$\iota_{\rho, \tilde{\rho}} \H^\fermi$ which clearly depends on our choice of~$\tilde{\rho}$.
If we choose another measure~$\tilde{\rho}' \in {\mathfrak{M}}$ 
(where~${\mathfrak{M}}$ is again the set of admissible measures introduced
at the beginning of Section~\ref{secextend}) we get~\eqref{diracgen}
on another subspace~$\iota_{\rho, \tilde{\rho}'} \H^\fermi$, possibly with a different potential~$R'_L$.
In order to give~\eqref{diracgen} a universal meaning, we must ensure that the different potentials~$R_L, R_L', \ldots$
are all compatible with each other, making it possible to ``lift'' all the equations~\eqref{TQR} to the
single equation~\eqref{diracgen}.

The resulting compatibility conditions can be analyzed most conveniently in a linear perturbation expansion.
To this end, we consider variations~$(\tilde{\rho}_\tau)_{\tau \in [0,\delta]}$ of measures in
a subset~${\mathfrak{M}}^\gen \subset {\mathfrak{M}}$.
Using~\eqref{rhotau}, the infinitesimal generators of these variations
\[ \v := \frac{d}{d\tau} (f_\tau, F_\tau) \Big|_{\tau=0} \]
generate a space of jets denoted by~$\J^\gen_\rho$. We denote the first variation of the corresponding
wave functions in the surface layers by
\[ D\Psi^\Omega(\v, u) := \frac{d}{d\tau} \iota^\Omega_{\rho, \tilde{\rho}_\tau} \,u \Big|_{\tau=0} \qquad \text{and} \qquad
D\Psi^{\Omega'}(\v, u) := \frac{d}{d\tau} \iota^{\Omega'}_{\rho, \tilde{\rho}_\tau} \,u \Big|_{\tau=0} \:. \]
We thus obtain mappings
\[ D\Psi^\Omega \::\: \J^\gen_\rho \times \H^\fermi \rightarrow \H^{\fermi, \Omega}_\rho \qquad \text{and} \qquad
D\Psi^{\Omega'} \::\: \J^\gen_\rho \times \H^\fermi \rightarrow \H^{\fermi, \Omega'}_\rho \:, \]
which are real-linear in the first and complex-linear in the second argument.
We want to achieve that~\eqref{diracgen} holds for all the wave functions~$\psi$ obtained by our
variation~$(\tilde{\rho}_\tau)$ for all~$\tau$. Since~\eqref{diracgen} is linear in~$\psi$,
this means that this equation must hold order by order in perturbation theory.
In particular, it must hold for the contribution linear in~$\tau$. Repeating the consideration
leading to~\eqref{TQR} backwards, we conclude that also the conservation law in~\eqref{omcons} should hold if
we replace the wave function by its linearizations, i.e.\
\[ \la D\Psi^\Omega(\v, u) \,|\, D\Psi^\Omega(\v, u) \ra^{\Omega}_\rho
= \la D\Psi^{\Omega'}(\v, u) \,|\, D\Psi^{\Omega'}(\v, u) \ra^{\Omega'}_\rho \:. \]
This conservation law is indeed a good starting point for the formulation of the compatibility conditions,
which can be regarded as a polarized version of this equation.
\begin{Def} \label{defgen}
The jet space~$\J^\gen_\rho \subset \Jlin_\rho$ is a {\bf{compatible generator}} of the extended Hilbert spaces~$\H^{\fermi, \Omega}_\rho$
and~$\H^{\fermi, \Omega'}_\rho$ if the following conditions hold:
\bitem
\item[\rm{(i)}] Every~$\v \in \J^\gen_\rho$ is the infinitesimal generator of a variation~$(\tilde{\rho}_\tau)_{\tau \in [0,\delta]}$ of measures in the variation-admissible set~${\mathfrak{M}}$.
\item[\rm{(ii)}] The images of~$\Psi$ and~$D\Psi^\Omega$ generate a dense
subspace of~$\H^{\fermi, \Omega}_\rho$.
Likewise, the images of~$\Psi$ and~$D\Psi^{\Omega'}$ are dense in~$\H^{\fermi, \Omega'}_\rho$.
\item[\rm{(iii)}] The corresponding scalar products are preserved by the time evolution, i.e.\
for all~$\v, \v' \in \J^\gen_\rho$ and all~$u,u' \in \H^\fermi$,
\begin{align}
\big\la D\Psi^\Omega(\v,u) \,\big|\, \Psi(u') \big\ra^\Omega_\rho &= 
\big\la D\Psi^{\Omega'}(\v,u) \,\big|\, \Psi(u') \big\ra^{\Omega'}_\rho \label{Apres0} \\
\big\la D\Psi^\Omega(\v,u) \,\big|\, D\Psi^\Omega(\v',u') \big\ra^\Omega_\rho &= 
\big\la D\Psi^{\Omega'}(\v,u) \,\big|\, D\Psi^{\Omega'}(\v',u') \big\ra^{\Omega'}_\rho \:. \label{Apres}
\end{align}
\end{itemize}
\end{Def} \noindent
We remark for clarity that~(i) entails that the measures~$\tilde{\rho}_\tau$ are variation-admissible
(see Definition~\ref{defadmissible}) and that every~$\v \in \J^\gen_\rho$ satisfies~\eqref{sigmapreserve}.
We also remark that the condition in~(iii) will be analyzed further in Section~\ref{secsolvcons} below.

This definition immediately gives rise to the desired unitary operator~\eqref{Utime}.
Moreover, the dynamical equation~\eqref{diracgen} can be established as follows.
Again extending the wave functions arbitrarily in the region away from the surface layers, we obtain a mapping
\[ D\Psi \::\: \J^\gen_\rho \times \H^\fermi \rightarrow \H^\fermi_\rho \]
with the property that the restrictions to the surface layers coincides with~$D\Psi^\Omega$ and~$D\Psi^{\Omega'}$,
where the scalar product~$\la .|. \ra_\rho$ is induced by~$\la .|. \ra^\Omega_\rho$ or~$\la .|. \ra^{\Omega'}_\rho$
(which coincide according to~\eqref{Apres0} and~\eqref{Apres}).
Proceeding as in~\eqref{consL}, it follows that~\eqref{Qsymm} holds on~$(\H^\fermi_\rho, \la .|. \ra_\rho)$, i.e.\
\[ 
\bra \chi_L \,(Q^\reg\, \psi) \:|\: \chi_L \,\psi' \ket_{\K_t}  = \bra \chi_L \,\psi \:|\: \chi_L \,(Q^\reg\, \psi') \ket_{\K_L} \qquad \text{for all~$\psi, \psi' \in \H^\fermi_\rho$}\:. \]
Assuming similar to~\eqref{kerL} that
\[ 
\ker \chi_L \big|_{\H^\fermi_\rho} \subset \ker \chi_L \,Q^\reg \big|_{\H^\fermi_\rho} \:, \]
we finally obtain the evolution equation~\eqref{diracgen} 
for all~$\H^\fermi_\rho$ with~$R_L$ a symmetric operator
on the Krein space~$\K_L$.

\subsection{Solving the Compatibility Conditions} \label{secsolvcons}
The above general construction has the shortcoming that it is not obvious 
in which situations the condition in Definition~\ref{defgen}~(iii) can be fulfilled. 
The basic difficulty is that~\eqref{Apres} involves pairs of jets~$\v, \v' \in \J^\gen_\rho$,
making it impossible to choose~$\J^\gen_\rho$ as a subspace of the set of all jets
with certain properties. We now improve the situation by giving a strategy
for satisfying the condition in Definition~\ref{defgen}~(iii).

Assume that~$\J^\gen_\rho$ has a complex structure. Thus every~$\v \in \J^\gen_\rho$ has the decomposition
into holomorphic and anti-holomorphic components
\[ \v = z + \overline{z} \:. \]
Setting
\[ D\Psi^\Omega(\alpha z + \beta \overline{z}, u) :=
\frac{1}{2}\: D\Psi^\Omega \big((\alpha+\beta)(z+\overline{z}), u \big) + 
\frac{1}{2i}\: D\Psi^\Omega \big(i (\alpha-\beta)(z-\overline{z}), u \big) \:, \]
we extend~$D\Psi^\Omega$ to a complex linear functional its first argument.
\begin{Lemma} \label{lemmaconsist}
Assume that for all~$\v \in \J^\gen_\rho$ and~$u \in \H^\fermi$,
\begin{align}
\big\la D\Psi^\Omega(z,u) \,\big|\,\Psi(u) \big\ra^\Omega_\rho &= 
\big\la D\Psi^{\Omega'}(z,u) \,\big|\, \Psi(u) \big\ra^{\Omega'}_\rho \label{c0} \\
\big\la D\Psi^\Omega(z,u) \,\big|\, D\Psi^\Omega(z,u) \big\ra^\Omega_\rho &= 
\big\la D\Psi^{\Omega'}(z,u) \,\big|\, D\Psi^{\Omega'}(z,u) \big\ra^{\Omega'}_\rho \label{c1} \\
\big\la D\Psi^\Omega(\overline{z},u) \,\big|\, D\Psi^\Omega(\overline{z},u) \big\ra^\Omega_\rho &= 
\big\la D\Psi^{\Omega'}(\overline{z},u) \,\big|\, D\Psi^{\Omega'}(\overline{z},u) \big\ra^{\Omega'}_\rho \\
\big\la D\Psi^\Omega(z,u) \,\big|\, D\Psi^\Omega(\overline{z},u) \big\ra^\Omega_\rho &= 
\big\la D\Psi^{\Omega'}(z,u) \,\big|\, D\Psi^{\Omega'}(\overline{z},u) \big\ra^{\Omega'}_\rho \:. \label{c3}
\end{align}
Then the condition in Definition~\ref{defgen}~{\rm{(iii)}} is satisfied.
\end{Lemma}
\Proof We first show that~\eqref{c0} implies~\eqref{Apres0}. To this end, we first consider
the conservation law~\eqref{omcons} for the variation generated by~$e^{i \alpha} z + e^{-i \alpha} \overline{z}$
to first order. We thus obtain
\begin{align*}
0 &= \frac{d}{d\tau} \Big( \la \tilde{\psi}^{u,\Omega} | \tilde{\psi}^{u,\Omega} \ra^{\Omega}_\rho -
\la \tilde{\psi}^{u,\Omega'} | \tilde{\psi}^{u,\Omega'} \ra^{\Omega'}_\rho \Big) \Big|_{\tau=0} \\
&= 2\, \re \bigg\{ e^{i \alpha} \Big(
\big\la D\Psi^\Omega(\overline{z},u) \,\big|\,\Psi(u) \big\ra^\Omega_\rho + 
\big\la \Psi(u) \,\big|\, D\Psi^\Omega(z,u) \big\ra^\Omega_\rho \\
&\qquad\qquad\quad - \big\la D\Psi^{\Omega'}(\overline{z},u) \,\big|\, \Psi(u) \big\ra^{\Omega'}_\rho
- \big\la \Psi(u) \,\big|\, D\Psi^{\Omega'}(z,u) \big\ra^{\Omega'}_\rho \Big) \bigg\} \:.
\end{align*}
Since this holds for all~$\alpha$, the term in the round brackets must vanish. Combining this
with~\eqref{c0}, we conclude that
\beq \label{c01}
\big\la D\Psi^\Omega(\overline{z},u) \,\big|\,\Psi(u) \big\ra^\Omega_\rho = 
\big\la D\Psi^{\Omega'}(\overline{z},u) \,\big|\, \Psi(u) \big\ra^{\Omega'}_\rho \:.
\eeq

By complex polarization in~$u$, it follows that~\eqref{c0} and~\eqref{c01} also hold
for general~$u,u' \in \H$, i.e.\
\beq \label{cpol}
\begin{split}
\big\la D\Psi^\Omega(z,u) \,\big|\,\Psi(u') \big\ra^\Omega_\rho &= 
\big\la D\Psi^{\Omega'}(z,u) \,\big|\, \Psi(u') \big\ra^{\Omega'}_\rho \\
\big\la D\Psi^\Omega(\overline{z},u) \,\big|\,\Psi(u') \big\ra^\Omega_\rho &= 
\big\la D\Psi^{\Omega'}(\overline{z},u) \,\big|\, \Psi(u') \big\ra^{\Omega'}_\rho \:.
\end{split}
\eeq
Rewriting the scalar products~\eqref{Apres0} according to
\[ \big\la D\Psi^\Omega(\v,u) \,\big|\, \Psi(u') \big\ra^\Omega_\rho =
\big\la D\Psi^\Omega(\v,z) \,\big|\, \Psi(u') \big\ra^\Omega_\rho
+ \big\la D\Psi^\Omega(\v,\overline{z}) \,\big|\, \Psi(u') \big\ra^\Omega_\rho \]
(and similarly for~$\Omega'$), multiplying out and applying~\eqref{cpol}
gives~\eqref{Apres0}.

In order to derive~\eqref{Apres}, we first note that, by complex polarization, 
the equations~\eqref{c1}--\eqref{c3} are
also satisfied if the arguments on the right side of the scalar products take more general values, i.e.\
for all~$\v,\v' \in \J^\gen_\rho$ and~$u,u' \in \H^\fermi$,
\begin{align}
\big\la D\Psi^\Omega(z,u) \,\big|\, D\Psi^\Omega(z',u') \big\ra^\Omega_\rho &= 
\big\la D\Psi^{\Omega'}(z,u) \,\big|\, D\Psi^{\Omega'}(z',u') \big\ra^{\Omega'}_\rho \label{cn1} \\
\big\la D\Psi^\Omega(\overline{z},u) \,\big|\, D\Psi^\Omega(\overline{z'},u') \big\ra^\Omega_\rho &= 
\big\la D\Psi^{\Omega'}(\overline{z},u) \,\big|\, D\Psi^{\Omega'}(\overline{z'},u') \big\ra^{\Omega'}_\rho \\
\big\la D\Psi^\Omega(z,u) \,\big|\, D\Psi^\Omega(\overline{z'},u') \big\ra^\Omega_\rho &= 
\big\la D\Psi^{\Omega'}(z,u) \,\big|\, D\Psi^{\Omega'}(\overline{z'},u') \big\ra^{\Omega'}_\rho \:. \label{cn3}
\end{align}
Rewriting the scalar products in~\eqref{Apres} according to
\[ \big\la D\Psi^\Omega(\v,u) \,\big|\, D\Psi^\Omega(\v',u') \big\ra^\Omega_\rho
= \big\la D\Psi^\Omega(z+\overline{z},u) \,\big|\, D\Psi^\Omega(z'+\overline{z}',u') \big\ra^\Omega_\rho \]
(and similarly for~$\Omega'$), we can multiply out and apply~\eqref{cn1}--\eqref{cn3}.
This gives the result.
\QED
We point out that this Lemma poses conditions for each jet~$\v \in \J^\gen_\rho$.
Therefore, choosing~${\mathfrak{C}}$ as the set of all jets~$\v$ which satisfy these conditions
as well as the condition in Definition~\ref{defgen}~(i),
one can choose~$\J^\gen_\rho$ as a maximal subspace of~${\mathfrak{C}}$.
After doing so, the remaining question is whether the condition
in Definition~\ref{defgen}~(ii) is satisfied. In other words, the remaining issue is whether
the class of wave functions generated by~$\J^\gen_\rho$ is sufficiently large to
include all the wave functions in~$\H^{\fermi, t}_\rho$. This question also depends
on how large the set of measures~${\mathfrak{M}}$ in~\eqref{HFextend} is chosen.
Indeed, by adapting~${\mathfrak{M}}$ one can always arrange that~$\J^\gen_\rho$ is a compatible
generator. Proceeding in this way, the question remains whether the resulting extended Hilbert
space~$\H^{\fermi, t}_\rho$ includes all the wave functions of physical interest. Clearly, this question
can be answered only in a case-by-case basis depending on the concrete applications being considered.

\begin{Remark} {\bf{(The compatibility conditions for commutator jets)}} \label{remrolecomm} {\em{
It is a natural question whether the commutator jets should be included in~$\J^\gen_\rho$.
We now explain why in general this is not possible.

Let~$\v$ be a commutator jet of the form~\eqref{lincomm}.
Clearly, this jet is the infinitesimal generator of a variation given by~\eqref{Urhodef}
with~$\scrU=\scrU_\tau=e^{i \tau \scrA}$.
Moreover, as is verified in detail in Appendix~\ref{appcommute},
this variation preserves the symplectic form~\eqref{sigmapreserve} (see Corollary~\ref{corollaryA}).
This suggests that commutator jets should be included in~$\J^\gen_\rho$.
However, in general commutator jets do {\em{not}} satisfy the conditions~\eqref{Apres0} and~\eqref{Apres}
(see Proposition~\ref{prpA} and the explanation thereafter, where the connection to local gauge transformations is made). Therefore, the commutator jet cannot be included in the jet space~$\J^\gen_\rho$.
One should think of~$\J^\gen_\rho$ as formed of jets which do change the physical system,
but in a way where the perturbations of the physical wave functions respect current conservation. }} \QEDrem
\end{Remark}

We finally remark that in Appendix~\ref{appinner}, the compatibility conditions 
in Definition~\ref{defgen}~{\rm{(iii)}} are discussed for {\em{inner solutions}}.

\subsection{The Dynamical Wave Equation} \label{secdwe}
In the previous section, we considered the time evolution from a surface layer near~$\partial \Omega$
to a surface layer near~$\partial \Omega'$ (see Figure~\ref{figevolve}).
For our constructions to apply, we had to assume that the surface layers were separated
by a sufficiently large time strip. Under this assumption, the extended wave functions in the
two surface layers could be matched to a global wave function~$\psi$, leaving us
with a freedom to modify the wave function in the intermediate time strip, away from the
surface layers. We now explain how this construction can be extended to obtain a unitary dynamics
on a globally defined Hilbert space denoted by~$(\H^\fermi_\rho, \la .|. \ra_{\rho})$.

There are two possible methods. We begin with the first method and mention the alternative
method at the end of this section. The first method is to proceed step by step in time
by joining time strips together, as shown on the left of Figure~\ref{figevolvecont}.
\begin{figure}
\psscalebox{1.0 1.0} 
{
\begin{pspicture}(0,28.102541)(12.0931835,30.176922)
\definecolor{colour0}{rgb}{0.5019608,0.5019608,0.5019608}
\pspolygon[linecolor=colour0, linewidth=0.02, fillstyle=solid,fillcolor=colour0](6.670642,29.602192)(6.670642,29.51719)(6.990642,29.502192)(7.470642,29.482191)(7.890642,29.482191)(8.370642,29.492191)(8.870642,29.522192)(9.440642,29.547192)(10.110642,29.547192)(10.640642,29.512192)(11.265642,29.457191)(12.058142,29.38219)(12.058142,29.457191)(11.613142,29.492191)(11.058142,29.547192)(10.618142,29.582191)(10.118142,29.617191)(9.573142,29.627192)(8.633142,29.592192)(8.043142,29.567192)(7.523142,29.562191)(6.988142,29.57719)
\pspolygon[linecolor=colour0, linewidth=0.02, fillstyle=solid,fillcolor=colour0](0.050642088,29.942192)(0.050642088,29.877192)(0.3756421,29.867191)(0.75064206,29.87219)(1.2506421,29.902191)(1.8406421,29.937191)(2.375642,29.942192)(2.940642,29.922192)(3.455642,29.89219)(4.150642,29.83719)(4.610642,29.80219)(5.088142,29.76719)(5.433142,29.73719)(5.433142,29.797192)(4.858142,29.83719)(4.2281423,29.88219)(3.7981422,29.912191)(3.4531422,29.942192)(2.938142,29.97219)(2.498142,29.992191)(1.9081421,29.992191)(1.463142,29.97219)(1.0331421,29.947191)(0.6431421,29.922192)(0.2681421,29.92719)
\pspolygon[linecolor=colour0, linewidth=0.02, fillstyle=solid,fillcolor=colour0](0.02064209,29.164692)(0.02564209,29.074692)(0.29564208,29.064692)(0.6706421,29.04969)(1.0806421,29.044691)(1.4906421,29.04969)(1.9806421,29.064692)(2.3606422,29.089691)(2.960642,29.119692)(3.385642,29.109692)(3.835642,29.084692)(4.328142,29.03469)(4.843142,28.984692)(5.413142,28.94469)(5.413142,29.019691)(4.828142,29.05969)(4.253142,29.119692)(3.738142,29.164692)(3.3181422,29.199692)(2.893142,29.199692)(2.4931421,29.17469)(1.890642,29.13969)(1.4156421,29.129692)(0.6456421,29.129692)(0.19564208,29.144691)
\pspolygon[linecolor=colour0, linewidth=0.02, fillstyle=solid,fillcolor=colour0](0.028142089,28.289692)(0.028142089,28.19469)(0.23814209,28.189692)(0.6531421,28.199692)(0.99314207,28.21969)(1.4081421,28.249691)(1.833142,28.28469)(2.178142,28.289692)(2.553142,28.294691)(2.9931421,28.279692)(3.4631422,28.254692)(3.9006422,28.22969)(4.350642,28.20469)(5.4181423,28.14219)(5.4181423,28.20219)(5.148142,28.207191)(4.7081423,28.23719)(4.358142,28.262192)(3.938142,28.307192)(3.478142,28.33719)(3.053142,28.36219)(2.503142,28.377192)(2.018142,28.367191)(1.5581421,28.34719)(0.9106421,28.29969)(0.3056421,28.27469)
\psbezier[linecolor=black, linewidth=0.02](0.020306535,28.290926)(0.9318769,28.276152)(1.42182,28.381847)(2.4717948,28.385490481591567)(3.5217695,28.389133)(4.3478827,28.24095)(5.429403,28.21346)
\psbezier[linecolor=black, linewidth=0.02](0.020589633,28.183907)(0.9318101,28.154648)(1.4433712,28.292543)(2.493271,28.279498168738066)(3.5431712,28.266453)(4.346825,28.17516)(5.4277716,28.130486)
\pspolygon[linecolor=colour0, linewidth=0.02, fillstyle=solid,fillcolor=colour0](6.6681423,28.289692)(6.6681423,28.19469)(6.878142,28.189692)(7.2931423,28.199692)(7.633142,28.21969)(8.048142,28.249691)(8.473142,28.28469)(8.818142,28.289692)(9.193142,28.294691)(9.633142,28.279692)(10.103142,28.254692)(10.540642,28.22969)(10.990643,28.20469)(12.058142,28.14219)(12.058142,28.20219)(11.788142,28.207191)(11.348142,28.23719)(10.998142,28.262192)(10.578142,28.307192)(10.118142,28.33719)(9.693142,28.36219)(9.143142,28.377192)(8.658142,28.367191)(8.198142,28.34719)(7.550642,28.29969)(6.945642,28.27469)
\psbezier[linecolor=black, linewidth=0.02](6.660674,28.186874)(7.5713143,28.143135)(8.175002,28.307884)(9.224563,28.27815524065594)(10.274122,28.248426)(10.971224,28.184372)(12.051324,28.122524)
\psbezier[linecolor=black, linewidth=0.02](6.6605897,28.293907)(7.5718102,28.264647)(8.068371,28.392544)(9.118271,28.379498168738063)(10.168171,28.366453)(11.001825,28.25016)(12.082771,28.205486)
\psbezier[linecolor=black, linewidth=0.02](0.020674055,29.066874)(0.9313144,29.023136)(1.4500024,29.012884)(2.5195625,29.08315524065594)(3.5891225,29.153425)(4.3462243,28.989372)(5.426324,28.927525)
\psbezier[linecolor=black, linewidth=0.02](0.020674055,29.166874)(0.9313144,29.123137)(1.7300024,29.127884)(2.6045625,29.19315524065594)(3.4791226,29.258427)(4.3462243,29.089373)(5.426324,29.027525)
\psbezier[linecolor=black, linewidth=0.02](0.040674053,29.946875)(0.9513144,29.903135)(1.4650024,30.032885)(2.5145626,30.003155240655943)(3.5641227,29.973427)(4.3662243,29.869371)(5.446324,29.807524)
\psbezier[linecolor=black, linewidth=0.02](0.040674053,29.866875)(0.9513144,29.823135)(1.4650024,29.952885)(2.5145626,29.92315524065594)(3.5641227,29.893427)(4.3662243,29.789371)(5.446324,29.727524)
\psbezier[linecolor=black, linewidth=0.02](6.6617537,29.509823)(7.666584,29.441618)(8.247269,29.473185)(9.236224,29.526779254836764)(10.225179,29.580374)(10.985526,29.448586)(12.064507,29.36958)
\psbezier[linecolor=black, linewidth=0.02](6.6617537,29.609823)(7.5715837,29.551617)(8.047269,29.578186)(9.136224,29.626779254836762)(10.225179,29.675373)(10.990525,29.548584)(12.069507,29.46958)
\psline[linecolor=black, linewidth=0.04, arrowsize=0.05291667cm 2.0,arrowlength=1.4,arrowinset=0.0]{<-}(9.060642,29.037191)(9.060642,29.437191)
\psline[linecolor=black, linewidth=0.04, arrowsize=0.05291667cm 2.0,arrowlength=1.4,arrowinset=0.0]{<-}(9.060642,30.137192)(9.060642,29.73719)
\rput[bl](2.3,28.55){$L_1$}
\rput[bl](2.3,29.4){$L_2$}
\rput[bl](2.3,30.2){$L_3$}
\rput[bl](8,28.8){$L$}
\rput[bl](12.2,29.28){$\partial \Omega$}
\rput[bl](12.2,28.03){$\partial \Omega'$}
\rput[bl](-0.75,28.1){$\partial \Omega_1$}
\rput[bl](-0.75,29){$\partial \Omega_2$}
\rput[bl](-0.75,29.8){$\partial \Omega_3$}
\end{pspicture}
}
\caption{Deriving the continuous time evolution by patching together time strips (left)
or by varying a surface layer (right).}
\label{figevolvecont}
\end{figure}
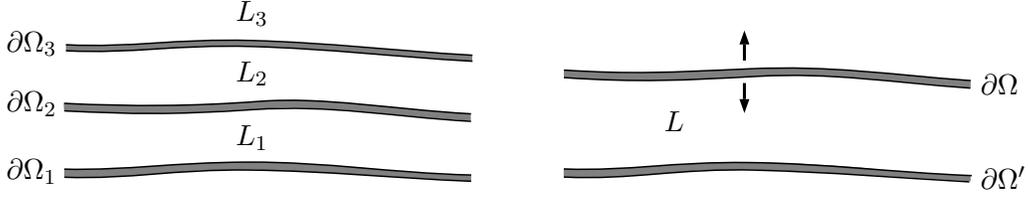
In order to get a good ``time resolution,'' it is clearly preferable to choose the time strips as ``thin'' as possible,
subject to the constraint that the wave functions in the adjacent surface layers must be compatible.
In this way, we obtain a Hilbert space of wave functions~$\H^\fermi_\rho$, where the
scalar product coincides with~$\la .|. \ra^{\Omega_\ell}_\rho$ for any of the
sets~$\Omega_1 \subset \Omega_2 \subset \ldots$. In order to combine the equations~\eqref{diracgen}
in the individual strips to a single equation, we write them as
\beq \label{diracgenalt}
\chi_{L_\ell} \,\big( Q^\reg - R_{L_\ell} \,\chi_{L_\ell} \big) \,\psi = 0 \:.
\eeq
Since the~$R_{L_\ell}$ map to~$\K_{L_\ell}$, extending by zero to~$\K_M$ and using orthogonality,
we can also write~\eqref{diracgenalt} as
\[ \chi_{L_\ell} \,\Big( Q^\reg - \sum_{\ell'} R_{L_{\ell'}} \,\chi_{L_{\ell'}} \Big) \,\psi = 0
\qquad \text{with} \qquad R_{L_{\ell'}} : \K_{L_{\ell'}} \rightarrow \K_M \:. \]
Carrying out the inner sum by setting
\[ R := \sum_{\ell'} R_{L_{\ell'}} \,\chi_{L_{\ell'}} \::\: \K_M \rightarrow \K_M \quad \text{symmetric} \:, \]
we conclude that for all~$\ell$ the equation
\beq \label{QL}
\chi_{L_\ell} \,\big( Q^\reg - R \big) \,\psi = 0
\eeq
holds. Moreover, the fact that~$R$ was constructed as a sum of operators acting on the
mutually orthogonal subspaces~$\K_{L_{\ell'}}$, we know that
\beq \label{Rinv}
\chi_{L_\ell} \,R = \chi_{L_\ell} \,R\, \chi_{L_\ell} \:.
\eeq
Writing~$R$ as an integral operator with kernel~$R(x,y)$, this kernel is symmetric in the sense that
\beq \label{Rsymm}
R(x,y)^* = R(y,x) \:.
\eeq
Moreover, the relation~\eqref{Rinv} means that the kernel
vanishes unless both arguments lie in the same time strip, i.e.
\[ R(x,y) = 0 \qquad \text{if~$x \in L_\ell$ and~$y \in L_{\ell'}$ with~$\ell \neq \ell'$}\:. \]
This in turn means that surface layer integrals formed of the kernel~$R(x,y)$ vanish
on any of the surface layers under consideration, i.e.\ symbolically
\[ \int_{\Omega_\ell} d\rho(x) \int_{M \setminus \Omega_\ell} d\rho(y)\: \cdots R(x,y) \cdots = 0 \:. \]
As a consequence, modifying the kernel~$Q^\reg(x,y)$ in our surface layer integrals
by the kernel~$R(x,y)$ does not change the values of the surface layer integrals.
This makes it possible to simplify our setting by introducing the abbreviation
\beq \label{Qdyndef}
Q^\dyn(x,y) := Q^\reg(x,y) - R(x,y) \:.
\eeq
and to change the definition~\eqref{Vtinner} to
\beq \label{OSIdyn}
\begin{split}
\la \psi | \phi \ra^\Omega_\rho = -2i \,\bigg( \int_{\Omega} \!d\rho(x) \int_{M \setminus \Omega} \!\!\!\!\!\!\!d\rho(y) 
&- \int_{M \setminus \Omega} \!\!\!\!\!\!\!d\rho(x) \int_{\Omega} \!d\rho(y) \bigg)\\
&\times\:
\Sl \psi(x) \:|\: Q^\dyn(x,y)\, \phi(y) \Sr_x \:.
\end{split}
\eeq
Since both~$Q^\reg(x,y)$ and~$R(x,y)$ are symmetric (see~\eqref{Qregsymm} and~\eqref{Rsymm}), so is~$Q^\dyn(x,y)$,
\beq \label{Qdynsymm}
Q^\dyn(x,y)^* = Q^\dyn(y,x) \:.
\eeq

We again point out that this redefinition does not change the value of any of our surface layer integrals.
Moreover, we write~\eqref{QL} in the shorter form
\beq \label{dynwave}
L \,Q^\dyn \,\psi = 0 \:,
\eeq
which holds for any set~$L$ of the form
\beq \label{Uform}
L \in {\mathfrak{A}} := \big\{ \Omega_\ell \setminus \Omega_{\ell-1} \:\big|\: \ell = 2,3,\ldots \big\} \:.
\eeq
The equation~\eqref{dynwave} is the {\em{dynamical wave equation}}
already discussed in the introduction. The sets in~${\mathfrak{A}}$ are referred to as
being {\em{con\-serva\-tion-admissible}}. The kernel~$Q^\dyn(x,y)$ is symmetric~\eqref{Qdynsymm}.
We note for clarity that we did not specify the regularity of the kernel~$R(x,y)$ in~\eqref{Qdyndef}.
For doing so, we would need to have more information on the operators~$R_{L_\ell}$ in~\eqref{diracgenalt}.
Consequently, also the new kernel~$Q^\reg(x,y)$ need not be continuous.
Not specifying the regularity of this kernel has the advantage that it became possible to
also absorb the right side in~\eqref{ELQreg} into this kernel. In particular, this kernel
may involve a $\delta$-contribution on the diagonal like~$-\mathfrak{r}\: \delta(x,y)$,
where the ``Dirac distribution'' merely is a convenient notation for the computation rule
\[ \int_M f(y)\: \delta(x,y)\: d\rho(y) = f(x) \qquad \text{for all~$f \in C^0(M, \R)$}\:. \]

The second, alternative strategy would be to work with one time strip, but to vary the set~$\Omega$, as
shown on the right of Figure~\ref{figevolvecont}. This procedure has the advantage that
the variation may involve arbitrarily ``thin'' time strips or even work with a continuous foliation
of spacetime. But there is the drawback that it is not obvious whether the operators~$R_L$
in~\eqref{diracgen} are compatible.
Here we do not decide for one or the other strategy. We rather take the point of view
that, as far as the macroscopic dynamics is concerned, both method give the same result.
Namely, the dynamics is described by an equation of the form~\eqref{dynwave},
where~$U$ can be chosen in a discrete or continuous family of conservation-admissible sets~${\mathfrak{A}}$,
which can be thought of as time strips between Cauchy surface layers.

\subsection{Example: The Regularized Minkowski Vacuum} \label{secexmink2}
We now return to the example in Section~\ref{secexmink}.
Choosing~$Q^\sing$ as in Proposition~\ref{prpQsing}, we arranged that~$Q^\reg$
is well-defined in the limit~$\varepsilon \searrow 0$.
The EL equations~\eqref{ELQreg} are satisfied for all Dirac solutions on the lower mass shell.
Moreover, it was shown in~\cite[Section~5]{noether} that the
surface layer integral~\eqref{OSIreg} coincides
(up to an irrelevant prefactor) with the usual scalar product on the Dirac solutions
induced by the conserved probability current.

Varying the system as explained in Section~\ref{secvary} generates physical wave functions
having contributions on the upper mass shell. After the transformations in Section~\ref{secositransform},
the scalar product of these wave functions is again given by the
surface layer integral~\eqref{Vtinner}. This surface layer integral can be computed again
with the methods in~\cite[Section~5]{noether}. It involves the $\omega$-derivative of~$\hat{Q}^\reg$
on the upper mass shell. For the more detailed explanation, we first consider the
simplified situation with one Dirac sea and discuss the case with several generations afterward.
In analogy to~\cite[eq.~(5.42)]{noether}, for a system of one Dirac sea
the surface layer integral is proportional to the momentum integral over the upper and lower
mass shells
\beq \label{OSIonesea}
\la \psi | \phi \ra^\Omega_\rho = c \sum_\pm
\int \frac{d^3k}{(2 \pi)^3} \, \Sl  \hat{\psi}(\vec k ) \: |\: 
\Big( (\partial^+_\omega + \partial^-_\omega) \hat{Q}^\reg \big( \pm \omega(\vec k) , \vec k \,\big) \Big)
 \hat{\phi}_\pm(\vec k ) \: \Sr \:,
\eeq
where~$c$ is a constant, and~$\omega(\vec{k}):= \sqrt{|\vec{k}|^2+m^2}$ and~$\phi_\pm$ are the Dirac solutions of
positive and negative frequency (and~$(\partial^+_\omega + \partial^-_\omega) \hat{Q}^\reg$ denote
directional derivatives).
Next, choosing the hypersurfaces on the left of Figure~\ref{figevolvecont} as the surface~$t=\text{const}$
and considering the limiting case where the ``time steps'' between two neighboring hypersurfaces is small,
the condition~\eqref{Rinv} means that the operator~$R$ is purely spatial.
Using that our system is translation invariant, this means that the Fourier transform of~$R$
is a multiplication operator depending on~$\vec{k}$. Hence, following~\eqref{Qdyndef},
\beq \label{Rhat}
\hat{Q}^\dyn(\omega, \vec{k}) = \hat{Q}^\reg(\omega, \vec{k}) + \hat{R}(\vec{k}) \:.
\eeq
Since~$\hat{R}$ does not depend on~$\omega$, we may replace the kernel~$\hat{Q}^\reg$
in~\eqref{OSIonesea} by~$\hat{Q}^\dyn$, giving agreement with~\eqref{OSIdyn}.
The dynamical wave equation~\eqref{dynwavenew} takes the form
\beq \label{Qdynonesea}
\hat{Q}^\dyn \big(\pm \omega(\vec{k}), \vec{k} \big) \, \hat{\psi}_\pm(\vec{k}) = 0 \:.
\eeq

In order to get agreement with Dirac theory, the dynamical wave equation~\eqref{Qdynonesea}
should give back the Dirac equation, whereas~\eqref{OSIonesea} should give Dirac current conservation.
Let us discuss how to get agreement, and how the construction depends on the choices of~$Q^\sing$
and~$R$: In order for~\eqref{OSIonesea} to reproduce current conservation, the
$\omega$-derivatives of~$\hat{Q}$ must combine to a constant times~$\gamma^0$,
\beq \label{currcorr}
(\partial^+_\omega + \partial^-_\omega) \hat{Q}^\reg \big( \pm \omega(\vec k) , \vec k \,\big)
= c'\, \gamma^0
\eeq
(with another constant~$c'$). On the lower mass shell, this identity was verified
by explicit computation in~\cite[Section~5]{noether}.
However, there is no reason why~\eqref{currcorr} should also hold on the upper mass shell.
Indeed, whether this relation holds or not may depend on the choice of~$\hat{Q}^\sing$.
Similarly, the relation~\eqref{Qdynonesea} holds on the lower mass shell
if we replace~$Q^\dyn$ by~$Q^\reg$ and assume that the
system is state stable (for details see~\cite[Section~5.6]{pfp} or again~\cite[Section~5]{noether}).
Again, there is no reason why~\eqref{currcorr} should hold on the upper mass shell.
Whether this equation holds or not depends on the choice of~$Q^\sing$.

Our task is to show that the operator~$\hat{R}$ in~\eqref{Rhat} can be chosen in such a way
that~\eqref{Qdynonesea} holds and to verify that replacing~$Q^\reg$ in~\eqref{OSIonesea}
by~$Q^\sing$ gives current conservation.
In order to satisfy~\eqref{currcorr}, we make the ansatz
\beq \label{Rk}
\hat{R}(\vec{k}) = h(\vec{k})\: (-\omega(k) \gamma^0 - \vec{k} \vec{\gamma} - m)
\eeq
with a real-valued function~$h$. Then~$\hat{R}$ vanishes on the Dirac solutions on the lower mass shell,
and hence~\eqref{Qdynonesea} holds under the assumption of state stability
even after the replacement~$Q^\reg \rightarrow Q^\sing$.
By choosing the function~$h(\vec{k})$ appropriately, one can arrange that~\eqref{Qdynonesea} also
holds for the Dirac solutions on the upper mass shell.
In this way, \eqref{currcorr} gives agreement with and generalizes the Dirac equation.

After this construction, the dynamics of the waves is the usual Dirac dynamics.
But the current integral~\eqref{OSIonesea} is different from the usual form, because~\eqref{currcorr}
may be violated even after the replacement~$Q^\reg \rightarrow Q^\sing$.
One strategy for dealing with this issue is to choose~$Q^\sing$ in such a way
that~\eqref{currcorr} holds. While this procedure seems most convenient for computational issues,
it is not compelling. Therefore, we prefer to take the point of view that the different form of the
conserved quantity merely corresponds to a different representation of the wave functions.
In order to understand how this comes about, we must return to the transformation
of the surface layer integrals in Section~\ref{secositransform},
where the operator~${\mathscr{I}}_{\rho, \tilde{\rho}}$ in~\eqref{Idef} was introduced
as an isometric embedding of Krein spaces~\eqref{itrans}.
The transformation of the wave functions by this operator ensures that the conservation law holds,
no matter how~$Q^\sing$ was chosen. More concretely, if 
the function~$(\partial^+_\omega + \partial^-_\omega) \hat{Q}^\reg$ in~\eqref{OSIonesea}
is multiplied at given~$(\omega(\vec{k}), \vec{k})$ by a constant~$\sigma$,
then the corresponding wave functions~$\hat{\psi}_+$ and~$\hat{\phi}_+$ are multiplied
by a factor~$\sigma^{-\frac{1}{2}}$, so that~\eqref{OSIonesea} is unchanged.
While this transformation clearly changes the wave functions, it has no effect on the
conserved current, nor on any other measurable quantity.
Thus, similar to a gauge transformation, the wave functions are changed, but the physics
remains the same, simply because the form of the scalar product, the currents and all
other quantities entering the interaction are transformed accordingly.

We now outline how the above arguments carry over to systems involving several generations.
In this case, the conditions~\eqref{currcorr} must be satisfied on each mass shell.
The ansatz~\eqref{Rk} does not involve enough degrees of freedom for satisfying all these conditions.
One way out is to also use the freedom in choosing~$Q^\sing$.
Alternatively, one can take the point of view that~$R(\vec{k})$ is not sufficient for
describing the dynamics in the ``time strips'' on the left of Figure~\ref{figevolvecont},
and that one should work instead with an operator which has a non-trivial dependence on the
generation index. Here we do not need to be specific, because each method gives enough
degrees of freedom for obtaining agreement with the Dirac dynamics.

We finally mention that, writing the dynamical wave equation in momentum space as
\[ \hat{Q}^\dyn(p)\: \hat{\psi}(p) = 0 \]
makes it possible to compute the Green's operators with Fourier methods.
Indeed, writing the Green's operator formally as the Fourier integral
\[ s(x,y) = \int \frac{d^4p}{(2 \pi)^4}\:  \big( \hat{Q}^\dyn(p) \big)^{-1}\: e^{-ip(x-y)}\: d^4p \:. \]
Provided that the integrand is meromorphic, one can deform
the complex $\omega$-contour so as to avoid the momenta 
on the upper and lower mass shell where~$\hat{Q}^\dyn(p)$
is not invertible, one gets the usual advanced and retarded Green's operators as well as the
Feynman propagator. 

\section{Analysis of the Dynamical Wave Equation} \label{secQdyn}
The considerations and constructions of the previous section led us to the
{\em{dynamical wave equation}} (see~\eqref{dynwave} and~\eqref{Uform})
\beq \label{dynwavenew}
L \,Q^\dyn \,\psi = 0 \qquad \text{for all $L \in {\mathfrak{A}}$}\:,
\eeq
where~${\mathfrak{A}}$ is a family of subsets of~$M$ referred to as the {\em{conservation-admissible sets}}.
We now turn attention to the mathematical analysis of this equation, with a focus
to the Cauchy problem and finite propagation speed.

\subsection{Finite Propagation Speed} \label{seccausal}
In order to analyze finite propagation speed, we need to be able to localize
wave functions to compact regions in spacetime. Since the resulting compactly supported
wave functions will not satisfy the dynamical wave equation, we need to extend~$\H^\fermi_\rho$.
To this end, we choose a set~$C^\vary(M, SM) \subset C^0(M, SM)$ of wave functions in spacetime.
We do not need to specify this set. The picture is that this
space should contain those wave functions with compact support which arise in the ``localization''
of the solutions.
Moreover, the wave functions in~$C^\vary(M, SM)$ should be thought of as being ``macroscopic'' in the
sense that they only vary on length scales which are large compared to the range
of the surface layer integrals.

\begin{Def} \label{defspacelike}
Let~$\Omega \subset M$ be a past set such that the commutator inner product~$\la .|. \ra^\Omega_\rho$
represents the scalar product (see Definition~\ref{defSLrep}). The set~$\partial \Omega$ is called {\bf{spacelike}} if
the extended commutator inner product~\eqref{OSIdyn} is positive
definite when restricted to~$\scrW^\vary_\rho \times \scrW^\vary_\rho$ with
\[ \scrW^\vary_\rho := \text{\rm{span}}\,\big(\H^\fermi_\rho, C^\vary(M, SM) \big) \:. \]
\end{Def}

We now work out in which sense the dynamics on~$\H^{\fermi, \Omega}_\rho$ respects causality.
One possible strategy would be to make use of corresponding results for the solutions
of the linearized field equations derived in~\cite{linhyp} and to analyze what they mean for the dynamics
of the physical wave functions. However, this strategy has the drawback that it works only under the
assumption that the jets~$\v \in \J^\gen_\rho$ generating the extended wave functions
satisfy the hyperbolicity conditions. This is not obvious because, as pointed out in~\cite[Section~6]{linhyp},
these hyperbolicity conditions are only satisfied by those degrees  of freedom which have a wave-like behavior in spacetime, and it is not at all obvious why the jets in~$\J^\gen_\rho$ should behave in this way.
This is the reason why we do not rely on hyperbolicity properties of~$\v$, but rather work
directly with the wave functions in the extended Hilbert space~$\H^\fermi_\rho$.

\begin{Def} Let~$\Omega, \Omega' \subset M$ be past sets such that the corresponding
commutator inner products represent the scalar product (see Definition~\ref{defSLrep}). 
A function~$\eta \in C^0(M, \R)$ {\bf{localizes to~$\partial \Omega \cap \partial \Omega'$}}
if for all~$\psi \in \scrW^\vary_\rho$, also~$\eta \psi \in \scrW^\vary_\rho$ and
\beq \label{localize}
\la \eta\, \psi \:|\: \phi \ra^{\Omega}_\rho = \la \eta\, \psi \:|\: \phi \ra^{\Omega'}_\rho \qquad
\text{for all~$\phi \in \scrW^\vary_\rho$}\:.
\eeq
The restriction~$\psi|_{\partial \Omega}$ of a wave function~$\psi \in \scrW^\vary_\rho$
to~$\partial \Omega$ is said to be {\bf{supported in~$\partial \Omega \cap \partial \Omega'$}} if there is~$\eta$
which localizes to~$\partial \Omega \cap \partial \Omega'$ such that~$\|(1-\eta) \psi\|^\Omega_\rho=0$. 
\end{Def}

\begin{Thm} \label{thmcausal} Let~$\Omega \subset \Omega'$ be past sets with spacelike boundaries.
Moreover assume that
\[ \qquad L := \Omega' \setminus \Omega
\quad \text{is conservation-admissible and relatively compact} \:. \]
Let~$\psi \in \H^\fermi_\rho$ be an extended solution whose restriction to~$\partial \Omega$
is supported in~$\partial \Omega \cap \partial \Omega'$. Then also its restriction to~$\partial \Omega'$
is supported in~$\partial \Omega \cap \partial \Omega'$.
\end{Thm} \noindent
The set~$L$ is also referred to as a {\em{lens-shaped region}}. The
statement of the theorem is illustrated in Figure~\ref{figcausal}.
\begin{figure}
\psscalebox{1.0 1.0} 
{
\begin{pspicture}(0,28.000273)(10.337705,31.269781)
\definecolor{colour0}{rgb}{0.9019608,0.9019608,0.9019608}
\definecolor{colour1}{rgb}{0.7019608,0.7019608,0.7019608}
\definecolor{colour2}{rgb}{0.5019608,0.5019608,0.5019608}
\pspolygon[linecolor=colour0, linewidth=0.02, fillstyle=solid,fillcolor=colour0](0.022648925,29.059145)(0.4476489,29.139145)(1.1926489,29.204144)(2.157649,29.244144)(3.0976489,29.249144)(4.127649,29.184145)(5.962649,29.039145)(7.510149,28.944145)(8.607649,28.899145)(9.452649,28.889145)(10.312649,28.964144)(10.317649,28.061644)(0.027648926,28.066645)
\pspolygon[linecolor=colour1, linewidth=0.04, fillstyle=solid,fillcolor=colour1](7.660149,28.056644)(8.175149,28.506645)(8.480149,28.801643)(8.565149,28.896645)(9.475149,28.901644)(9.565149,28.746645)(9.770149,28.491644)(9.955149,28.296644)(10.080149,28.186644)(10.170149,28.086645)
\pspolygon[linecolor=colour1, linewidth=0.04, fillstyle=solid,fillcolor=colour1](6.512649,31.256645)(10.295149,31.261644)(10.295149,30.086645)(10.160149,29.856644)(10.035149,29.676643)(9.795149,29.371645)(9.605149,29.131645)(9.455149,28.941645)(8.570148,28.946644)(8.160149,29.306644)(7.6851487,29.811644)(7.260149,30.286644)(6.785149,30.886644)
\pspolygon[linecolor=colour1, linewidth=0.02, fillstyle=solid,fillcolor=colour1](1.1576489,29.246645)(1.7726489,29.501644)(2.412649,29.836645)(3.027649,30.186644)(3.567649,30.441645)(4.242649,30.616644)(4.827649,30.461645)(5.480149,30.061644)(5.970149,29.696644)(6.355149,29.401644)(6.875149,29.121645)(7.1851487,28.971643)(6.160149,29.011644)(4.805149,29.131645)(3.520149,29.231644)(2.155149,29.276644)(1.1251489,29.251644)
\psbezier[linecolor=black, linewidth=0.08](0.6726489,29.171644)(2.1170938,29.305378)(3.3337808,30.62421)(4.3326488,30.576644287109374)(5.331517,30.52908)(6.669426,28.921633)(7.4893155,28.954977)
\psbezier[linecolor=black, linewidth=0.08](0.022648925,29.056644)(0.62264895,29.256645)(2.422649,29.256645)(3.2226489,29.256644287109374)(4.022649,29.256645)(8.122649,28.756645)(10.322649,28.956644)
\psbezier[linecolor=black, linewidth=0.02](8.542649,28.941645)(7.9576488,29.426643)(7.042649,30.451645)(6.487649,31.271644287109375)
\psbezier[linecolor=black, linewidth=0.02](9.470149,28.926643)(9.800149,29.341644)(10.097649,29.699144)(10.312649,30.079144287109376)
\psbezier[linecolor=black, linewidth=0.02](9.470149,28.924145)(9.815149,28.359144)(10.085149,28.254145)(10.200149,28.069144287109374)
\psbezier[linecolor=black, linewidth=0.02](8.535149,28.894144)(8.230149,28.559145)(8.007648,28.434145)(7.6226487,28.039144287109377)
\psbezier[linecolor=colour2, linewidth=0.16](7.5826488,31.276644)(7.542649,30.586645)(7.777649,29.131645)(7.612649,28.026644287109374)
\psbezier[linecolor=colour2, linewidth=0.16](0.42264894,31.296644)(0.52264893,30.606644)(0.40264893,29.171644)(0.45264894,28.046644287109373)
\psbezier[linecolor=black, linewidth=0.02, arrowsize=0.05291667cm 2.0,arrowlength=1.4,arrowinset=0.0]{->}(1.0526489,30.726645)(0.8676489,30.8472)(0.7776489,30.8272)(0.47264892,30.701644287109374)
\rput[bl](1.3,28.3){$\Omega$}
\rput[bl](1,29.9){$M \setminus \Omega'$}
\rput[bl](3.4,29.5){$L:= \Omega' \setminus \Omega$}
\rput[bl](8.3,29.9){$\supp \psi$}
\rput[bl](1.2,30.55){$\{ 0 < \eta < 1\}$}
\rput[bl](4,30.9){$\eta \equiv 0$}
\rput[bl](8.5,30.9){$\eta \equiv 1$}
\end{pspicture}
}
\caption{A lens-shaped region.}
\label{figcausal}
\end{figure}
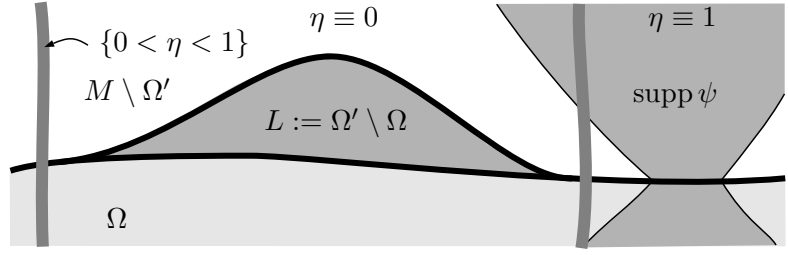
\Proof[Proof of Theorem~\ref{thmcausal}] Since~$L$ is relatively compact and conservation-admissible,
\[ \la \psi | \psi \ra^\Omega_\rho - \la \psi | \psi \ra^{\Omega'}_\rho
= \bra \psi \,|\, Q^\dyn\, \psi \ket_L -  \bra Q^\dyn\, \psi \,|\, \psi \ket_{\K_L} = 0 \:. \]
Setting~$\psi = \eta \psi + (1-\eta) \psi$ and multiplying out, we obtain
\begin{align}
0 &= \la \eta \psi \,|\, \eta \psi \ra^\Omega_\rho - \la \eta \psi \,|\, \eta \psi \ra^{\Omega'}_\rho 
+ 2 \re \Big( \big\la \eta \psi \,\big|\, (1-\eta) \psi \big\ra^\Omega_\rho
- \big\la \eta \psi \,\big|\, (1-\eta) \psi \big\ra^{\Omega'}_\rho \Big) \label{term1} \\
&\quad\, + \big\la (1-\eta) \psi \,\big|\, (1-\eta) \psi \big\ra^\Omega_\rho
- \big\la (1-\eta) \psi \,\big|\, (1-\eta) \psi \big\ra^{\Omega'}_\rho \:. \label{term2}
\end{align}
Using~\eqref{localize}, the summands~\eqref{term1} vanish, and thus
\[ \|(1-\eta) \psi\|^\Omega_\rho =  \|(1-\eta) \psi\|^{\Omega'}_\rho \:. \]
Since~$\psi|_{\partial \Omega}$ is supported in~$\partial \Omega \cap \partial \Omega'$,
the left side of this equation vanishes, and therefore also the right.
Hence also~$\psi|_{\partial \Omega'}$ is supported in~$\partial \Omega \cap \partial \Omega'$,
concluding the proof.
\QED

\subsection{Existence of Weak Solutions in Time Strips} \label{secexistence}
Characterizing spacelike hypersurfaces by the positivity of the commutator inner product
(see Definition~\ref{defspacelike}) also makes it possible to prove existence of solutions.
Our method is inspired by energy estimates for linear symmetric hyperbolic systems
as introduced by K.O.\ Friedrichs~\cite{friedrichs}. In the setting of causal variational principles,
similar methods were applied in~\cite{linhyp} to the linearized field equations.
We now adapt these energy methods to the dynamical wave equations.
For technical simplicity, we only consider the situation when spacetime admits
a global foliation by surface layers. Following the constructions in~\cite{linhyp},
our methods and results could be extended in a straightforward
way to lens-shaped regions.

We want to construct solutions of the dynamical wave equation~\eqref{dynwavenew}, without
the need to the restriction to conservation-admissible sets. Moreover, we want to
include an inhomogeneity~$w \in \K$. We thus consider the equation
\beq \label{dynfull}
Q^\dyn \psi = w \:.
\eeq
Similar as in~\cite[Section~3]{linhyp} we work with ``softened'' surface layer integrals
and foliations by surface layers, which we now introduce.
\begin{Def} \label{deflocfoliate}
Let~$\eta \in C^\infty(\R \times M, \R)$ be a function with~$0 \leq \eta \leq 1$ which for
all~$t \in \R$ has the following properties: 
\bitem
\item[{\rm{(i)}}] The support of the functions~$\eta(t,.)$ is a past set (see Definition~\ref{deffuturepast}).
\item[{\rm{(ii)}}] The function~$\theta(t,.) := \partial_t \eta(t,.)$ is non-negative.
\item[{\rm{(iii)}}] For every~$t \in \R$ there are~$t_0, t_1 \in \R$ such that the following implication holds,
\[ \eta(t,x) \neq 0 \qquad \Longrightarrow \qquad \eta(t_0,x) = 0  \quad \text{and} \quad \eta(t_1,x) = 1 \:. \] 
\item[{\rm{(iv)}}] The surface layers cover all of~$M$ in the sense that
\[ M = \bigcup_{t \in \R} \supp \,\theta(t,.) \:. \]
\eitem
We also write~$\eta(t,x)$ as~$\eta_t(x)$ and~$\theta(t,x)$ as~$\theta_t(x)$.
We refer to~$(\eta_t)_{t \in \R}$ as a {\bf{global foliation}} of~$M$.
\end{Def} \noindent
Assuming a global foliation can be understood as an implicit assumption on spacetime which
corresponds to global hyperbolicity of a Lorentzian manifold (for more details on this connection
in the context of causal variational principles see~\cite[Section~4]{linhyp}).

The ``softened'' version of the sesquilinear form~\eqref{OSIdyn} is defined for wave functions $\psi,\phi\in C_0^\vary(M,SM)$ by
\beq \label{OSIsoft}
\begin{split}
\la \psi | \phi \ra^t_\rho = -2i \int_M d\rho(x) \int_M d\rho(y) \:
\Big( \eta_t(x)\: \big( 1 &-\eta_t(y) \big) - \big(1-\eta_t(x) \big)\:\eta_t(y) \Big) \\
&\times \Sl \psi(x) \:|\: Q^\dyn(x,y)\, \phi(y) \Sr_x
\end{split}
\eeq
(where~$C^\vary_0(M, SM)$ denote the wave functions in~$C^\vary(M, SM)$ with compact support).
We also introduce the measure
\[ 
d\rho_t(x) := \theta_t(x)\: d\rho(x) \]
with~$t \in \R$. It is supported in the surface layer at time~$t$. 
Next, we introduce the scalar product in the time strip 
\begin{equation} \label{L2prod}
L:=\bigcup_{t\in [t_0,\tmax]} \mathrm{supp}\,\theta_t\qquad\text{by}\qquad \la \psi | \phi \ra_{L^2(L)} := \int_M  \lla \psi | \phi \rra_x \: \eta_I(x) \, d\rho(x) \:,
\end{equation}
where~$\eta_I = \eta_{\tmax} - \eta_{t_0}$, and~$\lla .|. \rra_x$ is the scalar product on the
spinors~\eqref{spinscalar}.
The corresponding Hilbert space is denoted by
$$
L^2(L,SM):=\left\{\psi\in L_{\loc}^2(L,SM)\:|\: \la \psi|\psi \ra_{L^2(L)} <\infty \ \right\}.
$$
In preparation of the construction of weak solutions, we introduce suitable function spaces.
The wave functions which are square integrable in every time strip are denoted by
\[ 
L^2_{\mathrm{loc,t}}(M,SM):=\{\psi\in L^2_{\loc}(M,SM)\:|\: \psi\in L^2(L,SM)\mbox{ for any } L\subset M \} \:. \]
Next, we want to introduce a concept similar to the notion of future- and past-compactness in Lorentzian geometry.
However, for technical simplicity we do not want to introduce the notion of spatial compactness
(this would make it necessary to invoke constructions similar to~\cite[Section~5.1]{linhyp}).
Therefore, we refer to a wave function as being future (past) compact if it vanishes in the future (past)
of some surface layer $\theta_{t_0}$ and if it is square integrable on every time strip $L$ as in~\eqref{L2prod}. 

\begin{Def}
A measurable wave function $\psi\in L^2_{\mathrm{loc,t}}(M,SM)$ is said to be
{\bf{past}} or {\bf{future compact}} if there exists $t_0\in\R$ such that
$$
\eta_{t_{0}}\psi=0\quad\mbox{or}\quad (1-\eta_{t_0})\,\psi=0 \:,
$$ 
respectively. A wave function which is both future and past compact is called
{\bf{timelike compact}}. The spaces of future, past and timelike compact wave functions are denoted by
$$
L_{\mathrm{fc}}^2(M,SM),\quad L_{\mathrm{pc}}^2(M,SM)\quad\mbox{and}\quad L_{\mathrm{tc}}^2(M,SM):=L_{\mathrm{fc}}^2(M,SM)\cap L_{\mathrm{pc}}^2(M.SM) \:,
$$
respectively. We also set
\begin{align*}
C^{\vary}_{\mathrm{tc}}(M,SM) &= C^\vary(M, SM) \cap L^2_{\mathrm{tc}}(M,SM) \\
C^{\vary}_{\mathrm{loc,t}}(M,SM) &=C^\vary(M,SM)\cap L^2_{\mathrm{loc,t}}(M,SM) \:.
\end{align*}
\end{Def} \noindent
It follows immediately from the definition that a wave function belongs to the space  $L^2_{\mathrm{tc}}(M,SM)$ if and only it is supported in some time strip $L$ and square integrable therein. In particular, by choosing $L$ large enough, it follows that any compactly supported square integrable wave function belongs to
$L_{\mathrm{tc}}^2(M,SM)$.

The following notion will serve as a technical simplification.
\begin{Def}\label{finitetimerange}
The operator $Q^\dyn$ is said to have {\bf{finite time range}} if the following properties are satisfied:\\[-0.9em]
\bitem
\item[{\rm{(i)}}] There is $r>0$ such that
\begin{equation}\label{deltarange}
Q^\dyn(x,y)=0\quad\mbox{whenever}\ x\in\supp\theta_{t_0},\ y\in \supp\theta_{t_1}\ \text{ and }\ |t_1-t_0|>r \:.
\end{equation}
\item[{\rm{(ii)}}] $Q^\dyn\big(L^2_{\mathrm{loc,t}}(M,SM)\big)\subset L^2_{\mathrm{loc,t}}(M,SM).$
\eitem
\end{Def} \noindent
In the applications like the example of the regularized Minkowski vacuum discussed in Section \ref{secexmink},
the operator~$Q^\dyn$ will in general {\em{not}} have this property. However, 
the kernel~$Q^\dyn(x,y)$ typically has good decay properties, making it possible to approximate it
by a kernel of finite time range. With this in mind, from now on we always assume that~$Q^\dyn$ has finite time range.
Under this assumption, given any $\psi\in L^2_{\mathrm{tc}}(M,SM)$ and choosing~$L$ such that~$\supp \psi \subset L$, the only contribution to the integral
$$
Q^\dyn\psi(x)=\int_M Q^\dyn(x,y)\:\psi(y)\,d\rho(y)=\int_L Q^\dyn(x,y)\:\psi(y)\,d\rho(y)
$$ 
comes from the points $y$ which lie in an $r$-neighborhood of $L$. Combining this fact
with Definition~\ref{finitetimerange}~(ii), we obtain the inclusion
$$
Q^\dyn\big(L^2_{\mathrm{tc}}(M,SM)\big)\subset L_{\mathrm{tc}}^2(M,SM) \:.
$$
Using this relation, for wave functions
\[ \psi \in L^2_{\mathrm{loc,t}}(M,SM) \qquad \text{and} \qquad \phi \in L^2_{\mathrm{tc}}(M,SM) \:, \]
the softened surface layer integral~\eqref{OSIsoft} can be rewritten as
\begin{equation}\label{OSIsoft2}
\la \psi | \phi \ra^t_\rho = -2i\,\int_L \Big( \Sl \phi(x) \,|\, (Q^\dyn \psi)(x) \Sr_x
- \Sl (Q^\dyn \phi)(x) \,|\, \psi(x) \Sr_x \Big)\,\eta_I(x)\,d\rho(x) \:,
\end{equation}
where~$L$ is the time strip corresponding to~$I=[t_0,t]$ and~$t_0$ sufficiently small.
Indeed, if we choose~$L$ such that
it contains the support of both wave functions~$\psi$ and~$Q^\dyn \psi$, then~\eqref{OSIsoft2}
is obtained from~\eqref{OSIsoft} by adding and subtracting the double integral
\[ \la \psi | \phi \ra^t_\rho = -2i \int_L d\rho(x) \int_L d\rho(y) \:
\eta_t(x)\: \eta_t(y)\: \Sl \psi(x) \:|\: Q^\dyn(x,y)\, \phi(y) \Sr_x \:, \]
which exists and is finite in view of the assumption in Definition~\ref{finitetimerange}~(ii).

We are now ready to enter the analysis of the Cauchy problem.
We begin with the following result, which can be proved by direct computation. 
\begin{Lemma} {\bf{(energy identity)}} \label{lemmaenid}
	For any wave functions~$\psi, \phi \in L^2_{\mathrm{loc,t}}(M,SM)$, 
	\beq 
	\frac{d}{dt} \la \psi | \phi \ra^t_\rho = -2i \int_M  \Big( \Sl \phi(x) \,|\, (Q^\dyn \psi)(x) \Sr_x
	- \Sl (Q^\dyn \phi)(x) \,|\, \psi(x) \Sr_x \Big) \: d\rho_t(x) \:. \label{enid}
	\eeq
\end{Lemma} \noindent
We again note that the right-hand side is well defined, because the support of the measure $d\rho_t$ can always be included in a sufficiently large time strip $L$.

\begin{Def} \label{defhypcond}
A global foliation~$(\eta_t)_{t \in \R}$ satisfies the {\bf{hyperbolicity condition}}
if for any compact interval~$I=[t_0, \tmax]$ there is a constant~$C(I)>0$ such that for all~$t \in I$
and all~$\psi \in C^\vary_\text{\rm{loc,t}}(M, SM)$,
\beq
\la \psi | \psi \ra^t_\rho \geq \frac{1}{C^2} \int_M \norm \psi(x)\norm_x^2 \: d\rho_t(x) \:. \label{hypcond}
\eeq
\end{Def} \noindent
This lower bound is a stronger and more quantitative version of positivity.
This positivity statement plays a similar role as the positivity of the energy
in the theory of linear symmetric hyperbolic systems (see~\cite{friedrichs} or~\cite[Section~5.3]{john}).

\begin{Prp} \label{prpenergy} Let~$\psi \in C^\vary_\text{\rm{loc,t}}(M, SM)$ be a wave function which vanishes initially,
	\beq \label{zeroinit}
	\|\psi\|^{t_0}_\rho=0 \:.
	\eeq
	Then, choosing
	\beq \label{Gammachoose}
	\Gamma = 2C^2\, (\tmax - t_0) \:,
	\eeq
	the following a-priori estimates hold,
	\begin{gather}
	\|\psi \|^t_\rho \leq 2C\, \sqrt{\tmax - t_0} \;\| Q^\dyn \psi \|_{L^2(L)} \qquad \text{for all~$t \in I$} \label{ees1} \\
	\| \psi \|_{L^2(L)} \leq \Gamma\, \| Q^\dyn \psi \|_{L^2(L)} \:. \label{ees2}
	\end{gather}
\end{Prp}
\Proof Applying the Schwarz inequality in~\eqref{enid} for~$\phi=\psi$ and using~\eqref{hypcond} gives
\[ \frac{d}{dt} \la \psi | \psi \ra^t_\rho \leq 4\,\| \psi \|_{L^2(d\rho_t)} \: \| Q^\dyn \psi \|_{L^2(d\rho_t)} 
\leq 4C\, \|\psi \|^t_\rho \: \| Q^\dyn \psi \|_{L^2(d\rho_t)} \]
and thus
\[ \frac{d}{dt} \| \psi \|^t_\rho \leq 2C\, \| Q^\dyn \psi \|_{L^2(d\rho_t)} \:. \]
Integrating over~$t$ and using~\eqref{zeroinit} gives~\eqref{ees1}.
Using again the hyperbolicity condition~\eqref{hypcond}, we obtain
\[ \| \psi \|_{L^2(d\rho_t)} \leq 2C^2\, \sqrt{\tmax - t_0} \:\| Q^\dyn \psi \|_{L^2(L)} \:. \]
Now we take the square, integrate again over~$t$ and take the square root. This gives
\[ \| \psi \|_{L^2(L)} \leq 2C^2\, (\tmax - t_0) \| Q^\dyn \psi \|_{L^2(L)} \:, \]
concluding the proof.
\QED

The estimate~\eqref{ees1} immediately gives the following result.
\begin{Corollary} {\bf{(Uniqueness of strong solutions)}}\label{uniqueness}
Let~$\psi, \psi' \in C^\vary_\text{\rm{loc, t}}(L,SM)$ be two solutions of the dynamical wave equation~\eqref{dynfull}
with zero initial data~\eqref{zeroinit} and inhomogeneity~$w \in L^2(L,SM)$. Then the solutions
coincide in~$L^2(L,SM)$.
\end{Corollary}

The estimate~\eqref{ees2}, on the other hand, yields the existence of weak solutions.
We closely follow the method in~\cite[Section~3]{linhyp}.
In preparation, we need to implement the initial data in a weak formulation.
To this end, we introduce the Krein space~$(\K_L, \bra .|. \ket_{\K_L})$ with inner product defined by
\beq \label{L2krein}
\bra \psi | \phi \ket_{\K_L} := \int_L  \Sl \psi | \phi \Sr_x \: \eta_I(x) \, d\rho(x)
\eeq
and the topology induced by the scalar product~\eqref{L2prod}.
\begin{Lemma} {\bf{(Green's formula)}} \label{lemmagreen}
For all~$\psi, \phi \in L^2_{\mathrm{loc,t}}(M, SM)$,
\[ 
\bra Q^\dyn \psi \,|\, \phi \ket_{\K_L} - 
\bra \psi \,|\, Q^\dyn \phi \ket_{\K_L} = - \la \psi | \phi \ra^{\tmax}_\rho + \la \psi | \phi \ra^{t_0}_\rho \:. \]
\end{Lemma}
\Proof Using the definition~\eqref{L2krein} and the symmetry of the kernel~$Q^\dyn(x,y)$, we obtain
\begin{align*}
&\bra Q^\dyn \psi \,|\, \phi \ket_{\K_L} -  \bra \psi \,|\, Q^\dyn \phi \ket_{\K_L} \\
&= \int_M d\rho(x) \int_M d\rho(y) \:\big( \eta_I(y) - \eta_I(x) \big)\: \Sl  \psi(x) \,|\, Q^\dyn(x,y)\, \phi(y) \Sr_x \:.
\end{align*}
Using the identity
\[ \eta_I(y) - \eta_I(x) =  \big(1-\eta_I(x) \big)\, \eta_I(y) - \eta_I(x)\, \big(1-\eta_I(y) \big) \:, \]
we can apply the definition of the ``softened'' scalar product~\eqref{OSIsoft} to obtain the result.	
\QED

For the weak formulation of the Cauchy problem, we need a space of wave functions
which vanish at and in the future of the time~$t$. For technical convenience, this space is
defined as
\beq \label{Cbardef}
\begin{split}
\overline{C}_0^t(M,SM):=&\left\{\psi\in  C^\vary_0(M, SM) \:|\: (1-\eta_t) \,\psi=0\right.\\
&\left.\hspace{3.7cm} \mbox{and}\  \|\psi\|^{t'}_\rho=0\ \mbox{for all}\ t'\ge t \right\} \:.
\end{split}
\eeq
Consider a strong solution of $Q^\dyn\psi=w$ with zero initial data as in~\eqref{dynfull}.
Taking the inner product of this identity with a vector~$\phi \in \overline{C}_0^\tmax$
and ``integrating by parts'' with the above Green's formula, we obtain
\[ \bra Q^\dyn \phi \,|\, \psi \ket_{\K_L} = \bra \phi \,|\, w \ket_{\K_L} - \la \psi | \phi \ra^{t_0}_\rho \:. \]
This makes it possible to implement boundary conditions by formulating
the weak equation as follows.
\begin{Def}
Let $w\in L^2(L,SM)$. A wave function $\psi\in L^2(L,SM)$ is said to be a {\bf{weak solution}}
of the dynamical wave equation in the time strip $L$  {\bf{with zero initial data}} if
\beq \label{weak}
\bra Q^\dyn \phi \,|\, \psi \ket_{\K_L} = \bra \phi \,|\, w \ket_{\K_L} \qquad
\text{for all~$\phi \in \overline{C}_0^\tmax(M,SM)$}\:.
\eeq
\end{Def}

\begin{Thm} {\bf{(existence of weak solutions)}} \label{thmexist}
	For every~$w \in L^2(L,SM)$ there is a solution~$\psi \in L^2(L,SM)$
	of the weak dynamical wave equation with zero initial data~\eqref{weak}. This solution is bounded by
	\beq \label{vbound}
	\|\psi\|_{L^2(L)} \leq \Gamma\, \|w\|_{L^2(L)} \:.
	\eeq
\end{Thm}
\Proof The weak dynamical wave equation~\eqref{weak} involves the indefinite Krein inner product.
The first step is to rewrite this equation in terms of the scalar product~$\la .|. \ra_{L^2(L)}$.
To this end, we make use of the Euclidean sign operator~$s_x \in \Lin(S_x)$ defined by
(for more details see~\cite[\S1.1.6]{cfs})
$$
\lla \chi|\xi\rra_x=\Sl \chi|s_x\,\xi \Sr_x \qquad \text{for all $\chi,\xi \in S_x$}\:.
$$
It has the properties~$s_x^*=s_x$ and $s_x^2=\1$. Introducing
the {\em{Euclidean operator}}~$\mathscr{E}$ as the operator which multiplies wave functions
pointwise by the Euclidean sign operator,
\[ \big( \mathscr{E} \psi \big)(x) := s_x\, \psi(x) \:, \]
the Krein inner product and the $L^2$-scalar product are related by
$$
\la \psi|\phi\ra_{L^2(L)}=\bra \psi \,|\, \mathscr{E} \phi\ket_{\K_L}\qquad\mbox{for all }\phi,\psi\in L_{\mathrm{loc,t}}^2(M,SM).
$$
Therefore, the weak equation~\eqref{weak} can be rewritten equivalently as
\[ 
\la Q^\dyn \phi \,|\, \mathscr{E} \psi \ra_{L^2(L)} = \la \phi \,|\, \mathscr{E} w \ra_{L^2(L)} \qquad \text{for all~$\phi \in \overline{C}_0^\tmax(M,SM)$} \]
(note that $\mathscr{E}^2=\1$ and $\la \mathscr{E} v| \mathscr{E}u\ra_{L^2(L)}=\la v|u\ra_{L^2(L)}$).

Clearly, the energy estimates of Proposition~\ref{prpenergy} also holds
if we exchange the roles of~$\tmax$ and~$t_0$, i.e.
\beq \label{hyprev}
\| \phi\|_{L^2(L)} \leq \Gamma\: \| Q^\dyn\phi\|_{L^2(L)} \qquad \text{for all~$\phi \in \overline{C}_0^\tmax(M,SM)$}
\eeq
(where ~$\Gamma$ is again the constant~\eqref{Gammachoose}).

We introduce the bilinear form
\begin{equation*}
\begin{split}
\llla .|. \rrra \::\:  &\overline{C}_0^\tmax(M,SM) \times \overline{C}_0^\tmax(M,SM) \rightarrow \R\\
&\llla \phi, v \rrra := \la Q^\dyn \phi \,|\, Q^\dyn v \ra_{L^2(L)} \:.
\end{split}
\end{equation*}
This is positive definitive, as follows from \eqref{hyprev}. Indeed, for any $\phi\neq 0$,
$$
\llla \phi|\phi \rrra = \| Q^\dyn\phi\|_{L^2(L)}^2\ge \Gamma^{-2} \:\| \phi\|_{L^2(L)}^2>0.
$$ 
Taking the completion with respect to this scalar product we obtain a Hilbert space
$(\mathcal{H}, \llla .,. \rrra)$. The corresponding norm is denoted by~$\norm . \norm$.

By construction, every  $u\in \mathcal{H}$ is the limit of a Cauchy sequence $(u_n)_{n \in \N}$ in
the normed space~$\overline{C}_0^\tmax(M,SM)$ with the norm $\norm . \norm$. The identity~\eqref{hyprev} gives for any $n,m\in\N$
$$
\| u_n-u_m\|_{L^2(L)}\le \Gamma\: \| Q^\dyn u_n-Q^\dyn u_m\|_{L^2(L)} = \Gamma\, \norm\! u_n-u_m\! \norm \:.
$$ 
In particular, $u_n$ and $Q^\dyn u_n$ are Cauchy in $L^2(L,SM)$. By completeness of $L^2(L,SM)$, we conclude that there is $\bar{u}\in L^2(L,SM)$ and $v\in L^2(L,SM)$ such that
$$
u_n\rightarrow \bar{u}\quad\mbox{ and }\quad Q^\dyn u_n\rightarrow v\quad\mbox{ in $L^2(L,SM)$}.
$$
Bearing in mind the definition of completion space, one can identify $u$ with $\bar{u}$. 
In view of Definition \ref{finitetimerange}~(ii), by extending $u$ to zero on the complement of $L$, we note  that $Q^\dyn u$ is a well-defined function in~$L^2_{\mathrm{tc}}(M,SM)$. It follows that~$v=(Q^\dyn u)|_L$.
Indeed,
\begin{equation*}
\begin{split}
\| v-Q^\dyn u\|_{L^2(L)}=\lim_{n\to \infty}\| Q^\dyn(u_n-u)\|_{L^2(L)}=\lim_{n\to\infty}\norm\! u_n-u \!\norm=0
\end{split}
\end{equation*}
In particular, 
$\H\subset L^2(L,SM)$ and $Q^\dyn(\H)\subset L^2(L,SM)$.
Moreover, by continuity, 
\begin{equation}\label{energyident2}
\| u\|_{L^2(L)}\le \Gamma\, \| Q^\dyn u\|_{L^2(L)}=\Gamma \norm u \norm\quad\mbox{for all $u\in \H$}.
\end{equation}

We now consider the linear functional~$\la \mathscr{E} w \,|\, . \ra_{L^2(L)}$ on $\H$.
Applying the Schwarz inequality and~\eqref{hyprev}, we obtain
\[ \big| \la \mathscr{E} w \,|\, u \ra_{L^2(L)} \big| \leq \| \mathscr{E} w \|_{L^2(L)} \:  \| u\|_{L^2(L)} 
\leq \Gamma\:\| \mathscr{E} w\|_{L^2(L)} \:  \norm u \norm \:, \]
proving that the linear functional~$\la \mathscr{E} w \,|\, . \ra_{L^2(L)}$ is bounded on~$\mathcal{H}$.
By the Fr{\'e}chet-Riesz theorem there is a unique vector~$V \in {\mathcal{H}}$ with
\[ 
\la \mathscr{E} w \,|\, u\ra_{L^2(L)} = \llla V | u \rrra = 
\la Q^\dyn V \,|\, Q^\dyn u \ra_{L^2(L)}\qquad \text{for all~$u \in {\mathcal{H}}$}\:. \] 
Hence
\beq \label{solform}
v := \mathscr{E}\,Q^\dyn V \big|_L \in L^2(L,SM)
\eeq
is the desired weak solution. 

It remains to prove the estimate~\eqref{vbound}. To this end,
we use that the Fr{\'e}chet-Riesz theorem also yields that the norm of~$v$
equals the sup-norm of the linear functional. Hence, using \eqref{energyident2}, it follows that
\begin{equation*}
\begin{split}
\| v\| _{L^2(L)} &= \| Q^\dyn V \|_{L^2(L)} = \norm V \norm = | \la \mathscr{E} w \,|\, . \ra_{L^2(L)} |_{\H^*}  \\
&\le \Gamma \norm \mathscr{E} w\norm_{L^2(L)}= \Gamma \norm w\norm_{L^2(L)}\:, 
\end{split}
\end{equation*}
This concludes the proof.
\QED

Before going on, we comment on the question of {\em{uniqueness}} of weak solutions in the time strip~$L$
with zero initial data. 
Given~$w\in L^2(L,SM)$, let~$\psi$ and~$\tilde{\psi}$ be two weak solutions to~\eqref{dynfull}.
Subtracting the corresponding weak equations~\eqref{weak}, it follows that
\[ \bra Q^\dyn \phi \,|\, \psi-\tilde{\psi} \ket_{\K_L} = 0 \qquad
\text{for all~$\phi \in \overline{C}_0^\tmax(M,SM)$}\:. \]
One should keep in mind that wave functions~$\psi$ and~$\tilde{\psi}$ satisfy the initial conditions only in the
weak sense. If we assume that these wave functions coincide strongly at initial time in the sense that
\[ \|\psi- \tilde{\psi} \|_\rho^{t_0}=0 \:, \]
the Green's formula of Lemma~\ref{lemmagreen} yields
\[ \bra \phi \,|\, Q^\dyn (\psi-\tilde{\psi}) \ket_{\K_L} = 0 \qquad
\text{for all~$\phi \in \overline{C}_0^\tmax(M,SM)$}\:. \]
Now, \textit{if the vector space~$\overline{C}_0^{\tmax}(M,SM)$ were dense in $L^2(L,SM)$}, 
it would follow that~$Q^\dyn (\psi-\tilde{\psi})=0$. If we knew
in addition that~$\psi-\tilde{\psi} \in C^\vary_{\text{loc, t}}(M, SM)$, Proposition~\ref{prpenergy} would yield
\[ \big\|\psi - \tilde{\psi} \big\|_{L^2(L)} \le \Gamma \,\big\|Q^\dyn (\psi-\tilde{\psi}) \big\|_{L^2(L)}=0 \:, \]
implying that~$\psi=\tilde{\psi}$ almost everywhere.
In most applications, however, the vector space~$\overline{C}_0^{\tmax}(M,SM)$
can{\em{not}} be chosen to be dense, because the hyperbolicity conditions of Definition~\ref{defhypcond}
typically holds only for a space $C^\vary_0(M, SM)$ of wave functions which are
sufficiently ``nice'' in the sense that they vary only on macroscopic scales.
In this case, the weak equation~\eqref{weak} determines the solutions only up to microscopic fluctuations
on length scales which are not accessible to measurements.
Nevertheless, the construction of Theorem~\ref{thmexist} gives a canonical solution
of a particular form (see~\eqref{solform}).

\subsection{Construction of Global Retarded Weak Solutions}
We now explain how to construct global retarded weak solutions.
Following~\cite[Sections~3.10 and~4.3]{linhyp} we work with the concept of {\em{shielding}},
adapted and simplified to our setting. We remark that a more general 
construction of global solutions based on an iteration scheme is given in~\cite{localize}.

\begin{Def} \label{defshieldnew}
The dynamical wave equation is {\bf{shielded in time strips}} if the following condition holds.
For every~$t_0<t_1$ there are~$t'_1>t_1$ and~$t'_0<t_0$ such that for all~$t_{\max} \geq t'_1$
and all sufficiently small~$t_{\min} < t'_0$, in the time strips~$L=L_{t_{\min}}^{t_1}$ and~$L'=L_{t_{\min}}^{t'_1}$
the following implication holds for all~$\phi_1\in \H_{t_{\min}}^{t_{\max}}$ and~$\phi_2 \in \H_{t'_0}^{t_{\max}}$
\beq \label{shieldcond2}
\begin{split}
\la Q^\dyn \phi_1 + (1-\eta_{t'_0}) \,Q^\dyn &\phi_2 \,|\, Q^\dyn \psi \ra_{L^2(L', d\rho)} = 0 \quad \forall\, \psi \in \overline{C}^{t'_1}_0 \\
&\Longrightarrow \qquad \big(Q^\dyn \phi_1 + (1-\eta_{t'_0}) \,Q^\dyn \phi_2 \big) \big|_L = 0\:.
\end{split}
\eeq
\end{Def} \noindent
We point out that this condition depends on the choice of the functions~$\eta_t$;
it can be understood as an implicit condition on these functions for large negative~$t$.
Alternatively, one could work with the weaker condition where~$\eta_{t'_0}$ is
replaced by a convolution, i.e.\
\[ \eta_{t'_0} \;\rightarrow\; \int_{-\infty}^\infty \Xi(\tau)\: \eta_{t'_0-\tau}\: d\tau \]
for a suitable test function~$\Xi$. This generalization is a direct consequence of the linearity of
the equation. For notational simplicity, we shall prove our results only for the stronger condition~\eqref{shieldcond2}.

Our strategy is to consider the weak solution constructed in Theorem~\ref{thmexist} in
time strips~$L_{t_0}^{t_1}$ and to take the limits~$t_1 \rightarrow \infty$ and~$t_0 \rightarrow -\infty$.
Similar to~\eqref{kreinL}, the Krein inner product in the whole spacetime is defined by
\beq \label{krein}
\bra \eta | \eta' \ket_\K := \int_M \Sl \eta(x) \,|\, \eta'(x) \Sr_x\: d\rho(x) \:.
\eeq

\begin{Prp} \label{lemmafuturepast}
Let~$w \in L_{\mathrm{pc}}^2(M,SM)$ be a past compact inhomogeneity.
Moreover, let~$\psi^{t_{\min}}_{t_{\max}}$ be the corresponding weak solution of Theorem~\ref{thmexist} in the time
strip~$L_{t_{\min}}^{t_{\max}}$. Then the following limit exists,
\beq \label{limL2loc}
\lim_{t_{\min} \rightarrow -\infty} \:\lim_{t_{\max} \rightarrow \infty} \psi_{t_{\min}}^{t_{\max}}
= \psi \qquad \text{with convergence in~$L^2_{\mathrm{loc,t}}(M, SM)$} \:.
\eeq
The resulting wave function is past compact, $\psi \in L^2_{\mathrm{pc}}(M, SM)$.
Moreover, it is a {\bf{global weak solution}}, i.e.\
\beq \label{globweak}
\bra Q^\dyn \phi \,|\, \psi \ket_{\K} = \bra \phi \,|\, w \ket_{\K} \qquad
\text{for all~$\phi \in C^\vary_0(M,SM)$} \:.
\eeq
\end{Prp}
\Proof Given~$t_0<t_1$, we choose~$t'_0$ and~$t'_1$ as
in Definition~\ref{defshieldnew}. 
Our first step is to extend the weak solution~$\psi_{t'_0}^{t'_1}$ by zero to the past.
To this end, we write out the spacetime integrals in the weak equation to obtain 
\[ \int_M \Sl (Q^\dyn \phi)(x) \,|\, \psi_{t'_0}^{t'_1}(x) \Sr_x\: \eta_{[t'_0, t'_1]}(x) \: d\rho(x)
= \int_M \Sl \phi(x) \,|\, w(x) \Sr_x\: \eta_{[t'_0, t'_1]}(x) \: \: d\rho(x) \:, \] 
valid for all~$\phi \in \overline{C}_0^{t'_1}(M,SM)$.
By increasing~$t'_1$ we can arrange that~$\eta_{[t'_0, t'_1]} = \eta_{t'_1}\, (1-\eta_{t'_0})$.
Since~$w$ is supported in the future of~$t'_0$, we obtain
\[ \int_M \Sl (Q^\dyn \phi) \,|\, (1-\eta_{t'_0})\, \psi_{t'_0}^{t'_1} \Sr_x\: \eta_{t'_1}(x)\: d\rho(x)
= \int_M \Sl \phi \,|\, w \Sr_x\: \eta_{t'_1}(x) \: d\rho(x) \:, \]
where we extended~$\psi_{t'_0}^{t'_1}$ by zero to the past of~$t'_0$.
Choosing~$t_{\min}<t'_0$ sufficiently small, both integrands vanish in the past of~$t_{\min}$. We thus obtain
the weak equation
\[ \int_M \Sl (Q^\dyn \phi) \,|\, (1-\eta_{t'_0})\, \psi_{t'_0}^{t'_1} \Sr_x\: \eta_{[t_{\min}, t'_1]}(x)\: d\rho(x)
= \int_M \Sl \phi \,|\, w \Sr_x\: \eta_{[t_{\min}, t'_1]}(x) \: d\rho(x) \:. \]
In the next step we choose~$t_{\min}<t_1$ and
subtract the weak equation for~$\psi_{t_{\min}}^{t_{\max}}$
restricted to the time strip~$[t_{\min}, t'_1]$.
Using that~$\overline{C}_0^{t_{\max}} \supset \overline{C}_0^{t'_1}$, we obtain
\[ 
\bra \psi_{t_{\min}}^{t_{\max}} -
(1-\eta_{t'_0})\, \psi_{t'_0}^{t'_1} \,\big|\, Q^\dyn \psi \ket_{\K_{L'}}
= 0 \quad \forall\, \psi \in \overline{C}^{t'_1}_0 \:. \]
Representing these weak solutions as~$\psi_{t_{\min}}^{t_{\max}}= \mathscr{E} Q V_{t_{\min}}^{t_{\max}}$
and~$\psi_{t'_0}^{t'_1}= \mathscr{E} Q V_{t'_0}^{t'_1}$ with
\[ V_{t_{\min}}^{t_{\max}} \in \H_{t_{\min}}^{t_{\max}} \qquad \text{and} \qquad V_{t'_0}^{t'_1}
\in \H_{t'_0}^{t'_1} \subset \H_{t'_0}^{t_{\max}} \]
(where in the last inclusion we extend the wave functions by zero),
we can apply~\eqref{shieldcond2} to conclude that
\[ \big( \psi_{t_{\min}}^{t_{\max}} - (1-\eta_{t'_0})\, \psi_{t'_0}^{t'_1} \big) \big|_{L_{t_{\min}}^{t_1}}=0\:.  \]
If $t_0'$ and $t_0''\in (t_{\min},t_0')$ are sufficiently small, the above identity has two consequences:
\bitem
\item[(a)] $\quad\displaystyle 0 = \big( \psi_{t_{\min}}^{t_{\max}} - (1-\eta_{t'_0})\, \psi_{t'_0}^{t'_1} \big) \big|_{L_{t_{0}}^{t_1}}=\big( \psi_{t_{\min}}^{t_{\max}} - \, \psi_{t'_0}^{t'_1} \big) \big|_{L_{t_{0}}^{t_1}} $
\item[(b)] $\quad \displaystyle 0 = \big( \psi_{t_{\min}}^{t_{\max}} - (1-\eta_{t'_0})\, \psi_{t'_0}^{t'_1} \big) \big|_{L_{t_{\min}}^{t_0''}}= \psi_{t_{\min}}^{t_{\max}}  \big|_{L_{t_{\min}}^{t''_0}}$
\eitem
From~(a) we see that the solution does not change on $L_{t_0}^{t_1}$ if $t_{\min}$ is further decreased or $\tmax$ further increased. This shows that the limit~\eqref{limL2loc}
exists. Next, (b) and the arbitrariness of $t_{\min}$ show that the global solution~\eqref{limL2loc} is past compact,
because it vanishes in the past of~$t_0''$.
\QED

\subsection{Causal Green's Operators} \label{secgreenretarded} 
\begin{Def} \label{defretarded}   The {\bf{retarded Green's operator}} $s^\wedge$ is defined as the mapping
\[ s^\wedge \:: L^2_{\mathrm{pc}}(M,SM) \rightarrow L^2_{\mathrm{pc}}(M,SM)\ , \qquad
s^\wedge(w) = -\lim_{t_{\min} \rightarrow -\infty} \:\lim_{t_{\max} \rightarrow \infty} \psi^{t_{\max}}_{t_{\min}} \:. \]
\end{Def} \noindent
Reverting the time direction and adapting the conditions in Definition~\ref{defshieldnew}
in an obvious manner, one obtains similarly the {\em{advanced Green's operator}}
\[ s^\vee \:: L^2_{\mathrm{fc}}(M,SM) \rightarrow L^2_{\mathrm{fc}}(M,SM)\:. \]

\begin{Lemma}
The Green's operators~$s^\wedge$ and~$s^\vee$ are linear.
\end{Lemma}
\begin{proof}
We only consider the retarded Green's operator (the proof for the advanced Green's operator is similar).
We choose a time strip~$L$. Let $w,v\in L^2_{\mathrm{pc}}(M,SM)$. From the proof of Proposition~\ref{lemmafuturepast} we know that, for sufficiently large $\tmax$ and sufficiently small $t_{\min}$,
$$
s^\wedge w=-\psi(w)_{t_{\min}}^\tmax,\quad s^\wedge v=-\psi(v)_{t_{\min}}^\tmax,\quad s^\wedge(v+\lambda w)=-\psi(v+\lambda w)_{t_{\min}}^\tmax\quad\mbox{on $L$} \:,
$$
where $\psi(u)_{t_{\min}}^{\tmax}$ is the weak solution on $L_{t_{\min}}^{\tmax}$ with inhomogeneity $u|_{L_{t_{\min}}^{\tmax}}$ as constructed in Theorem \ref{thmexist}.
Following the proof of Theorem \ref{thmexist}, one sees that 
\begin{equation}\label{formsolution}
\psi(u)_{t_{\min}}^{\tmax}=\mathcal{E}\,Q^\dyn V_u|_{L_{t_{\min}}^{\tmax}} \:,
\end{equation}
where $V_u$ is the unique vector $\H_{t_{\min}}^{\tmax}$ satisfying
$$
\bra u|\phi\ket_{\K_{t_{\min}}^{\tmax}}=\llla V_u|\phi\rrra\quad\mbox{for all }\phi\in \H_{t_{\min}}^{\tmax}.
$$
By uniqueness, one sees that $V_{u+\lambda u'}=V_u+\lambda V_{u'}$. 
Combining this result with~\eqref{formsolution} concludes the proof.
\end{proof}

We next analyze the kernel of the Green's operators. Our starting point is the observation
that in the global weak equation~\eqref{globweak},
the inhomogeneity~$w$ can be changed arbitrarily by wave functions in the orthogonal
complement of the test wave functions. More precisely, denoting the orthogonal complement
of~$C^\vary_0(M, SM)$ with respect to the Krein inner product~$\bra .|. \ket_\K$ by~$(C^\vary_0(M, SM))^\perp$, 
the right side of~\eqref{globweak} vanishes identically for any~$w \in (C^\vary_0(M, SM))^\perp$.
This suggests that also the Green's operators should vanish on such wave functions.
This is indeed the case, as is shown in the next lemma.

\begin{Lemma}\label{lemmaquotient}
The Green's operators vanish on~$\big( C^\vary_0(M,SM) \big)^\perp$ in the sense that
\begin{align}
\big( C^\vary_0(M,SM) \big)^\perp \cap L^2_{\mathrm{pc}}(M,SM) &\subset \ker s^\wedge \label{perp1} \\
\big( C^\vary_0(M,SM) \big)^\perp \cap L^2_{\mathrm{fc}}(M,SM) &\subset \ker s^\vee \label{perp2} \:.
\end{align}
\end{Lemma}
\begin{proof} We only prove~\eqref{perp1}, because the proof of~\eqref{perp2} is similar.
Thus let~$w \in \big( C^\vary_0(M,SM) \big)^\perp \cap L^2_{\mathrm{pc}}(M,SM)$.
In view of the definition of~$s^\wedge w$ as a limit (see Definition~\ref{defretarded}),
it clearly suffices to show that the weak solution~$\psi^{t_{\max}}_{t_{\min}}$ vanishes in
sufficiently large time strips. To this end, we return to the existence proof of Theorem~\ref{thmexist}.
We choose~$L$ so large that~$w$ vanishes in the initial time surface layer as well as in its past.
Then for any~$u \in \overline{C}_0^\tmax(M,SM)$,
\[ \la {\mathscr{E}} w | u \ra_{L^2(L)} = \la {\mathscr{E}} w | u \ra_{L^2(M, d\rho)}
= \bra w | u \ket_\K = 0 \]
(here we made use of the fact that~$u$ vanishes in the surface layer at time~$\tmax$; see~\eqref{Cbardef}).
Therefore, the linear functional~$\la \mathscr{E} w \,|\, . \ra_{L^2(L)}$ vanishes on a dense subset
of the Hilbert space~${\mathcal{H}}$, implying that it is represented by the zero vector~$0=V \in {\mathcal{H}}$.
As a consequence, also the weak solution~$v$ in~\eqref{solform} vanishes.
\QED

\subsection{The Causal Fundamental Solution and its Properties} \label{secfundamental}
The fermionic {\em{causal fundamental solution}} is introduced as the mapping
\beq \label{kdef}
k:= \frac{i}{2} \big( s^\vee - s^\wedge \big) \::\: L^2_{\text{tc}}(M, SM) \rightarrow L^2_{\text{loc,t}}(M, SM) \:.
\eeq
For the applications,  it is convenient to restrict attention to a smaller domain of ``nice''
wave functions. To this end, we denote the wave functions in~$C^\vary(M, SM)$ with spatially
compact support by~$C^\vary_\sc(M,SM)$ (meaning that the wave function has compact support
in any time strip). We define
\begin{align*}
\scrW^{**}_0(M,SM) &:= \big\{ \psi \in C^\vary_0(M,SM) \:\big|\: s^\vee Q^\dyn \psi = s^\wedge Q^\dyn \psi = -\psi
\big\} \\
\scrW^*_{\mathrm{tc}}(M,SM) &:= \big\{ \psi \in L^2_{\mathrm{tc}}(M,SM) \:\big|\: s^\vee \psi, s^\wedge \psi \in C^\vary_\sc(M, SM)\big\} \Big/ \big(C^\vary_0(M, SM) \big)^\perp \\
\scrW_{\mathrm{E}}(M,SM) &:= \big\{ s^\vee \psi_1 + s^\wedge \psi_2 \:\big|\: \psi_1 \in L^2_{\mathrm{fc}}(M,SM)
\text{ and }\  s^\vee \psi_1 \in C^\vary_\sc(M, SM), \\
&\qquad\qquad\qquad\qquad\;\: \psi_2 \in  L^2_{\mathrm{pc}}(M,SM) \,\text{ and } s^\wedge \psi_2 \in C^\vary_\sc(M, SM) \big\} \\
\scrW^*_{\mathrm{E}}(M,SM)&:= \big\{ \psi_1 + \psi_2 \:\big|\: \psi_1 \in L^2_{\mathrm{fc}}(M,SM) \text{ and } s^\vee \psi_1 \in C^\vary_\sc(M, SM), \\
&\qquad \qquad \qquad\:\, \psi_2 \in L^2_{\mathrm{pc}}(M,SM) \,\text{ and } s^\wedge \psi_2 \in C^\vary_\sc(M, SM) \big\}  \Big/ \\
&\!\!\!\!\!\!\!\!\!\!\!\! \Big( \big(C^\vary_0(M, SM) \big)^\perp \cap L^2_{\mathrm{fc}}(M,SM) \Big) +
\Big( \big(C^\vary_0(M, SM) \big)^\perp \cap L^2_{\mathrm{pc}}(M,SM) \Big) \:,
\end{align*}
where~$(C^\vary_0(M, SM))^\perp$ again denotes the orthogonal complement of~$C^\vary_0(M, SM)$ in~$L_{\mathrm{loc,t}}^2(M,SM)$ with respect to the Krein inner product~$\bra .|. \ket_\K$.
As explained before Lemma~\ref{lemmaquotient}, dividing out such wave functions reflects the
general structure of the global weak equation~\eqref{globweak}.
Working modulo such wave functions, in what follows we do not need to distinguish between weak and strong
solutions of the dynamical wave equation, making it possible to work with the relations~$Q^\dyn s^\vee=
Q^\dyn s^\wedge=-\1$. Moreover, in view of Lemma~\ref{lemmaquotient}, the Green's operators
are compatible with the quotient linear structure of $\scrW_{\mathrm{tc}}^*(M,SM)$,
giving rise to well-defined operators on the equivalence classes,
$$
s^\vee, s^\wedge \::\: \scrW_{\mathrm{tc}}^*(M,SM)\rightarrow \scrW_{\mathrm{E}}(M,SM) \:.
$$
Using~\eqref{kdef}, also the fundamental solution is well-defined on~$\scrW_{\mathrm{tc}}^*(M,SM)$.
In the definition of~$\scrW^*_{\mathrm{E}}(M,SM)$ we again mod out~$(C^\vary_0(M, SM))^\perp$,
but this time also respecting the decomposition into the sum of a future and a past compact wave function.
The index~``E'' indicates that the wave functions have finite {\em{energy}} in the sense that 
their ``energy norm'' $\norm . \norm$ introduced in the proof of Theorem~\ref{thmexist} is finite.

After these preparations, we can state the main result of this section.
\begin{Thm} \label{thmexact}
The following sequence is exact:
\begin{equation*} \label{exact}
0 \rightarrow \scrW^{**}_0(M,SM) \stackrel{Q^\dyn}{\longrightarrow} \scrW^*_{\mathrm{tc}}(M,SM)
\stackrel{k}{\longrightarrow} \scrW_{\mathrm{E}}(M,SM)
\stackrel{Q^\dyn}{\longrightarrow} \scrW^*_{\mathrm{E}}(M,SM)\rightarrow 0 \:.
\end{equation*}
\end{Thm}
\begin{proof}
We proceed in several steps. For simplicity of notation, we omit the arguments~$M$ and $SM$.
\begin{itemize}[leftmargin=2.3em]
\item[(i)] $Q^\dyn$ maps~$\scrW^{**}_0$ to~$\scrW^*_{\mathrm{tc}}$:
Let~$\psi \in \scrW^{**}_0$. Then the assumption of finite time range implies that $Q^\dyn\psi\in L^2_{\mathrm{tc}}$.
Moreover, by definition of~$\scrW^{**}_0$, we know that~$s^\vee Q^\dyn \psi = s^\wedge Q^\dyn \psi = -\psi
\in C^\vary_0 \subset C^\vary_\sc$. Using the definition of~$\scrW_{\mathrm{tc}}^*$, we conclude that~$Q^\dyn\psi\in \scrW_{\mathrm{tc}}^*$.
\item[(ii)] The mapping~$Q^\dyn : \scrW^{**}_0 \rightarrow \scrW^*_{\mathrm{tc}}$
is injective: Let $\psi\in \scrW^{**}_0$ with $Q^\dyn\psi=0$. Multiplying by~$s^\vee$
and using again the definition of~$\scrW^{**}_0$, we conclude that~$\psi = -s^\vee Q^\dyn \psi = 0$.
\item[(iii)] $Q^\dyn(\scrW^{**}_0) \subset \ker k$: Let $\psi\in \scrW^{**}_0$, again by definition of~$\scrW^{**}_0$,
$$
k(Q^\dyn \psi)=\frac{i}{2} \big(s^\vee Q^\dyn \psi-s^\wedge Q^\dyn \psi\big)=\frac{i}{2} 
\big( \psi - \psi \big) = 0 \:.
$$
\item[(iv)] $Q^\dyn(\scrW^{**}_0) \supset \ker k$: Let $\psi\in \scrW_{\mathrm{tc}}^*$ such that $k\psi=0$. Then,
$$
\phi:=-s^\vee\psi=-s^\wedge\psi\in C^\vary_\sc\:.
$$
Since $s^\vee\psi$ is supported in the past of some $t_1$ and $s^\wedge\psi$ is supported in the future of some $t_0<t_1$, we conclude that $\phi\in C^\vary_0$. 
Moreover, using that~$Q^\dyn s^\vee=-\1$, it follows that
$$
s^\vee Q^\dyn \phi =  -s^\vee Q^\dyn s^\vee \psi = s^\vee \psi = -\phi \:,
$$
and similarly for the retarded Green's operator. We conclude that~$\phi \in \scrW^{**}_0$ as desired. Finally, $\psi=Q^\dyn\phi$ by construction.
\item[(v)] $k(\scrW^*_{\mathrm{tc}})\subset \ker Q^\dyn$: Let $\psi=k\phi$ for some $\phi\in \scrW^*_{\mathrm{tc}}$. 
Then
$$
-2 i\, Q^\dyn(k\psi)=Q^\dyn(s^\vee\psi-s^\wedge \psi)= -\psi+\psi =0 \:.
$$
\item[(vi)] $\ker Q^\dyn\subset k(\scrW^*_{\mathrm{tc}})$: Let $\psi\in \scrW_{\mathrm{E}}$ with~$Q^\dyn\psi=0$. 
Then there are~$\psi_1$ and~$\psi_2$ as in the definition of~$\scrW^*_{\mathrm{tc}}$ such that
$$
0=Q^\dyn \psi= Q^\dyn(s^\vee \psi_1 +s^\wedge \psi_2)= -\psi_1-\psi_2 \:.
$$
As a consequence,
\[ \psi = s^\vee \psi_1 - s^\wedge \psi_1 = -2i\, k \psi_1 \:. \]
Moreover, by definition of~$\scrW^*_{\mathrm{tc}}$ we know that~$\psi_1 \in L^2_{\mathrm{fc}}$
and~$s^\vee \psi, s^\wedge \psi \in C^\vary_\sc$. We conclude that~$\psi_1 \in \scrW^*_{\mathrm{tc}}$.
\item[(vii)] $Q^\dyn(\scrW_{\mathrm{E}})= \scrW_{\mathrm{E}}^*$: We represent any~$\psi\in \scrW_{\mathrm{E}}^*$ 
as in the definition of~$\scrW_{\mathrm{E}}^*$ as~$\psi=\psi_1 + \psi_2$. Using the definition of~$\scrW_{\mathrm{E}}$,
it follows that the wave function~$\phi:=s^\vee\psi_1 + s^\wedge \psi_2$ is in~$\scrW_{\mathrm{E}}$. 
Moreover, $Q^\dyn\phi = -\psi_1-\psi_2=-\psi$. Hence~$\psi$ lies in the image of~$Q^\dyn$.
\end{itemize}
This concludes the proof.
\end{proof}

\subsection{Current Conservation for Weak Solutions}
As in the exact sequence of Theorem~\ref{thmexact}, we now consider the causal fundamental
solution as a mapping
\[ k \,:\, \scrW^*_{\mathrm{tc}}(M,SM) \rightarrow \scrW_{\mathrm{E}}(M,SM) \subset C^\vary_\sc(M, SM) \:. \]
In the following construction we are facing the basic problem that the causal fundamental solution
maps to {\em{weak}} solutions of the dynamical wave equation. For the conservation of the
sesquilinear form~\eqref{OSIdyn} or its softened version~\eqref{OSIsoft}, however, we need
{\em{strong}} solutions. The way out is to ``soften'' the surface layer
integral with the help of cutoff operators, which we now introduce.
\begin{Def} \label{deffp}
A linear operator~$\check{\pi} : C^\vary_\sc(M, SM) \rightarrow C^\vary_\sc(M, SM)$ is
called {\bf{cutoff operator}} in a time strip~$[t_0, t_1]$ if
\[ 
\eta_{t_0}\:(1-\check{\pi}) \equiv 0 \equiv (1-\eta_{t_1})\:\check{\pi} \:. \]
Spacetime is {\bf{asymptotically strip partitioned}} if for every~$t \in \R$,
there is a cutoff operator in a time strip~$[t_0, t_1]$ with~$t_0>t$
and a cutoff operator in a time strip~$[t_0, t_1]$ with~$t_1<t$.
\end{Def} \noindent
The operator can be thought of as a smooth cutoff, which inside the time strip~$[t_0,t_1]$, however,
may be a nonlocal operator. Typically, $\check{\pi}$ is chosen as an idempotent operator
which decomposes~$C^\vary_\sc(M, SM)$ into the images of~$\check{\pi}$ and~$\1-\check{\pi}$.
Multiplying a wave function~$\psi \in C^\vary_\sc(M, SM)$ by a cutoff operator~$\check{\pi}$ gives a wave function
which is future compact. Likewise, multiplying by~$\1-\check{\pi}'$ (with~$\check{\pi}'$ another
cutoff operator) gives a past compact wave function.
As a consequence, $\check{\pi}\,(\1-\check{\pi}') \psi$ has compact support.
Likewise, the difference~$\check{\pi} \psi - \check{\pi}' \psi$ has compact support. To summarize,
\beq \label{compactgen}
\check{\pi}\,(\1-\check{\pi}'),\; (\check{\pi} - \check{\pi}')\::\: C^\vary_\sc(M, SM) \rightarrow C^\vary_0(M, SM) \:.
\eeq

In what follows, we assume that spacetime is asymptotically strip partitioned.
Replacing the cutoff function~$\eta$
by the operator~$\check{\pi}$ acting on the corresponding wave functions, we obtain the
surface layer integral 
\begin{align}
\la \psi | \phi \ra^{\check{\pi}}_\rho \,&\!:= -2i \int_M d\rho(x) \int_M d\rho(y) \Big( \Sl (\check{\pi} \psi)(x) \:|\: Q^\dyn(x,y)\, \big((\1-\check{\pi})\phi \big)(y) \Sr_x \notag \\
&\qquad\qquad\qquad\qquad\qquad\;\;\:- \Sl \big((\1-\check{\pi})\psi \big)(x) \:|\: Q^\dyn(x,y)\, (\check{\pi} \phi)(y) \Sr_x \Big) \notag \\
&= -2i \,\Big( \bra (\check{\pi} \psi) \,|\, Q^\dyn \big((\1-\check{\pi})\phi \big) \ket_\K
- \bra \big((\1-\check{\pi})\psi \,|\, Q^\dyn (\check{\pi} \phi) \ket_\K \Big) \:. \label{softpi}
\end{align}
Softening the surface layer integral this way, we obtain current conservation for weak
solutions:
\begin{Prp} \label{prpindepend}
Let~$\psi, \phi \in C^\vary(M, SM)$ be weak solutions of the dynamical wave equation. Then
the sesquilinear form~\eqref{softpi} does not depend on the choice
of the cutoff operator~$\check{\pi}$.
\end{Prp}
\Proof We generalize~\eqref{softpi} by working with two cutoff operators~$\check{\pi}$ and~$\check{\pi}'$,
\[ {}^{\check{\pi}'}\!\la \psi | \phi \ra^{\check{\pi}}_\rho :=
 -2i \,\Big( \bra (\check{\pi}' \psi) \,|\, Q^\dyn \big((\1-\check{\pi})\phi \big) \ket_\K
- \bra \big((\1-\check{\pi}')\psi \,|\, Q^\dyn (\check{\pi} \phi) \ket_\K \Big) \:. \]
If the cutoff operator~$\check{\pi}'$ is changed to~$\check{\pi}''$, this sesquilinear form is modified by
\begin{align*}
{}^{\check{\pi}'}\! &\la \psi | \phi \ra^{\check{\pi}}_\rho -  {}^{\check{\pi}''}\!\la \psi | \phi \ra^{\check{\pi}}_\rho \\
&=-2i \,\Big( \bra \big((\check{\pi}'-\check{\pi}'') \psi \big) \,|\, Q^\dyn \big((\1-\check{\pi})\phi \big) \ket_\K
+ \bra \big((\check{\pi}'-\check{\pi}'')\psi \big) \,|\, Q^\dyn (\check{\pi} \phi) \ket_\K \Big) \\
&=-2i \,\bra \big((\check{\pi}'-\check{\pi}'') \psi \big) \,|\, Q^\dyn \phi \ket_\K
\end{align*}
By~\eqref{compactgen}, the wave function~$(\check{\pi}'-\check{\pi}'') \psi$ is a test wave function, 
making it possible to apply the weak dynamical wave equation to obtain zero.
We conclude that the above sesquilinear form
is independent of the choice of~$\check{\pi}'$. The independence of the choice of~$\check{\pi}$ is
proved analogously.
\QED
For clarity, we point out that, by Definition~\ref{deffp}, there are sequences of cutoff operators
whose supports move to infinity either to the future or to the past. This picture agrees with the conservation laws
for a sequence of surface layers as constructed in Section~\ref{secdwe} (see
the left of Figure~\ref{figevolvecont}). The existence of a continuous family of cutoff operators
(giving rise to a continuous time evolution of the current integral), although desirable in the applications,
is not necessary for our constructions.

We finally make a connection between this current integral and the Krein inner product
in the whole spacetime.
\begin{Prp} \label{prpkrein} The following relation holds for
any cutoff operator~$\check{\pi}$ and for all~$\eta, \eta' \in \scrW^*_{\mathrm{tc}}(M,SM)$,
\beq \label{kkrel}
\la k \eta \,|\, k \eta' \ra^{\check{\pi}}_\rho = \bra \eta \,|\, k \,\eta' \ket_\K \:,
\eeq
where~$\bra .|. \ket_\K$ is the Krein inner product~\eqref{krein}.
\end{Prp}
We begin with a preparatory lemma.
\begin{Lemma} \label{lemmaga} For all~$u, v \in \scrW^*_{\mathrm{tc}}(M,SM)$,
\[ \bra s^\wedge u \,|\, v \ket_\K = \bra u \,|\, s^\vee v \ket_\K \:. \]
\end{Lemma}
\Proof We choose a cutoff operator~$\check{\pi}$ in a time strip~$L$ which lies in the future
of the support of~$v$. Then
\[ \bra s^\wedge u \,|\, v \ket_\K = \bra \check{\pi}\, s^\wedge u \,|\, v \ket_\K \:. \]
Since~$s^\wedge u$ is past compact, the wave function on the left is in~$C^\vary_0(M, SM)$.
Therefore, we may apply the inhomogeneous weak equation to obtain
\[ \bra s^\wedge u \,|\, v \ket_\K = -\bra Q^\dyn \,\check{\pi}\, s^\wedge u \,|\,s^\vee v \ket_\K \:. \]
Since~$s^\wedge v$ is future compact
and~$Q^\dyn$ has finite time range, the cutoff operator~$\check{\pi}$ can be omitted if we
choose~$L$ sufficiently far in the future. We thus obtain
\[ \bra s^\wedge u \,|\, v \ket_\K = -\bra Q^\dyn s^\wedge u \,|\,s^\vee v \ket_\K \:. \]
Next, we choose a cutoff operator~$\check{\pi}'$ in a time strip~$L'$ which lies in the
past of the support of the wave function~$Q^\dyn s^\wedge u$. Then
\[ \bra s^\wedge u \,|\, v \ket_\K = -\bra Q^\dyn s^\wedge u \,|\, \check{\pi}' \,s^\vee v \ket_\K 
= \bra u \,|\, \check{\pi}' \,s^\vee v \ket_\K \:. \]
Choosing~$L'$ such that it lies in the future of the support of~$u$ give the result.
\QED

\Proof[Proof of Proposition~\ref{prpkrein}]
Let~$\eta, \eta' \in C^\vary_0$. Since the left side in~\eqref{kkrel} does not depend
on the choice of the cutoff operator (see Proposition~\ref{prpindepend}), we may choose~$\check{\pi}$
in a time strip~$L$ in the future of the supports of~$\eta$ and~$\eta'$.
Moreover, we can arrange that~$s^\vee \eta$ and~$s^\vee \eta'$ vanish in this time strip.
Then, applying~\eqref{kdef} only the terms involving the retarded Green's operators remain,
\begin{align*}
\la k \eta \,|\, k \eta' \ra^{\check{\pi}}_\rho 
&= \frac{1}{4}\: \la s^\wedge \eta \,|\, s^\wedge \eta' \ra^{\check{\pi}}_\rho \\
&= -\frac{i}{2} \,\Big( \bra \check{\pi} \,s^\wedge \eta \,|\, Q^\dyn s^\wedge \eta' \ket_\K
- \bra s^\wedge \eta \,|\, Q^\dyn\, \check{\pi} \, s^\wedge \eta' \ket_\K \Big) \:,
\end{align*}
where we used~\eqref{softpi} together with the fact that the wave functions~$\check{\pi} \,s^\wedge \eta$
and~$\check{\pi} \,s^\wedge \eta'$ have compact support.
Using the symmetry of~$Q^\dyn$, we can apply the weak equation to obtain
\begin{align*}
\la k \eta \,|\, k \eta' \ra^{\check{\pi}}_\rho 
&= -\frac{i}{2} \,\Big( \bra \check{\pi} \,s^\wedge \eta \,|\, Q^\dyn s^\wedge \eta' \ket_\K
- \bra Q^\dyn s^\wedge \eta \,|\, (\check{\pi} \, s^\wedge \eta' \ket_\K \Big) \\
&= \frac{i}{2} \,\Big( \bra \check{\pi} \,s^\wedge \eta \,|\, \eta' \ket_\K
- \bra \eta \,|\, \check{\pi} \, s^\wedge \eta' \ket_\K \Big)
= \frac{i}{2} \,\Big( \bra s^\wedge \eta \,|\, \eta' \ket_\K
- \bra \eta \,|\, s^\wedge \eta' \ket_\K \Big) \:,
\end{align*}
where in the last step we used that~$L$ lies to the future of the supports of~$\eta$ and~$\eta'$.
Applying Lemma~\ref{lemmaga} and again~\eqref{kdef} gives the result.
\QED

\appendix
\section{Variations of Surface Layer Integrals by Commutator Jets} \label{appcommute}
In this appendix, we compute a various surface layer integrals which involve commutator jets.
The main goals are to verify that the condition~\eqref{sigmapreserve}
in Lemma~\ref{lemmagpreserve} is satisfied if~$v_\tau$ is a commutator jet
(see Corollary~\ref{corollaryA}) and to show that commutator jets in general cannot be included
in the jet space~$\J^\gen$ (see Proposition~\ref{prpA}). We begin with a preparatory lemma.
\begin{Lemma} For two symmetric operators~$\scrA, \scrB \in \Lin(\H^\fermi)$,
the corresponding commutator jets satisfy the relation
\beq \label{gcomm}
\gamma^\Omega_\rho\Big( \big[ \Comm(\scrA), \Comm(\scrB) \big] \Big)
= - \frac{1}{2}\: \sigma^\Omega_\rho\big( \Comm(\scrA), \Comm(\scrB) \big) \:,
\eeq
where~$\big[ \Comm(\scrA), \Comm(\scrB) \big]$ denotes the commutator of
vector fields on~$\F$.
\end{Lemma}
\Proof We set~$\scrW_s = e^{i s \scrA}$ and~$\scrU_\tau = e^{i \tau \scrB}$. Then,
due to unitary invariance of the Lagrangian,
\[ \L_\kappa\big( \scrU_\tau \scrW_s \,x\, \scrW_s^{-1} \scrU_\tau^{-1}, 
\scrU_\tau \scrW_s^{-1} \,y\, \scrW_s \scrU_\tau^{-1} \big) =
\L_\kappa\big( \scrW_s \,x\, \scrW_s^{-1},  \scrW_s^{-1} \,y\, \scrW_s \big) \quad \text{for all~$s \in \R$}\:. \]
Differentiating with respect to~$s$ and~$\tau$ gives
\begin{align*}
0 &= \frac{d}{ds} \frac{d}{d\tau} \int_\Omega d\rho(x) \int_{M \setminus \Omega} d\rho(y)\:
\L_\kappa\big( \scrU_\tau \scrW_s \,x\, \scrW_s^{-1} \scrU_\tau^{-1}, 
\scrU_\tau \scrW_s^{-1} \,y\, \scrW_s \scrU_\tau^{-1} \big) \Big|_{s=\tau=0} \\
&= \int_\Omega d\rho(x) \int_{M \setminus \Omega} d\rho(y)\:
\big( D_{1,D_{\Comm(\scrA)} \Comm(\scrB)} - D_{2,D_{\Comm(\scrA)} \Comm(\scrB)} \big)
\L_\kappa(x,y) \\
&\quad\: + \int_\Omega d\rho(x) \int_{M \setminus \Omega} d\rho(y)\:\big( D_{1,\Comm(\scrB)} + D_{2,\Comm(\scrB)} \big) \big( D_{1,\Comm(\scrA)} - D_{2,\Comm(\scrA)} \big) \L_\kappa(x,y) \Big) \:.
\end{align*}
Anti-symmetrizing in the two commutator jets gives the result.
\QED
Assuming again that the commutator inner product represents the scalar product (see Definition~\ref{defSLrep}),
the left side of~\eqref{gcomm} can be computed further. In preparation, we need to compute
the commutator on the left side of~\eqref{gcomm}.
\begin{Lemma} The commutator of the vector fields~$\Comm(\scrA)$ and~$\Comm(\scrB)$ is
again a commutator jet, namely
\beq \label{commcomm}
\big[ \Comm(\scrA), \Comm(\scrB) \big] = -\Comm\big(i [\scrA, \scrB] \big) \:.
\eeq
\end{Lemma}
\Proof As in the previous proof, we set~$\scrW_s = e^{i s \scrA}$ and~$\scrU_\tau = e^{i \tau \scrB}$.
Then for any function~$f \in C^\infty(\F)$,
\begin{align*}
D_{\Comm(B)} f(x) &= \frac{d}{d\tau} f\big( \scrU_\tau \,x\, \scrU_\tau^{-1} \big) \Big|_{\tau=0} \\
D_{\Comm(A)} \big( D_{\Comm(B)} f(x) \big) &= 
\frac{d^2}{ds\, d\tau} f\big( \scrU_\tau \scrW_s \,x\, \scrW_s^{-1} \scrU_\tau^{-1} \big) \Big|_{s=\tau=0} \\
&= D^2f|_x \big( i [\scrA, x], i [\scrB, x] \big) + Df|_x \Big( i \big[ \scrB, i [\scrA, x] \big] \Big) \\
D_{[\Comm(A), \Comm(B)]} f(x) \big) &= D_{\Comm(A)} \big( D_{\Comm(B)} f(x) \big)
- D_{\Comm(B)} \big( D_{\Comm(A)} f(x) \big) \\
&= Df|_x \Big( i \big[ \scrB, i [\scrA, x] \big] -  i \big[ \scrA, i [\scrB, x] \big] \Big) \\
&= Df|_x \Big( - \big[ [\scrB,\scrA], x \big]\Big) = D_{\Comm(-i [\scrA, \scrB])} f(x) \:,
\end{align*}
giving the result.
\QED
Combining the two previous lemmas with~\eqref{repc}, we immediately obtain the following result.
\begin{Corollary} \label{corollaryA} If the commutator inner product represents the scalar product, then
for all symmetric operators~$\scrA, \scrB \in \Lin(\H)$ which vanish on~$(\H^\fermi)^\perp$, the
corresponding commutator jets satisfy the relations
\[ \sigma^\Omega_\rho\big( \Comm(\scrA), \Comm(\scrB) \big) = 0 \:. \]
\end{Corollary}
\Proof The computation
\begin{align*}
\sigma^\Omega_\rho\big( \Comm(\scrA), \Comm(\scrB) \big) &\overset{\eqref{gcomm}}{=}
-2\,\gamma^\Omega_\rho\Big( \big[ \Comm(\scrA), \Comm(\scrB) \big] \Big) \\
&\overset{\eqref{commcomm}}{=} 2  \gamma^\Omega_\rho\Big( \Comm\big( i [\scrA, \scrB] \big) \Big)
\overset{\eqref{repc}}{=} 2 \, \tr \big(  i [\scrA, \scrB] \big) = 0
\end{align*}
gives the result.
\QED
We conclude that the condition~\eqref{sigmapreserve} is satisfied for commutator jets.

We next work out in more detail how a commutator jet~$\v$ of the form~\eqref{jvdef}
modifies the physical wave functions. Setting~$F(x) = \scrU_\tau x \scrU_\tau^{-1}$,
the definition~\eqref{pirhodef} becomes
\begin{align*}
\big(\pi_{\rho, \tilde{\rho}}\, \psi \big)(x) &= |x| \big|_{S_x}^{-\frac{1}{2}} \,\pi_x
\, |\scrU_\tau x \scrU_\tau^{-1}| \big|_{\scrU_\tau S_x}^{\frac{1}{2}} \, \psi \big(\scrU_\tau x \scrU_\tau^{-1} \big) \\
&= |x| \big|_{S_x}^{-\frac{1}{2}} \,\pi_x\,\scrU_\tau
\, | x | \big|_{\tau S_x}^{\frac{1}{2}} \, \scrU_\tau^{-1}\,\psi \big(\scrU_\tau x \scrU_\tau^{-1} \big) \:,
\end{align*}
where in the last line we used that the operators~$x$ and~$\scrU_\tau x \scrU_\tau^{-1}$
are unitarily equivalent. Choosing~$\psi$ as the physical wave function corresponding to
a vector~$u \in \H^\fermi$, we obtain
\begin{align*}
\big(\pi_{\rho, \tilde{\rho}}\, \psi^u \big)(x) &= 
 |x| \big|_{S_x}^{-\frac{1}{2}} \,\pi_x\,\scrU_\tau
\, | x | \big|_{\tau S_x}^{\frac{1}{2}} \, \scrU_\tau^{-1}\, \pi_{\scrU_\tau x \scrU_\tau^{-1}} \,u \\
&= |x| \big|_{S_x}^{-\frac{1}{2}} \,\big( \pi_x
\, \scrU_\tau \,\pi_x \big) \,|x| \big|_{S_x}^{\frac{1}{2}}\, \pi_{x} \,\big( \scrU_\tau^{-1} u \big) \:.
\end{align*}
This formula involves the unitary operator~$\scrU_\tau$ twice: First, it transforms the
vector~$u$ to~$\scrU^{-1} u$. This transformation means that the physical wave functions
are transformed among each other. Second, the unitary operator~$\scrU_\tau$
also appears in the combination~$\pi_x \scrU_\tau \pi$. This term describes a transformation
of the spin space~$S_x$. We take its polar decomposition
\[ \pi_x \,\scrU_\tau\, \pi_x = V\,P \]
into a symmetric operator~$P$ and a unitary operator~$V$ on~$S_x$ (both with respect to
the spin scalar product~$\Sl .|. \Sr_x$).
The operator~$V$ can be interpreted as a local gauge transformation of the spinors at~$x$.

The above polar decomposition is useful for understanding how the local transformation of the spinors
enters the functional analytic construction in Section~\ref{secositransform}.
Since the commutator inner product~\eqref{OSIreg} is gauge invariant,
the transformation~$V$ drops out of the definition~\eqref{kreinrep} of the operator~$B$.
As a consequence, the mapping~${\mathscr{I}}_{\rho, \tilde{\rho}}$ in~\eqref{Idef}
still involves the local gauge transformation~$V$.
The resulting gauge phases do not drop out of the extended commutator inner product~\eqref{Vtinner}.
In particular, this consideration gives the following result.

\begin{Prp} \label{prpA}
A commutator jet~$\v \in \J^\Comm$ in general violates the conditions~\eqref{Apres0} and~\eqref{Apres}.
\end{Prp}

We finally remark that the above connection to gauge transformations suggests that
the definition~\eqref{pirhodef} should be modified to
\beq \label{pirhodefalt}
\big(\pi_{\rho, \tilde{\rho}}\, \psi \big)(x) := s_x\, (A_{xy})^{-\frac{1}{2}}\, P(x,y) \:.
\eeq
with~$y=F(x)$ (for the notation and related constructions see~\cite{lqg}).
This definition has the advantage that it describes a unitary mapping between from
the spin spaces~$S_y$ to~$S_x$, being geometrically more convincing than the
projection in the Hilbert space in~\eqref{pirhodef}. Working with this alternative definition,
a direct computation yields
\[ \big(\pi_{\rho, \tilde{\rho}}\, \psi^u \big)(x) 
= V(x)\, \pi_{x} \,\big( \scrU_\tau^{-1} u \big) \:, \]
where~$V(x)$ is a unitary transformation of~$S_x$ describing a local gauge transformation
of the physical wave functions. As a consequence, the relation~\eqref{kreinrep}
even implies that~$B$ is the identity, giving a cleaner argument.
Clearly, the definitions~\eqref{pirhodef} and~\eqref{pirhodefalt} have the same
physical content and merely modify the way the wave functions are represented in spacetime.
The reason why we prefer the first definition is that the projection in~\eqref{pirhodef}
seems more suitable for functional analytic constructions in Hilbert spaces.

\section{Variations of Surface Layer Integrals by Inner Solutions} \label{appinner}
A particular class of solutions of the linearized field equations are the
so-called {\em{inner solutions}} as introduced in~\cite[Section~3]{fockbosonic}.
They correspond to the symmetry of the causal action principle under diffeomorphisms
of~$M$. Such symmetry transformations can be used in order to arrange that the
scalar components of all linearized solutions vanish.
This procedure is carried out and explained in~\cite[Section~3.3]{fockbosonic} (see also~\cite[Section~2.1.4]{mass}
and~\cite[Section~2.9]{fockfermionic}). In the present paper, it is preferable {\em{not}} to
use this construction, but to allow for more flexibility by allowing
for linearized solutions with non-zero scalar components.
This is the reason why inner solutions have not been considered in this paper.
Nevertheless, we now explain how inner solutions fit into the picture.

The concept of inner solutions makes it necessary to assume that spacetime is 
has a smooth manifold structure, meaning that~$M$ is a smooth manifold and
that the measure~$d\rho$ is absolutely continuous with respect to the Lebesgue measure in a chart
with a smooth weight function, i.e.\
\[ d\rho = h(x)\: d^kx \qquad \text{with} \quad h \in C^\infty(\scrM, \R^+) \:. \]
Under these assumptions, a smooth vector field~$\bv \in \Gamma(M, TM)$ gives rise to a solution~$\v$ of the
linearized solution of the form (for details see~\cite[Section~3.1]{fockbosonic})
\[ \v = ( \div \bv, \bv ) \qquad \text{with} \qquad \div \bv := \frac{1}{h}\: \partial_j \big( h\, \bv^j \big) \:. \] 

It is a natural question
whether the variations generated inner solutions are admissible
in the sense of Definition~\ref{defadmissible}.
Indeed, inner solutions satisfy the conditions~\eqref{sigmapreserve} in Lemma~\ref{lemmagpreserve}
and can therefore be used to vary the commutator inner product (as is shown in~\cite[Proposition~3.5]{fockbosonic}).
Whether all the conditions in Definition~\ref{defadmissible} are satisfied is a rather subtle question.
But it is conceivable that at least special classes of such variations are admissible.
The resulting wave functions~${\mathscr{I}}^\Omega_{\rho, \tilde{\rho}_\tau} \tilde{\psi}^u$ in~\eqref{HFextend}
are unphysical in the sense that they do not correspond to changes of the physical system.
Instead, similar to a gauge freedom, they merely correspond to symmetry transformations
of the causal fermion system. This raises the question how to understand the resulting transformations.
The correct way to look at this freedom is to think of the transformations as unitary transformations~$U$ acting on
the whole extended Hilbert space,
\[ U \::\: \H^{\fermi, \Omega}_\rho \rightarrow \H^{\fermi, \Omega}_\rho \quad \text{unitary}\:. \]
Such a transformation changes the representation of the Hilbert space as wave functions in spacetime.
But, keeping in mind that diffeomorphisms and unitary transformations also change the form of the
Euler-Lagrange equations and the linearized field equations, the transformation as a whole does
not change the physical content of the causal fermion system. This can be seen in analogy to
Dirac theory, where a local phase transformation
\[ \psi(x) \rightarrow e^{i \Lambda(x)}\, \psi(x) \:,\qquad
i \Pdd + e \slashed{A} \rightarrow e^{i \Lambda(x)}\, (i \Pdd + e \slashed{A})\, e^{-i \Lambda(x)} \]
changes the form of the Dirac wave functions (and the subspaces of positive and negative energy are mixed),
but without an effect on any physical observables. 

Another related question is whether inner solutions can and should be included in
the jet space~$\J^\gen$ used for constructing a linear dynamics on the extended Hilbert space~$\H^\fermi_\rho$
(see Section~\ref{secdynamics}). Following up on our above conclusion that inner solutions do not change the
physical system but merely change the representation of the wave functions, inner solutions do not seem
suitable for extending the Hilbert space by wave functions having a new dynamics.
Moreover, it is not clear if and how the compatibility conditions in Definition~\ref{defgen}~(iii) could be satisfied
for inner solutions. Therefore, it seems best to choose~$\J^\gen$ as a space of jets which is disjoint from
the inner solutions. 

\Thanks{{{\em{Acknowledgments:}} We are grateful to the ``Universit\"atsstiftung Hans Vielberth'' for support.
N.K.'s research was also supported by the NSERC grant RGPIN~105490-2018.

\providecommand{\bysame}{\leavevmode\hbox to3em{\hrulefill}\thinspace}
\providecommand{\MR}{\relax\ifhmode\unskip\space\fi MR }
\providecommand{\MRhref}[2]{%
  \href{http://www.ams.org/mathscinet-getitem?mr=#1}{#2}
}
\providecommand{\href}[2]{#2}

\end{document}